%% file: RR-7659.tex
\documentclass[a4paper]{article}

\usepackage{a4wide}
\usepackage{amsfonts} 
\usepackage[cmex10]{amsmath} 
\usepackage{amssymb} 
\usepackage{amsthm} 
\usepackage{array}
\usepackage[american]{babel} 
\usepackage{cite} 
\usepackage{fancyhdr} 
\usepackage[utf8]{inputenc}
\usepackage[pdftex]{graphicx} 
\usepackage[colorlinks,frenchlinks,anchorcolor=blue,urlcolor=blue]{hyperref} 
\usepackage{longtable} 
\usepackage{microtype} 
\usepackage{relsize} 
\usepackage[tight,normalsize,sf,SF]{subfigure} 
\usepackage{svn} 
\usepackage[nofancy]{svninfo} 
\usepackage{times} 
\usepackage{xspace} 
\usepackage{paralist} 
\usepackage{figlatex} 

\graphicspath{{figures/}{logos/}}

\SVN $Date: 2011-06-24 08:54:56 +0200 (Fri, 24 Jun 2011) $
\SVN $Revision: 4021 $
\SVN $Id: RR-7659.tex 4021 2011-06-24 06:54:56Z marklee $

\newcommand{\paperkeywords}{cluster, scheduler, virtual machine, vector bin%
    packing, high performance computing, batch scheduling}

\usepackage{RR}
\RRdate{June 2011}
\RRNo{7659}
\RRauthor{Henri Casanova\and Mark Stillwell\and Fr{\'e}d{\'e}ric Vivien}
\authorhead{H. Casanova, M. Stillwell, F. Vivien}
\RRtitle{Ordonnancement dynamique et fractionnaire des ressources\protect\\%
  versus ordonnancement par \emph{batch}}
\RRetitle{Dynamic Fractional Resource Scheduling\protect\\%
    vs. Batch Scheduling}
\titlehead{DFRS vs. Batch Scheduling}

\RRresume{Nous proposons une nouvelle approche de l'ordonnancement des
  applications sur les calculateurs parallèles homogènes. Sa
  principale caractéristique est l'utilisation de machines virtuelles
  pour organiser le partage de \emph{fractions} des ressources de
  manière précise et contrôlée. Les approches existantes utilisant des
  machines virtuelles se sont principalement intéressées à des
  problèmes techniques ou à l'extension des systèmes de \emph{batch}
  existants. Notre approche est beaucoup plus aggressive et nous
  recherchons des heuristiques qui optimisent une métrique
  particulière. Nous établissons des bornes de performance absolues et
  nous développons des algorithmes pour la version en-ligne, non
  clairvoyante, de notre problème d'ordonnancement. Nous évaluons ces
  algorithmes au moyen de simulations impliquant soit des traces
  synthétiques, soit de traces d'un système HPC existant. Nous
  comparons par ce moyen nos solutions aux algorithmes
  d'ordonnancement par \emph{batch} les plus classiques. Nous montrons
  que notre approche permet d'améliorer de plusieurs ordres de
  grandeur le facteur de ralentissement (\emph{stretch}) subit par les
  applications par rapport aux systèmes de \emph{batch}, tout en ayant
  une utilisation comparable ou moindre des ressources. Nos résultats
  montrent que l'utilisation conjointe des techniques de
  virtualisation et de stratégies d'ordonnancement en-ligne permet
  d'améliorer très significativement l'exécution des applications dans
  les systèmes de calcul HPC.  }

\RRabstract{%
We propose a novel job scheduling approach for homogeneous cluster computing
platforms. Its key feature is the use of virtual machine technology to share
\emph{fractional} node resources in a precise and controlled manner. Other
VM-based scheduling approaches have focused primarily on technical
issues or on extensions to existing batch scheduling systems, while we take a more aggressive
approach and seek to find heuristics that maximize an objective metric
correlated with job performance. We derive absolute performance bounds and
develop algorithms for the online, non-clairvoyant version of our scheduling
problem. We further evaluate these algorithms in simulation against both
synthetic and real-world HPC workloads and compare our algorithms to standard
batch scheduling approaches. We find that our approach improves over batch
scheduling by orders of magnitude in terms of job stretch, while leading to
comparable or better resource utilization. Our results demonstrate that
virtualization technology coupled with lightweight online scheduling strategies
can afford dramatic improvements in performance for executing HPC workloads.}
\RRmotcle{ordonnancement, machines virtuelles, \emph{bin packing}
  vectoriel, calcul haute performance, ordonnancement par
  \emph{batch}, \emph{stretch}}

\RRkeyword{\paperkeywords, stretch}
\RRprojet{GRAAL}

\usepackage{RRthemes}
\RRthemeProj{graal}
\RCGrenoble

\hypersetup{
    pdfauthor   = {Mark Stillwell, Fr{\'e}d{\'e}ric Vivien, Henri Casanova},
    pdftitle    = {Dynamic Fractional Resource Scheduling vs. Batch Scheduling},
    pdfsubject  = {Virtual Cluster Scheduling},
    pdfkeywords = {\paperkeywords}
}

\input{remark} 
\remarktrue 

\newtheorem{theorem}{Theorem}

\newcommand{\jobset}{\mathcal{J}\xspace}
\newcommand{\procset}{\mathcal{P}\xspace}
\newcommand{\njtasks}{\ensuremath{\mathcal{T}_j}\xspace}
\newcommand{\maxstretch}{\mathcal{S}\xspace}
\newcommand{\dates}{\mathbb{D}\xspace}
\newcommand{\Int}{\ensuremath{t}\xspace}
\newcommand{\intset}{\mathcal{I}\xspace}
\newcommand{\length}{\ensuremath{\ell}\xspace}
\newcommand{\quanta}{\mathcal{Q}\xspace}
\newcommand{\R}{\mathcal{R}\xspace}

\newcommand{\demand}{\ensuremath{D}\xspace}

\newcommand{\greedy}{Greedy\xspace}
\newcommand{\greedyp}{GreedyP\xspace}
\newcommand{\greedypm}{GreedyPM\xspace}
\newcommand{\mcb}{MCB8\xspace}
\newcommand{\mcbs}{MCB8-stretch\xspace}
\newcommand{\mcbsp}{stretch-per\xspace}

\newcommand{\activeres}{*\xspace}

\newcommand{\periodic}{per\xspace}

\newcommand{\optavg}{\textsc{opt=avg}\xspace}
\newcommand{\optmin}{\textsc{opt=min}\xspace}
\newcommand{\optmax}{\textsc{opt=max}\xspace}

\makeatletter
\newcommand{\mft}[1][\@empty]{%
\ifx\@empty#1%
\textsc{minft}%
\else%
\textsc{minft=#1}%
\fi\xspace}

\newcommand{\mvt}[1][\@empty]{%
\ifx\@empty#1%
\textsc{minvt}%
\else%
\textsc{minvt=#1}%
\fi\xspace}
\makeatother

\newcommand{\equi}{\textsc{EquiPartition}\xspace}

\makeatletter

\newcommand{\eqnlabel}[1]{\xdef\@currentlabel{\theequation}\ltx@label{#1}}
\newcommand{\cnstslabel}[2][\thesubcounstraints]{
   \refstepcounter{subcounstraints}%
   \xdef\@currentlabel{#1}%
   \tagform@{#1}%
   \ltx@label{#2}\quad &}
\newcommand{\cnstlabelBIS}[2][\thesubcounstraints]{
   \xdef\@currentlabel{#1}%
   \ltx@label{#2}}
\newcounter{subcounstraints}[equation]

\newenvironment{LinearProgram}[2][Maximize]{%
   \begin{equation}%
     \left\{\begin{array}{l}%
       \textsc{#1 } #2
       \textsc{~under the constraints}\\%
       \begin{aligned}%
 }{%
        \end{aligned}
     \end{array}\right.
   \end{equation}%
}

\newenvironment{LinearProgramWOObjective}{%
   \begin{equation}%
     \left\{\begin{array}{l}%
       \begin{aligned}%
 }{%
        \end{aligned}
     \end{array}\right.
   \end{equation}%
}

\newcommand{\EquationsNumbered}[1]{%
  \let\label=\cnstlabelBIS%
   \setcounter{subcounstraints}{0}
   \renewcommand{\thesubcounstraints}{\theequation\alph{subcounstraints}}
   \def\set@counstraintscounter{
     \refstepcounter{subcounstraints}%
     \xdef\@currentlabel{\thesubcounstraints}%
     \tagform@{\thesubcounstraints}%
   }
   \def\markwithoutsep{\set@counstraintscounter\quad}
   \def\mark{\set@counstraintscounter\quad&}
   \def\n{\\\mark}
   \begin{aligned}%
     #1
   \end{aligned}%
   \let\label=\eqnlabel%
 }

\newcommand{\EquationsNumberedbis}[2]{%
  \let\label=\cnstlabelBIS%
   \setcounter{subcounstraints}{0}
   \renewcommand{\thesubcounstraints}{\theequation\alph{subcounstraints}}
   \def\set@counstraintscounter{
     \refstepcounter{subcounstraints}%
     \xdef\@currentlabel{\thesubcounstraints}%
     \tagform@{\thesubcounstraints}%
   }
   \def\markwithoutsep{\set@counstraintscounter\quad}
   \def\mark{\set@counstraintscounter\quad&}
   \def\n{\\\mark}
   \begin{aligned}%
     #1
   \end{aligned}\\%
   & \begin{aligned}
     #2
   \end{aligned}
   \let\label=\eqnlabel%
 }
\makeatother

\begin{document}

\makeRR

\section{Introduction}
\label{sec.introduction}

The standard method for sharing a cluster among High Performance Computing (HPC)
users is batch scheduling. With batch scheduling, users submit \emph{requests}
to run applications, or \emph{jobs}. Each request is placed in a queue and waits
to be granted an \emph{allocation}, that is, a subset of the cluster's compute
nodes, or \emph{nodes} for short. The job has exclusive access to these nodes
for a bounded duration.

One problem with batch scheduling is that it inherently limits overall resource
utilization. If a job uses only a fraction of a node's resource (e.g., half of
the processor cores, a third of the memory), then the remainder of it is wasted.
It turns out that this is the case for many jobs in HPC workloads. For example,
in a 2006 log of a large Linux cluster~\cite{feitelson-pwa}, more than 95\% of
the jobs use under 40\% of a node's memory, and more than 27\% of the jobs
effectively use less than 50\% of the node's CPU resource. Similar observations
have been made repeatedly in the literature~\cite{setia1999ijm, batat2000gsm,
chiang2001cls, li2004wcm}.  
Additionally, since batch schedulers use integral resource allocations with no
time-sharing of nodes, incoming jobs can be postponed even while some nodes are
sitting idle. 

A second problem is the known disconnect with user concerns (response time,
fairness)~\cite{lee2007pru, schwiegelshohn2000fpj}. While batch schedulers
provide myriad configuration parameters,
these parameters are not directly related to relevant user-centric metrics.

In this work we seek to remedy both of the above problems. We address the first
by allowing fractional resource allocations (e.g., allocating 70\% of a resource
to a job task) that can be modified on the fly (e.g., by changing allocated
fractions, by migrating job tasks to different nodes). We address the second
problem by defining an objective performance metric and developing algorithms
that attempt to optimize it.

Existing job scheduling approaches generally assume that job processing times
are known~\cite{bender2004aaa} or that reliable estimates are
available~\cite{srinivasan2002cbs}. Unfortunately, user-provided job processing
time estimates are often inaccurate~\cite{lee2006ous}, albeit used by batch
schedulers. We take a drastic departure from the literature and assume no
knowledge of job processing times.

Our approach, which we term \emph{dynamic fractional resource scheduling}
(DFRS), amounts to a carefully controlled time-sharing scheme enabled by virtual
machine (VM) technology. Other VM-based scheduling approaches have been
proposed, but the research in this area has focused primarily on technical
issues~\cite{bhatia2007vcm, hermenier2009ecm} or extensions to existing
scheduling schemes, such as combining best-effort and reservation based
jobs~\cite{sotomayor2008cbe}. In this work we:  
\begin{compactitem} 
    \item Define the offline and online DFRS problems and establish their 
          complexity; 
    \item Derive absolute performance bounds for any given problem instance;
    \item Propose algorithms for solving the online non-clairvoyant DFRS 
          problem; 
    \item Evaluate our algorithms in simulation using synthetic and real-world
          HPC workloads; 
    \item Identify algorithms that outperform batch scheduling by orders of 
          magnitude; 
    \item Define a new metric to capture the notion of efficient resource 
          utilization; 
    \item Demonstrate that our algorithms can be easily tuned so that they are  
          as or more resource efficient than batch scheduling while still 
          outperforming it by orders of magnitude.
\end{compactitem}

\noindent We formalize the DFRS problem in Section~\ref{sec.flex}, study its
complexity in Section~\ref{sec.theory},  and propose DFRS algorithms in
Section~\ref{sec.alg}.  We describe our experimental methodology in
Section~\ref{sec.experiments} and present results in Section~\ref{sec.results}.
We discuss related work in Section~\ref{sec.related}, and conclude with a
summary of results and a highlights of future directions in
Section~\ref{sec.conclusion}.

\section{The DFRS Approach}
\label{sec.flex}

DFRS implements fractional allocation of resources, such as CPU cycles, and thus
uses time-sharing. The classical time-sharing solution for parallel applications
is gang scheduling~\cite{ousterhout1982stc}. In gang scheduling, tasks in a
parallel job are executed during the same synchronized time slices across
cluster nodes. Gang scheduling requires distributed synchronized
context-switching, which may require significant overhead and thus long time
slices. Furthermore, in any time slice the use of any node is dedicated to a
single application, which leads to low system utilization for applications that
do not saturate CPU resources. Because of its drawbacks, gang scheduling is used
far less often than batch scheduling for managing HPC clusters. 

To circumvent the problems of batch scheduling without being victim of the
drawbacks of gang scheduling, in this work we opt for time-sharing in an
uncoordinated and low-overhead manner, enabled by virtual machine (VM)
technology. Beyond providing mechanisms for efficient time-sharing, VM
technology also allows for seamless job preemption and migration without any
modification of application code. Both preemption and migration can be used to
increase fairness among jobs and, importantly, to avoid starvation. Migration
can also be used to achieve better load balance, and hence better system
utilization and better overall job performance.

\subsection{System Overview and Use of VM Technology}
\label{sec.system}

We target clusters of homogeneous \emph{nodes} managed by a resource allocation
system that relies on VM technology. The system responds to job requests by
creating collections of VM instances on which to run the jobs. Each VM instance
runs on a physical node under the control of a VM Monitor that can enforce
specific resource fraction allocations for the instance. All VM Monitors are in
turn under the control of a VM Manager that specifies allocated resource
fractions for all VM instances. The VM Manager can also preempt instances, and
migrate instances among physical nodes. Several groups in academia and industry
have developed systems with this conceptual architecture~\cite{mcnett2007uef,
grit2007hvm, nurmi2008eoc, virtual_center, xen_enterprise}. Such use of VM
technology as the main resource consolidation and management mechanism is today
one of the tenets of ``cloud computing.''

VM technology allows for accurate sharing of hardware resources among VM
instances while achieving performance isolation. For instance, the popular Xen
VM Monitor~\cite{barham2003xav} enables CPU-sharing and performance isolation in
a way that is low-overhead, accurate, and rapidly
adaptable~\cite{schanzenbach2008arc}. Furthermore, sharing can be arbitrary. For
instance, the Xen Credit CPU scheduler can allow three VM instances to each
receive 33.3\% of the total CPU resource of a dual-core
machine~\cite{gupta2007ctc}. This allows a multi-core physical node to be
considered as an arbitrarily time-shared single core. Virtualization of other
resources, such as I/O resources, is more challenging~\cite{willmann2007cdn} but
is an active area of research~\cite{wiov08}.  Recent work also targets the
virtualization of full memory hierarchies (buses and
caches)~\cite{nesbit2007vpc}.  In this work we simply assume that one can start
a VM instance on a node and allocate to it reasonably precise fractions of the
resources on that node. Given this capability, whether available today or in the
future, our approach is applicable to many resource dimensions. In our
experiments we include solely CPU and memory resources, the sharing of which is
well supported by VM technology today.

An additional advantage of executing jobs within VM instances is that their
instantaneous resource needs can be discovered via
monitoring~\cite{gupta2005xqm, grit2007hvm}, introspection~\cite{jones2006atp,
jones2006gmb}, and/or configuration variation~\cite{jones2006atp, jones2006gmb}.

\subsection{Problem Statement}
\label{sec.problem}

We consider a homogeneous cluster based on a switched interconnect and with a
network-attached storage. Users submit requests to run jobs that consist of one
or more tasks to be executed in parallel. Each task runs within a VM instance.
Our goal is to make sound resource allocation decisions. These decisions include
selecting initial nodes for VM instances, setting allocated resource fractions
for each instance, migrating instances between nodes, preempting and pausing
instances (by saving them to local or network-attached storage), and postponing
incoming job requests.

Each task has a \emph{memory requirement}, expressed as a fraction of node
memory, and a \emph{CPU need}, which is the fraction of node CPU cycles that the
task needs to run at maximum speed. For instance, a task could require 40\% of
the memory of a node and would utilize 60\% of the node's CPU resource in
dedicated mode. We assume that these quantities are known and do not vary
throughout job execution. Memory requirements could be specified by users or be
discovered on-the-fly, along with CPU needs, using the discovery techniques
described in Section~\ref{sec.system}.  Memory capacities of nodes should not be
exceeded. In other words, we do not allow the allocation of a node to a set of
tasks whose cumulative memory requirement exceeds 100\%. This is to avoid the
use of process swapping, which can have a hard to predict but almost always
dramatic impact on task execution times. We do allow for overloading of CPU
resources, meaning that a node may be allocated to a set of tasks whose
cumulative CPU needs exceed 100\%. Further, the CPU fraction actually allocated
to the task can change over time, e.g., it may need to be decreased due to the
system becoming more heavily loaded. When a task is given less than its CPU need
we assume that its execution time is increased proportionally. The task then
completes once the cumulative CPU resource assigned to it up to the current time
is equal to the product of its CPU need and execution time on a dedicated
system. Note that a task can never be allocated more CPU than its need. In this
work we target HPC workloads, which mostly comprise regular parallel
applications.  Consequently, we assume that all tasks in a job have the same
memory requirements and CPU needs, and that they must progress at the same rate.
We enforce that allocations provide identical instantaneous CPU fractions to all
tasks of a job, as non-identical fractions needlessly waste resources.

One metric commonly used to evaluate batch schedules is the \emph{stretch} (or
slowdown)~\cite{bender1998fsm}. The stretch of a job is defined as its actual
turn-around time divided by its turn-around time had it been alone on the
cluster. For instance, a job that could have run in 2 hours on the dedicated
cluster but instead runs in 4 hours due to competition with other jobs
experiences a stretch of 2. In the literature a proposed way to optimize both
for average performance and for fairness is to minimize the maximum
stretch~\cite{bender1998fsm}, as opposed to simply minimizing average stretch,
the latter being prone to starvation~\cite{legrand2008mss}. Maximum stretch
minimization is known to be theoretically difficult. Even in clairvoyant
settings there does not exist any constant-ratio competitive
algorithm~\cite{legrand2008mss}, as seen in Section~\ref{sec.theory}.
Nevertheless, heuristics can lead to good results in
practice~\cite{legrand2008mss}. 

Stretch minimization, and especially maximum stretch minimization, tends to
favor short jobs, but on real clusters the jobs with shortest running times are
often those that fail at launch time. To prevent our evaluation of schedule
quality from being dominated by these faulty jobs, we adopt a variant of the
stretch called the \emph{bounded stretch}, or ``bounded slowdown'' utilizing the
terminology in~\cite{feitelson1997tpp}. In this variant, the turn-around time of
a job is replaced by a threshold value if this turn-around time is smaller than
that threshold. We set the threshold to 10 seconds, and hereafter we use the
term stretch to mean bounded stretch. 

In this work we do not assume \emph{any} knowledge about job processing times.
Batch scheduling systems require that users provide processing time estimates,
but these estimates are typically (wildly) inaccurate~\cite{lee2006ous}. Relying
on them is thus a losing proposition. Instead, we define a new metric, the
\emph{yield}, that does not use job processing time estimates. The \emph{yield}
of a task is the instantaneous fraction of the CPU resource of the node
allocated to the task divided by the task's CPU need. Since we assume that all
tasks within a job have identical CPU needs and are allocated identical CPU
fractions, they all have the same yield which is then the yield of the job. Both
the yield and the stretch capture an absolute level of job ``happiness'', and
are thus related.  In fact, the yield can be seen as the inverse of an
instantaneous stretch. We contend that yield optimization is more feasible than
stretch optimization given that job processing times are notorious for being
difficult to estimate while CPU needs can be discovered (see
Section~\ref{sec.system}).  In fact, in our experiments we make the conservative
assumption that all jobs are CPU bound (thereby not relying on CPU need
discovery).

Our goal is to develop algorithms that explicitly seek to maximize the minimum
yield. The key questions are whether this strategy will lead to good stretch
values in practice, and whether DFRS will be able to outperform standard batch
scheduling algorithms.

\section{Theoretical Analysis}
\label{sec.theory}

In this section we study the offline scenario so that we can derive a lower
bound on the optimal maximum stretch of any instance assuming a clairvoyant
scenario. We then quantify the difficulty of the online, non-clairvoyant case. 
Even in an offline scenario and even when ignoring CPU needs, memory constraints
make the problem NP-hard in the strong sense since it becomes a bin-packing
problem~\cite{garey1979cig}. Consequently, in this section we assume that all
jobs have null memory requirements, or, conversely, that memory resources are
infinite.

\subsection{The Offline Case}
\label{sec.offline}

Formally, an instance of the offline problem is defined by a set of jobs,
$\jobset$, and a set of nodes, $\procset$. Each job $j$ has a set, $\njtasks$,
of tasks and a CPU need, $c_j$, between $0$ and $1$. It is submitted at its
\emph{release date} $r_j$ and has \emph{processing time} $p_j$, representing its
execution time on an equivalent dedicated system. A target value $\maxstretch$
for the maximum stretch defines a deadline $d_j = r_j + \maxstretch\times p_j$
for the execution of each job $j$. The set of job release dates and deadlines,
$\dates = \bigcup_{j \in \jobset} \{r_j, d_j\}$, gives rise naturally to a set
$\intset$ of consecutive, non-overlapping, left-closed intervals that cover the
time span $\left[\min_{j \in \jobset}r_j,\max_{j \in \jobset}d_j\right)$, where
the upper and lower bounds of each interval in $\intset$ are members of $\dates$
and each member of $\dates$ is either the upper or lower bound of at least one
member of $\intset$. For any interval $\Int$ we define $\length(\Int) = \sup\Int
- \inf\Int$ (i.e., $\length(\Int)$ is the \emph{length} of $\Int$).

A \emph{schedule} is an allocation of processor resources to job tasks over
time. For a schedule to be valid it must satisfy the following conditions: 1)
every task of every job $j$ must receive $c_j \times p_j$ units of work over the
course of the schedule, 2) no task can begin before the release date of its job,
3) at any given moment in time, no more than 100\% of any CPU can be allocated
to running tasks, 4) the schedule can be broken down into a sequence of
consecutive, non-overlapping time spans no larger than some time quantum
$\quanta$, such that over each time span every task of a job $j$ receives the
same amount of work and each of these tasks receives no more than $c_j$ times
the length of the time span units of work. Within these small time spans any
task can fully utilize a CPU resource, regardless of the CPU need of its job,
and the tasks of a job can proceed independently. The exact size of $\quanta$
depends upon the system, but as the timescale of parallel job scheduling is
orders of magnitude larger than that of local process scheduling, we make the
reasonable assumption that for every $\Int \in \intset$, $\quanta <<
\length(\Int)$.

\begin{theorem}\label{thm-opt-stretch}

  Let us consider a system with infinite memory and assume that any task can be
  moved instantaneously and without penalty from one node to another. Then there
  exists a valid schedule whose maximum stretch is no greater than $\maxstretch$
  if and only if the following linear system has a solution, where each variable
  $\alpha^{\Int}_{j}$ represents the portion of job $j$ completed in time 
  interval $\Int$:
  \begin{LinearProgramWOObjective}
    &\EquationsNumbered{\mark
      \label{lp-opt-jobcomplete}\forall j \in \jobset \quad &
        \sum_{\Int \in \intset} \alpha_j^{\Int} = 1;\n%
      \label{lp-opt-release}\forall j \in \jobset, \forall \Int \in \intset
        \quad &
        r_j \geq \sup \Int \Rightarrow \alpha_j^{\Int} =  0;\n%
      \label{lp-opt-deadline}\forall j \in \jobset, \forall \Int \in \intset
        \quad &
        d_j \leq \inf \Int \Rightarrow \alpha_j^{\Int} =  0;\n%
      \label{lp-opt-int}\forall j \in \jobset, \forall \Int \in \intset \quad &
        \alpha_j^{\Int} p_j \leq \length(\Int);\n%
      \label{lp-opt-power}\forall \Int \in \intset \quad &
        \sum_{j \in \jobset} \alpha_j^{\Int} p_j c_j |\njtasks| \leq |\procset|
        \length(\Int).%
    }
    \label{lp-opt}
  \end{LinearProgramWOObjective}

\end{theorem}

\begin{proof}
  The constraints in Linear System~\eqref{lp-opt} are straightforward. They 
  restate a number of the conditions required for a valid schedule in terms of 
  $\alpha^{\Int}_{j}$, and also add a requirement that jobs be processed before 
  their deadlines. In greater detail:
  \begin{compactitem}
  \item Constraint~\eqref{lp-opt-jobcomplete} states that each job must be 
        fully processed;
  \item Constraint~\eqref{lp-opt-release} states that no work can be done on a 
        job before its release date;
  \item Constraint~\eqref{lp-opt-deadline} states that no work can can be done
        on a job after its deadline;
  \item Constraint~\eqref{lp-opt-int} states that a task cannot run longer 
        during a time interval than the length of the time interval;
  \item Constraint~\eqref{lp-opt-power} states that the cumulative computational
        power used by the different tasks during a time interval cannot exceed 
        what is available.
  \end{compactitem}
  These conditions are necessary. We now show that they suffice to insure the 
  existence of a schedule that achieves the desired maximum stretch (i.e., there
  exists a valid schedule in which each job completes before its deadline). 

  From any solution of Linear System~\eqref{lp-opt} we can build a valid 
  schedule. We show how to build the schedule for a single interval $\Int$, the
  whole schedule being obtained by concatenating all interval schedules. For 
  each job $j$, each of its tasks receives a cumulative computational 
  power equal to $\alpha^{\Int}_{j}p_j c_j$ during interval $\Int$. Let 
  $a^{\Int}_{j}$ and $b^{\Int}_{j}$ be two integers such that 
  $\frac{a^{\Int}_{j}}{b^{\Int}_{j}} = \alpha^{\Int}_{j}p_j c_j$. Without loss 
  of generality, we can assume that all integers $b^{\Int}_{j}$ are equal to a 
  constant $b$: 
  $\forall j \in \jobset, \forall \Int \in \intset, b^{\Int}_{j} = b$. Let 
  $\R$ be a value smaller than $\quanta$ such that there exists an integer
  $\lambda > 0$ where $\length(\Int) = \lambda \times \R$. 
  As $\R$ is smaller than $\quanta$, during any of the $\lambda$ sub-intervals
  of size $\R$ the different tasks of a job can fully saturate the CPU and be
  run in any order. During each of these sub-intervals, we greedily schedule the
  tasks on the nodes in any order: starting at time 0, we first run the first 
  task on the first node at full speed (100\% CPU utilization) for a time 
  $\frac{a^{\Int}_{j_1}}{b\lambda}$, where $j_1$ is the job this task belongs 
  to. Then we run the second task on the first node at full speed (100\% CPU 
  utilization) for a time $\frac{a^{\Int}_{j_2}}{b\lambda}$, where $j_2$ is the 
  job this task belongs to. If there is not enough remaining time on the first 
  node to accommodate the second task, we schedule the second task on the first 
  node for all the remaining time, and we schedule the remaining of the task on 
  the second node starting at the beginning of the sub-interval (thanks to our 
  assumption on task migration). We proceed in this manner until every task is
  scheduled. We now show that this schedule is valid.

  Note that our construction ensures that no task runs simultaneously on two 
  different nodes. Indeed, this could only happen if, for a task of some job 
  $j$, we had:
  \[ 
    \frac{a^{\Int}_{j}}{\lambda b} > \R \quad \Leftrightarrow \quad %
    \alpha^{\Int}_{j}p_j c_j = \frac{a^{\Int}_{j}}{b} > \length(\Int)\;,
  \]
  which is forbidden by Constraint~\eqref{lp-opt-int}. Then, by construction, 
  1) every task of every job $j$ receives $\sum_{\Int \in
  \intset}\alpha_j^{\Int}p_j c_j = p_j c_j$ units of work, 2) no task is started
  before the release date of its job, 3) at any given time no more than 100\% of
  any CPU is allocated to running tasks, and 4) the schedule can be broken down
  into a sequence of non-overlapping time spans of size no larger than
  $\quanta$, such that over each time span every task of each job $j$ receives 
  the same amount of work and no task receives more than $c_j$ times the length 
  of the time span units of work.
\end{proof}

Since all variables of Linear System~\eqref{lp-opt} are rational, one can check
in polynomial time whether this system has a solution. Using a binary search,
one can use the above theorem to find an approximation of the optimal maximum
stretch. In fact, one can find the optimal value in polynomial time using a
binary search and a version of Linear System~\eqref{lp-opt} tailored to check
the existence of a solution for a range of stretch values (see~\cite[Section
6]{legrand2008mss} for details). While the underlying assumptions in
Theorem~\ref{thm-opt-stretch} are not met in practice, the optimal maximum
stretch computed via this theorem is a lower bound on the actual optimal maximum
stretch.

\subsection{The Online Case}

Online maximum stretch minimization is known to be theoretically difficult.
Recall that the competitive ratio of an online algorithm is the worst-case ratio
of the performance of that algorithm with the optimal offline algorithm.  Even
in a clairvoyant scenario there is no constant-ratio competitive online
algorithm~\cite{legrand2008mss}. 

In this work we study an online, non-clairvoyant scenario. However, unlike the
work in~\cite{legrand2008mss}, time-sharing of compute nodes is allowed. The
question then is whether this added time-sharing capability can counter-balance
the fact that the scheduler does not know the processing time of jobs when they
arrive in the system. In general, bounds on the competitive ratios of online
algorithms can be expressed as a function of the number of jobs submitted to the
system, here denoted by $|\jobset|$, or as a function of $\Delta$, the ratio
between the processing times of the largest and shortest jobs.

In this section we assume that we have one single-core node at our disposal or,
equivalently, that all jobs are perfectly parallel. We show that, in spite of
this simplification, the problem is still very difficult (i.e., lower bounds on
competitive ratios are large). As a result, the addition of time-sharing does
not change the overall message of the work in~\cite{legrand2008mss}. The first
result is that the bound derived for online algorithms in a clairvoyant setting
without time-sharing holds in a non-clairvoyant, time-sharing context.

\begin{theorem}\label{thm:lb-online-maxstretch-delta}
  There is no $\frac{1}{2}\Delta^{\sqrt{2}-1}$-competitive preemptive
  time-sharing online algorithm for minimizing the maximum stretch if
  at least three jobs have distinct processing times.
\end{theorem}

This result is valid for both clairvoyant and non-clairvoyant scenarios and is
established by the proof of Theorem 14 in~\cite{legrand2008mss}, which holds
when time-sharing is allowed. Surprisingly, we were not able to increase this
bound by taking advantage of non-clairvoyance. However, as seen in the next
theorem, non-clairvoyance makes it possible to establish a very large bound with
respect to the number of jobs.

\begin{theorem}\label{thm:lb-online-maxstretch-n}
  There is no (preemptive) online algorithm for the non-clairvoyant
  minimization of max-stretch whose competitive ratio is strictly smaller than
  $|\jobset|$, the number of jobs.
\end{theorem}
\begin{proof}
  By contradiction, let us hypothesize that there exists an algorithm with a 
  competitive ratio strictly smaller than $|\jobset| -\varepsilon$ for some 
  $\varepsilon >0$.

  We consider an instance where jobs are all released at time 0, with 
  processing times large enough such that all jobs are kept running until time 
  $|\jobset|$ regardless of what the algorithm does. The job that has received 
  the smallest cumulative compute time up to time $|\jobset|$ has received at 
  most 1 unit of compute time (one $|\jobset|$-th of the $|\jobset|$ time 
  units). We sort the jobs in increasing order of the cumulative compute time 
  each has received up to time $|\jobset|$. For some value of $\lambda \geq
  |\jobset|$, we construct our instance so that the $i$-th job in this order has
  processing time $\lambda^{i-1}$. The completion time of job $1$, i.e., the job
  that has received the smallest cumulative compute time up to time $|\jobset|$,
  is at least $|\jobset|$. Consequently, its stretch is no smaller than 
  $|\jobset|$ because its processing time is $\lambda^0=1$.

  A possible schedule would have been to execute jobs in order of increasing 
  processing time. The stretch of the job of processing time $\lambda^{i-1}$ 
  would then be:
  \[
    \frac{\sum_{k=1}^{i} \lambda^{k-1}}{\lambda^{i-1}} =
    \frac{\lambda^i-1}{\lambda^{i-1}(\lambda-1)} 
    \xrightarrow[\lambda \to +\infty]{} 1.
  \]
  Therefore, if $\lambda$ is large enough, no job has a stretch greater 
  than $1+\frac{\varepsilon}{|\jobset|}$ in this schedule. Consequently, the 
  competitive ratio of our hypothetical algorithm is no smaller than:
  \[
    \frac{|\jobset|}{1+\frac{\varepsilon}{|\jobset|}} \geq 
    |\jobset| (1-\frac{\varepsilon}{|\jobset|}) = 
    |\jobset| - \varepsilon\;, 
  \]
  which is a contradiction.
\end{proof}

The \equi algorithm, which gives each job an equal share of the platform, is
known to deliver good performance in some non-clairvoyant
settings~\cite{edmonds1999sid}. We therefore assess its performance for maximum
stretch minimization.

\begin{theorem}
   In a non-clairvoyant scenario:\\
  \noindent 1)
  \equi is exactly a $|\jobset|$-competitive online algorithm for maximum 
        stretch minimization;\\
  \noindent 2) There exists an instance for which \equi achieves a maximum
        stretch at least $\frac{\Delta+1}{2+\ln(\Delta)}$ times the optimal.
\end{theorem}

\begin{proof} ~\\
  \noindent \textbf{Competitive ratio as a function of $|\jobset|$ --}\\
  At time $t$, \equi gives each of the $m(t)$ not-yet-completed jobs a share of 
  the node equal to $\frac{1}{m(t)} \geq \frac{1}{|\jobset|}$. Hence, no job has
  a stretch greater than $|\jobset|$ and the competitive ratio of \equi is no 
  greater than $|\jobset|$. We conclude using 
  Theorem~\ref{thm:lb-online-maxstretch-n}.

  \noindent \textbf{Competitive ratio as a function of $\Delta$ --}\\
  Let us consider a set $\jobset$ of $n = |\jobset|$ jobs as follows. Jobs 1 and
  2 are released at time 0 and have the same processing time. For 
  $i = 3, \ldots, n$, job $j_i$ is released at time 
  $r_{j_i} = r_{j_{i-1}} + p_{j_{i-1}}$. Job processing times  are 
  defined so that, using \equi, all jobs complete at time $r_{j_n}+n$. 
  Therefore, the $i$-th job is executed during the time interval 
  $[r_{j_i},r_{j_n}+n]$. There are two active jobs during the time interval 
  $[r_{j_1}=r_{j_2}=0,r_{j_3}]$, each receiving one half of the node's 
  processing time. For any $i \in [3,n-1]$, there are $i$ active jobs in 
  the time interval $[r_{j_i},r_{j_{i+1}}]$, each receiving a fraction $1/i$ of 
  the processing time. Finally, there are $n$ jobs active in the time interval 
  $[r_{j_n},r_{j_n}+n]$, each receiving a fraction $1/n$ of the node's compute 
  time.

  The goal of this construction is to have the $n$-th job experience a stretch 
  of $n$. However, by contrast with the previous theorem, the value of $\Delta$ 
  is ``small,'' leading to a large competitive ratio as a function of $\Delta$, 
  but smaller than $n$. Formally, to define the job processing times, we write 
  that the processing time of a job is equal to the cumulative compute time it 
  is given between its release date and its deadline using \equi:
  \[
    \left\{
    \begin{array}{ll}\displaystyle
      \forall i \in [1,2] &\displaystyle p_{j_i} = \frac{1}{2}(r_{j_3}-r_{j_1}) 
        + \sum_{k=3}^{n-1} \frac{1}{k}(r_{j_{k+1}}-r_{j_k}) 
        + \frac{1}{n}((r_{j_n}+n)-r_{j_n}) = 
        \frac{1}{2}p_{j_1} + \sum_{k=3}^{n-1} \frac{1}{k}p_{j_k} + 1\\
      \displaystyle
      \forall i \in [3,n] &\displaystyle p_{j_i} = 
        \sum_{k=i}^{n-1} \frac{1}{k}(r_{j_{k+1}}-r_{j_k}) 
        + \frac{1}{n}((r_n+n)-r_n) = \sum_{k=i}^{n-1} \frac{1}{k}p_{j_k} + 1
    \end{array}
    \right.
  \]
  We first infer from the above system of equations that $p_{j_n}=1$ (and the 
  $n$-th job has a stretch of $n$). Then, considering the equation for $p_{j_i}$
  for $i \in [3,n-1]$, we note that $p_{j_i} - p_{j_{i+1}} = 
  \frac{1}{i} p_{j_i}$. Therefore, $p_{j_i} = \frac{i}{i-1}p_{j_{i+1}}$ and, by 
  induction, $p_{j_i} = \frac{n-1}{i-1}$.  We also have $p_{j_2} - p_{j_3} = 
  \frac{1}{2}p_{j_1}= \frac{1}{2}p_{j_2}$. Therefore, $p_{j_2} = 2p_{j_3} = 
  n-1$.

  Now let us consider the schedule that, for $i \in [2,n]$, executes job $j_i$ 
  in the time interval $[r_{j_i},r_{j_{i+1}}=r_{j_i}+p_{j_i}]$, and that 
  executes the first job during the time interval 
  $[r_{j_n}+p_{j_n}=r_{j_n}+1,r_{j_n}+n]$. With this schedule all 
  jobs have a stretch of 1 except for the first job. The maximum stretch for 
  this schedule is thus the stretch of the first job. The makespan of this job, 
  i.e., the time between its release data and its completion, is:
  \[
    \sum_{i=1}^n p_{j_i} = 2\times p_{j_1} + \sum_{i=3}^n p_{j_i} = 2(n-1) + 
    \sum_{i=3}^n \frac{n-1}{i-1} 
    = (n-1)\left( 1+\sum_{i=2}^n \frac{1}{i-1}\right)
    = (n-1)\left( 1+\sum_{i=1}^{n-1} \frac{1}{i}\right).
  \]
  The first job being of size $n-1$, its stretch is thus: 
  $1+\sum_{i=1}^{n-1} \frac{1}{i} = 2+\sum_{i=2}^{n-1} \frac{1}{i}$. Using a 
  classical bounding technique:
  \[
    \sum_{i=2}^{n-1} \frac{1}{i} \leq 
      \sum_{i=2}^{n-1} \int_{i-1}^{i}\frac{1}{x}dx = 
      \int_{1}^{n-1}\frac{1}{x}dx = \ln(n-1).
  \]

  The competitive ratio of \equi on that instance is no smaller than the ratio 
  of the maximum stretch it achieves ($n$) and of the maximum stretch of 
  any other schedule on that instance. Therefore, the competitive ratio of \equi
  is no smaller than:
  \[
    \frac{n}{2+\sum_{i=2}^{n-1} \frac{1}{i}} \geq 
      \frac{n}{2+\ln(n-1)} = \frac{\Delta+1}{2+\ln(\Delta)}\;,
  \]
  as the smallest job ---the $n$-th one--- is of size 1, and the largest ones 
  ---the first two jobs--- are of size $n-1$.
\end{proof}

\noindent To put the performance of \equi into perspective, First Come First
Served (FCFS) is exactly $\Delta$-competitive~\cite{legrand2008mss}.

\section{DFRS Algorithms}
\label{sec.alg}

The theoretical results from the previous section indicate that non-clairvoyant
maximum stretch optimization is ``hopeless'': no algorithm can be designed with
a low worst-case competitive ratio because of the large number of jobs and/or
the large ratio between the largest and smallest jobs found in HPC workloads.
Instead, we focus on developing non-guaranteed algorithms (i.e., heuristics)
that perform well in practice, hopefully close to the offline bound given in
Section~\ref{sec.offline}.  These algorithms should not use more than a fraction
of the bandwidth available on a reasonable high-end system (for task
preemption/migration). Additionally, because of the online nature of the
problem, schedules should be computed quickly.

We propose to adapt the algorithms we designed in our study of the offline
resource allocation problem for static workloads~\cite{stillwell2010rav,
stillwell2009rav}. Due to memory constraints, it may not always be possible to
schedule all currently pending jobs simultaneously. In the offline scenario,
this leads to a failure condition. In the online scenario, however, some jobs
should be run while others are suspended or postponed. It is therefore
necessary to establish a measure of priority among jobs. In this section we
define and justify our job priority function (Section~\ref{sec.alg.priority}),
describe the greedy (Section~\ref{sec.alg.greedy}) and vector packing based
(Section~\ref{sec.alg.mcb}) task placement heuristics, explain how these
heuristics can be combined to create heuristics for the online problem
(Sections~\ref{sec.alg.when} and ~\ref{sec.alg.names}), give our basic
strategies for resource allocation once tasks are mapped to nodes
(Section~\ref{sec.alg.improve}), and finally provide an alternate algorithm
that attempts to optimize stretch directly instead of relying on the yield
(Section~\ref{sec.alg.stretch}).

\subsection{Prioritizing Jobs}
\label{sec.alg.priority}

We define a priority based on the \emph{virtual time} of a job, that is, the
total subjective execution time experienced by a job. Formally, this is the
integral of the job's yield since its release date. For example, a job that
starts and runs for 10 seconds with a yield of 1.0, that is then paused for 2
minutes, and then restarts and runs for 30 seconds with a yield 0.5 has
experienced 25 total seconds of virtual time ($10 \times 1.0 + 120 \times 0.0 +
30 \times0.5$). An intuitive choice for the priority function is the inverse of
the virtual time: the shorter the virtual time, the higher the priority. A job
that has not yet been allocated any CPU time has a zero virtual time, i.e., an
infinite priority. This ensures that no job is left waiting at its release date,
especially short jobs whose stretch would degrade the overall performance. This
rationale is similar to that found in~\cite{bansal2004nsm}.

This approach, however, has a prohibitive drawback: The priority of paused jobs
remains constant, which can induce starvation. Thus, the priority function
should also consider the \emph{flow time} of a job, i.e., the time elapsed since
its submission. This would prevent starvation by ensuring that the priority of
any paused job increases with time and tends to infinity.

Preliminary experimental results showed that using the inverse of the virtual
time as a priority function leads to good performance, while using the ratio of
flow time to virtual time leads to poor performance. We believe that this poor
performance is due to the priorities of jobs all converging to some constant
value related to the average load on the system. As a result, short-running
jobs, which suffer a greater penalty to their stretch when paused, are not
sufficiently prioritized over longer-running jobs. Consequently, we define the
priority function as: $\text{priority} = \frac {\text{flow time}}{(\text{virtual
time})^2}$. The power of two is used to increase the importance of the virtual
time with respect to the flow time, thereby giving an advantage to short-running
jobs. We break ties between equal-priority jobs by considering their order of
submission.

\subsection{Greedy Task Mapping}
\label{sec.alg.greedy}

A basic greedy algorithm, which we simply call \greedy, allocates nodes to an
incoming job $j$ without changing the mapping of tasks that may currently be
running. It first identifies the nodes that have sufficient available memory to
run at least one task of job $j$. For each of these nodes it computes its
\emph{CPU load} as the sum of the CPU needs of all the tasks currently allocated
to it. It then assigns one task of job $j$ to the node with the lowest CPU load,
thereby picking the node that can lead to the largest yield for the task. This
process is repeated until all tasks of job $j$ have been allocated to nodes, if
possible.

A clear weakness of this algorithm is its admission policy. If a short-running
job is submitted to the cluster but cannot be executed immediately due to memory
constraints, it is postponed. Since we assume no knowledge of job processing
times, there is no way to correlate how long a job is postponed with its
processing time.  A job could thus be postponed for an arbitrarily long period
of time, leading to unbounded maximum stretch. The only way to circumvent this
problem is to force the admission of all newly submitted jobs. This can be
accomplished by pausing (via preemption) and/or moving (via migration) tasks of
currently running jobs. 

We define two variants of \greedy that make use of the priority function as 
defined previously. \greedyp operates like \greedy except that it can pause some 
running jobs in favor of newly submitted jobs. To do so, \greedyp goes through 
the list of running jobs in order of increasing priority and marks them as 
candidates for pausing until the incoming job could be started if all these 
candidates were indeed paused. It then goes through the list of marked jobs in 
decreasing order of priority and determines for each whether it could instead be
left running due to sufficient available memory. Then, running jobs that are 
still marked as candidates for pausing are paused, and the new job is started. 
\greedypm extends \greedyp with the capability of moving rather than pausing
running jobs. This is done by trying to reschedule jobs selected for pausing in
order of their priority using \greedy.

\subsection{Task Mapping as Vector Packing} 
\label{sec.alg.mcb}

An alternative to the \greedy approach is to compute a global solution from
scratch and then preempt and/or migrate tasks as necessary to implement that
solution. As we have two resource dimensions (CPU and memory), our resource
allocation problem is related to a version of bin packing, known as
two-dimensional \emph{vector packing}. One important difference between our
problem and vector packing is that our jobs have fluid CPU needs. This
difference can be addressed as follows: Consider a fixed value of the yield,
$Y$, that must be achieved for all jobs. Fixing $Y$ amounts to transforming all
CPU needs into \emph{CPU requirements}: simply multiply each CPU need by $Y$.
The problem then becomes exactly vector packing and we can apply a preexisting
heuristic to solve it. We use a binary search on $Y$ to find the highest yield
for which the vector packing problem can be solved (our binary search has an
accuracy threshold of 0.01).

In previous work~\cite{stillwell2009rav} we developed an algorithm based on this
principle called \mcb. It makes use of a two-dimensional vector packing
heuristic based on that described by Leinberger et al. in
\cite{leinberger1999mcb}. The algorithm first splits the jobs into two lists,
one containing all the jobs with higher CPU requirements than memory
requirements. Each list is then sorted by non-increasing order of the largest
requirement. The algorithm described in \cite{leinberger1999mcb} uses the sum of
the two resource dimensions to sort the jobs. We found experimentally that, for
our problem and the two dimensional case, using the maximum performs marginally
better~\cite{stillwell2009rav}.

Initially the algorithm assigns the first task of the first job in one of the
lists (picked arbitrarily) to the first node. Subsequently, it searches, in the
list that goes against the current imbalance, for the first job with an
unallocated task that can fit on the node. For instance, if the node's available
memory exceeds its available CPU resource (in percentage), the algorithm searches
for a job in the list of memory-intensive jobs. The rationale is, on each node, to keep
the total requirements of both resources in balance. If no task from any job in
the preferred list fits on the node, then the algorithm considers jobs from the
other list. When no task of any job in either list can fit on the node, the
algorithm repeats this process for the next node. 

If all tasks of all jobs can be assigned to nodes in this manner then resource
allocation is successful. In the event that the \mcb algorithm cannot
find a valid allocation for all of the jobs currently in the system at any yield
value, it removes the lowest priority job from consideration and tries again.

\subsubsection*{Limiting Migration}

The preemption or migration of newly-submitted short-running jobs can lead to
poor performance. To mitigate this behavior, we introduce two parameters. If
set, the \mft parameter (respectively, the \mvt parameter), stipulates that jobs
whose flow-times (resp., virtual times) are smaller than a given bound may be
paused in order to run higher priority jobs, but, if they continue running,
their current node mapping must be maintained. Jobs whose flow-times (resp.,
virtual times) are greater than the specified bound may be moved as previously.
Migrations initiated by \greedypm are not affected by these parameters.

\subsection{When to Compute New Task Mappings}
\label{sec.alg.when}

So far, we have not stated \emph{when} our task mapping algorithms should be
invoked. The most obvious choice is to apply them each time a new job is
submitted to the system and each time some resources are freed due to a job
completion. The \mcb algorithm attempts a global optimization and, thus, can
(theoretically) ``reshuffle'' the whole mapping each time it is
invoked\footnote{In practice this does not happen because the algorithm is
deterministic and always considers the tasks and the nodes in the same order.}.
One may thus fear that applying \mcb at each job submission could lead to a
prohibitive number of preemptions and migrations. On the contrary, \greedy has
low overhead and the addition of a new job should not be overly disruptive to
currently running jobs. The counterpart is that \greedy may generate allocations
that use cluster resources inefficiently. For both of these reasons we
consider variations of the algorithms that: upon job submission, either do nothing
or apply \greedy, \greedyp, \greedypm, or \mcb; upon job completion, either do
nothing or apply \greedy or \mcb; and either apply or do not apply \mcb
periodically.

\subsection{Algorithm Naming Scheme}
\label{sec.alg.names}

\begin{table}
\caption{DFRS Scheduling Algorithms}
\label{tab.algorithms}
\centering
\begin{tabular}{llll}
\hline
Name  & \multicolumn{3}{c}{Action}\\
          & on submission & on completion & periodic\\
\hline
\greedy\activeres             & \greedy   & \greedy & none\\
\greedyp\activeres            & \greedyp  & \greedy & none\\
\greedypm\activeres           & \greedypm & \greedy & none\\
\greedy/\periodic             & \greedy   & none    & \mcb\\
\greedyp/\periodic            & \greedyp  & none    & \mcb\\
\greedypm/\periodic           & \greedypm & none    & \mcb\\
\greedy\activeres/\periodic   & \greedy   & \greedy & \mcb\\
\greedyp\activeres/\periodic  & \greedyp  & \greedy & \mcb\\
\greedypm\activeres/\periodic & \greedypm & \greedy & \mcb\\
\mcb\activeres                & \mcb      & \mcb    & none\\
\mcb/\periodic                & \mcb      & none    & \mcb\\
\mcb\activeres/\periodic      & \mcb      & \mcb    & \mcb\\
/\periodic                    & none      & none    & \mcb\\
/\mcbsp                       & none      & none    & \mcbs\\
\end{tabular}
\end{table}

We use a multi-part scheme for naming our algorithms, using `/' to separate the
parts. The first part corresponds to the policy used for scheduling jobs upon
submission, followed by a `\activeres' if jobs are also scheduled
opportunistically upon job completion (using \mcb if \mcb was used on
submission, and \greedy if \greedy, \greedyp, or \greedypm was used). If the
algorithm applies \mcb periodically, the second part contains ``\periodic''.
For example, the \greedyp\activeres/\periodic algorithm performs a \greedy
allocation with preemption upon job submission, opportunistically tries to
start currently paused jobs using \greedy whenever a job completes, and
periodically applies \mcb. Overall, we consider the 13 combinations shown in
Table~\ref{tab.algorithms} (the 14th row of the table is explained in
Section~\ref{sec.alg.stretch}). Parameters, such as \mvt or \mft, are appended
to the name of the algorithm if set (e.g., \mcb\activeres/\periodic/\mft[300]).

\subsection{Resource Allocation}
\label{sec.alg.improve}

Once tasks have been mapped to nodes one has to decide on a CPU allocation for
each job (all tasks in a job are given identical CPU allocations). All
previously described algorithms use the following procedure: First all jobs are
assigned yield values of $1/\max(1,\Lambda)$, where $\Lambda$ is the maximum CPU
load over all nodes. This maximizes the minimum yield given the mapping of tasks
to nodes. After this step there may be remaining CPU resources on some of the
nodes, which can be used for further improvement (without changing the mapping).
We use two different approaches to exploit remaining resource fractions.

\subsubsection*{Average Yield Optimization}

Once the task mapping is fixed and the minimum yield maximized, we can write a
rational linear program to find a resource allocation that maximizes the average
yield under the constraint that no job is given a yield lower than the maximized
minimum: 

\begin{LinearProgram}[Maximize]{\displaystyle\sum_{j \in \jobset} y_j}
  &\EquationsNumbered{\mark
    \label{lp-avg-minyield}  \forall j \in \jobset \quad
    \frac{1}{\max(1,\Lambda)} \leq y_j \leq 1;\n
    \label{lp-avg-power} \forall i \in \procset,  \quad 
    \sum_{j \in \jobset} \sum_{k \in \njtasks} y_j e_k^i c_j \leq 1.%
 }
  \label{lp-avg-yield}
\end{LinearProgram}

We reuse the notation of Section~\ref{sec.theory}. We use $e_k^i \in \{0,1\}$ to
indicate whether task $k$ of job $j$ is mapped on node $i$ (for any $k$, only
one of the $e_k^i$ is non-zero).  Since the mapping is fixed, the $e_k^i$'s are
constants. Finally, $y_j$ denotes the yield of job $j$. Linear
program~\eqref{lp-avg-yield} states that the yield of any job is no smaller than
the optimal minimum yield (Constraint~\eqref{lp-avg-minyield}), and that the
computational power of a node cannot be exceeded
(Constraint~\eqref{lp-avg-power}). Algorithms that use this optimization have
``\optavg'' as an additional part of their names.

\subsubsection*{Max-min Yield Optimization}

As an alternative to maximizing the average yield, we consider the iterative
maximization of the minimum yield. At each step the minimum yield is maximized
using the procedure described at the beginning of Section~\ref{sec.alg.improve}.
Those jobs whose yield cannot be further improved are removed from
consideration, and the minimum is further improved for the remaining jobs. This
process continues until no more jobs can be improved. This type of max-min
optimization is commonly used to allocate bandwidth to competing network
flows~\cite[Chapter 6]{bertsekas1992dn}. Algorithms that use this optimization
have ``\optmin'' as an additional part of their names.

\subsection{Optimizing the Stretch Directly}
\label{sec.alg.stretch}

All algorithms described thus far optimize the minimum yield as a way to
optimize the maximum stretch. We also consider a variant of /\periodic, called
/\mcbsp, that periodically tries to minimize the maximum stretch directly,
still assuming no knowledge of job processing times. The algorithm it uses,
called \mcbs, can only be applied periodically as it is based on knowledge of
the time between scheduling events. It follows the same general procedure as
\mcb but with the following differences. At scheduling event $i$, the best
estimate of the stretch of job $j$ is the ratio of its flow time (time since
submission), $ft_j(i)$, to its virtual time, $vt_j(i)$: $\hat{S}_j(i) = ft_j(i)
/ vt_j(i)$. Assuming that the job does not complete before scheduling event
$i+1$, then $\hat{S}_j(i+1) = ft_j(i+1) / vt_j(i+1) = (ft_j(i) + T) / (vt_j(i)
+ y_j(i)\times T)$, where $T$ is the scheduling period and $y_j(i)$ is the
yield that /\mcbsp assigns to job $j$ between scheduling events $i$ and $i+1$.
Similar to the binary search on the yield, here we do a binary search to
minimize $\hat{S}(i+1)=\max_j \hat{S}_j(i+1)$. At each iteration of the search,
a target value $\hat{S}(i+1)$ is tested. From $\hat{S}(i+1)$ the algorithm
computes the yield for each job $j$ by solving the above equation for $y_j(i)$
(if, for any job, $y_j(i) > 1$, then $\hat{S}(i+1)$ is infeasible and the
iteration fails). At that point, CPU requirements are defined and \mcb can be
applied to try to produce a resource allocation for the target value. This is
repeated until the lowest feasible value is found. Since the stretch is an
unbounded positive value, the algorithm actually performs a binary search over
the inverse of the stretch, which is between $0$ and $1$. If \mcbs cannot find
a valid allocation for any value, then the job with lowest priority is removed
from consideration and the search is re-initiated. Once a mapping of jobs to
nodes has been computed each task is initially assigned a CPU allocation equal
to the amount of resources it needs to reach the desired stretch. For the
resource allocation improvement phase we use algorithms similar to those
described in Section~\ref{sec.alg.improve}, except that the first (\optavg)
minimizes the average stretch and the second (\optmax) iteratively
minimizes the maximum stretch.

\section{Simulation Methodology}
\label{sec.experiments}

\subsection{Discrete-Event Simulator}

We have developed a discrete event simulator that implements our scheduling
algorithms and takes as input a number of nodes and a list of jobs. Each job is
described by a submit time, a number of tasks, one CPU need and one
memory requirement specification (since all tasks within a job have the same
needs and requirements), and a processing time. Jobs are allocated shares of
memory and CPU resources on simulated compute nodes. As stated previously in
Section~\ref{sec.system}, the use of VM technology allows the CPU resources of
a (likely multi-core) node to be shared precisely and fluidly as a single
resource~\cite{gupta2007ctc}. Thus, for each simulated node, the total amount
of allocated CPU resource is simply constrained to be less than or equal to
$100$\%. However, when simulating a multi-core node, $100$\% CPU resource
utilization can be reached by a single task only if that task is CPU-bound and
is implemented using multiple threads (or processes). A CPU-bound sequential
task can use at most $100/n$\% of the node's CPU resource, where $n$ is the
number of processor cores.

The question of properly accounting for preemption and migration overheads is a
complicated one. For this reason in each simulation experiment we assume this
overhead to be 5 minutes of wall clock time, whatever the job's characteristics
and the number of its tasks being migrated, which is justifiably
high~\footnote{Consider a 128-task job with 1 TB total memory, or 8 GB per task
(our simulations are for a 128-node cluster). Current technology affords
aggregate bandwidth to storage area networks up to tens of GB/sec for reading
and writing~\cite{cochrane2009shi}. Conservatively assuming 10 GB/sec, moving
this job between node memory and secondary storage can be done in under two
minutes.}. We call this overhead the \emph{rescheduling penalty}. In the real
world there are facilities that can allow for the live migration of a running
task between nodes~\cite{clark2005lmv}, but to avoid introducing additional
complexity we make the pessimistic assumption that all migrations are carried
out through a pause/resume mechanism. 

Note that none of the scheduling algorithms are aware of this penalty. Based on
preliminary results, we opt for a default period equal to twice the
rescheduling penalty for all periodic algorithm, i.e., 10 minutes.  We study
the impact of the duration of the period in Section~\ref{sec.period}. The \mft
and \mvt parameters for \mcb and \mcbs are evaluated using time bounds equal to
one and two penalties (i.e., 300s and 600s).

\subsection{Batch Scheduling Algorithms}

We consider two batch scheduling algorithms: FCFS and EASY. FCFS, often used as
a baseline comparator, holds incoming jobs in a queue and assigns them to nodes
in order as nodes become available. EASY~\cite{lifka1995ais}, which is
representative of production batch schedulers, is similar to FCFS but enables
backfilling to reduce resource fragmentation. EASY gives the first job in the
queue a reservation for the earliest possible time it would be able to run with
FCFS, but other jobs in the queue are scheduled opportunistically as nodes
become available, as long as they do not interfere with the reservation for the
first job. EASY thus improves on FCFS by allowing small jobs to run while large
jobs are waiting for a sufficiently large number of nodes. A drawback of EASY
is that it requires estimations of job processing times. In all simulations we
conservatively assume that EASY has perfect knowledge of job processing times.
While this seems a best-case scenario for EASY, studies have shown that for
some workloads some batch scheduling algorithms can produce better schedules
when using non-perfectly accurate processing times (see the discussion
in~\cite{lee2004aur} for more details). In those studies the potential
advantage of using inaccurate estimates is shown to be relatively small, while
our results show that our approach outperforms EASY by orders of magnitude. Our
conclusions thus still hold when EASY uses non-accurate processing time
estimates.

\subsection{Workloads}

\subsubsection{Real-World Workload}

We perform experiments using a real-world workload from a well-established
online repository~\cite{feitelson-pwa}. Most publicly available logs provide
standard information such as job arrival times, start time, completion time,
requested duration, size in number of nodes, etc. For our purpose, we need to
quantify the fraction of the resource allocated to jobs that are effectively
used. We selected the ``cleaned'' version of the HPC2N workload from
\cite{feitelson-pwa}, which is a 182-week trace from a 120-node dual-core
cluster running Linux that has been scrubbed to remove workload flurries and
other anomalous data that could skew the performance comparisons on different
scheduling algorithms~\cite{feitelson2006wsp}. A primary reason for choosing
this workload was that it contains almost complete information regarding memory
requirements, while other publicly available workloads contain no or little such
information.

The HPC2N workload required some amount of processing for use in our simulation
experiments. First, job per-processor memory requirements were set as the
maximum of either requested or used memory as a fraction of the system memory of
2GB, to a minimum of 10\%. Of the 202,876 jobs in the trace, only 2,142 ($\sim
1$\%) did not give values for either used or requested memory and so were
assigned the lower bound of 10\%. Second, the \textsf{swf} file
format~\cite{feitelson-pwa} contains information about the required number of
``processors,'' but not the number of tasks, and so this value had to be
inferred. For jobs that required an even number of processors and had a
per-processor memory requirement less than 50\% of the available node memory, we
assumed that the job used a number of multi-threaded tasks equal to half the
number of processors. In this case, we assume that each task has a CPU need of
100\% (saturating the two cores of a dual-core node) and the memory requirement
was doubled from its initial per-processor value. For jobs requiring an odd
number of processors or more than 50\% of the available node memory per
processor, we assumed that the number of tasks was equal to the number of
processors and that each of these tasks had a CPU need of 50\% (saturating one
core of a dual-core mode). Since we assume CPU-bound tasks, performance
degradation due to CPU resource sharing is maximal. Consequently, our
assumptions are detrimental to our approach and should benefit batch scheduling
algorithms. We split the HPC2N workload into week-long segments, resulting in
182 different experimental scenarios.

\subsubsection{Synthetic Workloads}

We also use synthetic workloads based on the model by Lublin and
Feitelson~\cite{lublin2003wps}, augmented with additional information as
described hereafter. There are several reasons for preferring synthetic
workloads to real workloads for this type of relative performance evaluation:
Real workloads often do not contain all of the information that we require.
Further, real traces may be misleading, or lacking in critical information,
such as system down times that might affect jobs running on the
system~\cite{lublin2003wps}. Also, a real workload trace only provides a single
data point, and may not be generally representative~\cite{feitelson1998mbp}.
That is, the idiosyncrasies of a trace from one site may make it inappropriate
for evaluating or predicting performance at another site~\cite{lo1998csr}.
Further, a real workload trace is the product of an existing system that uses a
particular scheduling policy, and so may be biased or affected by that policy,
while synthetic traces can provide a more neutral environment~\cite{lo1998csr}.
Finally, real workloads may contain spurious or anomalous events like user
flurries that can confound evaluations of the performance of scheduling
algorithms~\cite{frachtenberg2005ppj, feitelson2006wsp}. In fact, long workload
traces often contain such events, and including them can seriously impact
relative performance evaluation~\cite{tsafrir2006ipj}.

For the synthetic workloads we arbitrarily assume quad-core nodes, meaning that
a sequential task would use at most $25$\% of a node's CPU resource. Due to the
lack of real-world data and models regarding the CPU needs of HPC jobs, we make
pessimistic assumptions similar to those described in the previous section. We
assume that the task in a one-task job is sequential, but that all other tasks
are multi-threaded. We assume that all tasks are CPU-bound: CPU needs of
sequential tasks are $25$\% and those of other tasks are $100$\%. 

The general consensus is that there is ample memory available for allocating
multiple HPC job tasks on the same node~\cite{batat2000gsm,
chiang2001cls,li2004wcm}, but no explicit model is available in the literature.
We opt for a simple model suggested by data in Setia et
al.~\cite{setia1999ijm}: 55\% of the jobs have tasks with a memory requirement
of $10$\%. The remaining 45\% of the jobs have tasks with memory requirements
$10\times x$\%, where $x$ is an integer value uniformly distributed over
\{2,\ldots,10\}.

We generated 100 distinct traces of 1,000 jobs assuming a 128-node cluster.  The
time between the submission of the first job and the submission of the last job
is on the order of 4-6 days. To study the effect of the load on our algorithms,
we multiplied the job inter-arrival times in each generated trace by 9 computed
constants, obtaining 9 new traces with identical job mixes but offered
load~\cite{batat2000gsm}, or \emph{load}, levels of $0.1$ to $0.9$ in increments
of $.1$. We thus have 900 \emph{scaled} traces.

\section{Experimental Results}
\label{sec.results}

\subsection{Stretch Results}
\label{sec.stretchresults}

\begin{table*}[Htb]
\caption{%
    Degradation from bound results for all three trace sets, with 
    a 5-minute rescheduling penalty
    \label{tab.deg-from-bound-300-delay}}
\centering
\resizebox{\textwidth}{!}{
\begin{tabular}{|l||r|r|r||r|r|r||r|r|r|} 
\hline
            & \multicolumn{3}{c||}{Real-world trace}
            & \multicolumn{3}{c||}{Unscaled synthetic traces} 
            & \multicolumn{3}{c|}{Scaled synthetic traces}\\
\cline{2-10}
Algorithms  & \multicolumn{3}{c||}{Degradation from bound} 
            & \multicolumn{3}{c||}{Degradation from bound}
            & \multicolumn{3}{c|}{Degradation from bound}\\

            & avg. & std. & max & avg. & std. & max & avg. & std. & max\\
\hline
\hline
FCFS&3,578.5&3,727.8&21,718.4&5,457.2&2,958.5&15,102.7&5,869.3&2,789.1&17,403.3\\
EASY&3,041.9&3,438.0&21,317.4&4,955.4&2,730.6&14,036.8&5,262.0&2,588.9&14,534.1\\
\hline
\hline
\greedy\activeres/\optmin&949.8&1,828.5&11,778.4&2,435.0&2,285.6&11,229.9&3,204.3&2,517.5&19,129.2\\
\greedyp\activeres/\optmin&13.5&68.0&819.2&37.5&156.0&1,204.9&115.7&644.0&10,354.2\\
\greedypm\activeres/\optmin&13.8&68.2&819.2&33.8&154.0&1,321.7&124.0&673.5&9,598.8\\
\hline
\greedy/\periodic/\optmin&28.3&24.9&163.7&30.1&10.2&58.1&29.3&14.3&153.2\\
\greedyp/\periodic/\optmin&18.5&18.6&152.4&20.1&7.3&38.1&17.8&9.6&84.6\\
\greedypm/\periodic/\optmin&18.4&18.8&158.7&20.2&7.3&38.1&17.9&9.8&93.0\\
\hline
\greedy\activeres/\periodic/\optmin&24.3&15.9&81.6&30.4&9.7&65.7&29.1&12.3&101.4\\
\greedyp\activeres/\periodic/\optmin&17.9&19.6&213.5&20.3&6.8&32.0&17.9&8.6&89.9\\
\greedypm\activeres/\periodic/\optmin&17.9&19.3&198.6&20.3&6.9&32.0&17.9&8.6&89.9\\
\hline
\greedyp/\periodic/\optmin/\mvt[600]&8.9&18.9&152.4&5.9&4.5&38.1&7.3&8.5&96.8\\
\greedypm/\periodic/\optmin/\mvt[600]&8.8&18.9&158.7&5.9&4.5&38.1&7.3&8.1&96.8\\
\greedyp\activeres/\periodic/\optmin/\mvt[600]&6.9&14.2&149.3&4.9&2.9&19.2&6.1&6.3&103.5\\
\greedypm\activeres/\periodic/\optmin/\mvt[600]&6.9&14.4&149.6&4.8&2.4&13.6&6.1&5.4&90.2\\
\hline
\hline
\mcb\activeres/\optmin/\mvt[600]&12.0&32.8&370.0&6.9&5.4&44.4&13.2&21.6&270.9\\
\mcb/\periodic/\optmin/\mvt[600]&10.8&25.3&287.6&8.1&6.6&53.3&11.0&12.6&127.5\\
\mcb\activeres/\periodic/\optmin/\mvt[600]&13.6&30.2&318.9&7.8&3.9&21.9&12.2&15.3&195.7\\
\hline
\hline
/\periodic/\optmin/\mvt[600]&105.0&445.6&5,011.9&43.0&19.7&134.7&40.4&25.0&238.3\\
\hline
\hline
/\mcbsp/\optmax/\mvt[600]&105.0&445.6&5,011.9&43.0&19.6&134.7&40.2&24.8&236.9\\
\hline
\end{tabular}
}
\end{table*}

For a given problem instance, and for each algorithm, we define the
\emph{degradation from bound} as the ratio between the maximum stretch achieved
by the algorithm on the instance and the theoretical lower bound on the optimal
maximum stretch obtained in Section~\ref{sec.offline}. A lower value of this
ratio thus denotes better performance.  We report average and standard deviation
values computed over sets of problem instances, i.e., for each of our set of
traces. We also report maximum values, i.e., result for the ``worst trace'' for
each algorithm.

Results are shown in Table~\ref{tab.deg-from-bound-300-delay} for the FCFS and
EASY batch scheduling algorithms and for 18 of our proposed algorithms. Recall
that Table~\ref{tab.algorithms} lists 14 general combinations of mechanisms for
mapping tasks to processors. All these combinations can use either \optavg or
\optmin to compute resource allocations once a mapping has been determined.
Furthermore, the last 11 combinations in Table~\ref{tab.algorithms} use the \mcb
algorithm and thus can use \mft[300], \mft[600], \mvt[300], \mvt[600], or no
mechanism to limit task remapping. Therefore, the total number of potential
algorithms is $3 \times 2 + 11 \times 2 \times 5 = 116$. However, the full
results (see Appendix~\ref{sec.extratabs}) show that on average \optmin is never
worse and often slightly better than \optavg. Consequently, we present results
only for algorithms that use \optmin. Furthermore, we found that among the
mechanisms for limiting task remapping, \mvt is always better than \mft, and
slightly better with the larger 600s bound. Accordingly we also exclude
algorithms that use \mft, and algorithms that use \mvt with a 300s bound.
Table~\ref{tab.deg-from-bound-300-delay} presents results for all 9 greedy
combinations with no mechanism for limiting remapping, and results for 4
selected greedy combinations with \mvt[600], as explained hereafter. It also
presents results for the 3 \mcb combinations, for the /\periodic algorithm,  and
for the /\mcbsp algorithm. All these are with \mvt[600]. We leave results
without \mvt[600] for the 3 \mcb combinations out of the table because these
algorithms perform very poorly.  Indeed, as they apply \mcb upon each job
arrival, they lead to inordinate amounts of remapping and thus are more than one
order of magnitude further away from the bound than when \mvt[600] is used.  For
/\periodic and /\mcbsp, the addition of \mvt[600] does not matter since the
scheduling period is no shorter than 600s. We are left with the 18 algorithms in
the table, which we discuss hereafter.

The results in Table~\ref{tab.deg-from-bound-300-delay} are mostly consistent
across our three sets of traces, with a few exceptions.  Expectedly, EASY
outperforms FCFS. The key observation is that EASY and FCFS are outperformed by
our proposed algorithms by several orders of magnitude, thereby showing that
DFRS is an attractive alternative to batch scheduling. 

The table shows results for 4 groups of greedy algorithms. In the first group
are algorithms that do not apply \mcb periodically (i.e., those without
``\periodic" in their names). On average, these algorithms lead to results
poorer than the 6 algorithms in the next 2 groups, which do apply \mcb
periodically, on the synthetic traces. For real-world traces, we see that
\greedyp\activeres and \greedypm\activeres lead to the best average case results
for algorithms from these 3 groups. For all traces, however, the purely greedy
algorithms lead to standard deviation and maximum values that are orders of
magnitude larger than those obtained by the greedy algorithms that also make use
of \mcb periodically. We conclude that applying \mcb periodically is beneficial
for greedy algorithms. The results also show that \greedyp is better than
\greedy, demonstrating the benefits of preemption. However, the use of migration
by \greedypm does not lead to significant further improvement and can even lead
to a slight decrease in performance. It turns out that the jobs migrated in the
\greedypm approach are often low priority and thus have a high probability of
being preempted at an upcoming scheduling event anyway. Finally, by comparing
the results for the 2nd and 3rd groups of greedy algorithms, we see that
scheduling jobs opportunistically (i.e., as done by algorithms with \activeres
in their names) also seems to have limited, but generally positive, impact on
performance when combined with applying \mcb periodically.

The 4 algorithms from the last group of greedy algorithms use either \greedyp or
\greedypm, with or without opportunistic scheduling, but with the \mvt[600]
feature to limit task remapping by \mcb. We see that these algorithms outperform
all previously discussed algorithms on all 3 sets of traces. For these
algorithms the use of opportunistic scheduling is always beneficial. In this
case, \greedyp and \greedypm lead to similar results. These algorithms are the
best overall, including in terms of standard deviation.

The 3 algorithms in the next group all use \mcb to assign tasks to processors
upon job submission (and possibly completion) rather than a greedy approach.
They all use \mvt[600] because without limiting task remapping they all lead to
poorer performance by orders of magnitude due to job thrashing. Overall, while
these algorithms perform very well, they are not among the best.

The next algorithm, /\periodic/\optmin/\mvt[600], simply applies \mcb
periodically without taking action upon job arrival or job completion. It is
outperformed by the best greedy algorithm more than 6-fold. This result confirms
the notion that a scheduling algorithm should react to job submissions. The
algorithm in the last row of the table, /\mcbsp/\optmax/\mvt[600], optimizes the
stretch directly. It performs on par with /\periodic/\optmin/\mvt[600], but
more than one order of magnitude worse than our best yield-based algorithm. This
demonstrates that yield optimization is a good approach and that an algorithm
for minimizing the stretch directly may not be achievable.

Our overall conclusion from the results in
Table~\ref{tab.deg-from-bound-300-delay} is that, to achieve good performance,
all our techniques should be combined: an aggressive greedy job admission policy
with preemption of running jobs, a periodic use of the \mcb vector-packing
algorithm, an opportunistic use of resources freed upon job completion, and a
grace period that prevents remapping of tasks that have just begun executing.
Note that while the algorithms are executed in an online, non-clairvoyant
context, the computation of the bound on the optimal performance relies on
knowledge of both the release dates and processing times of all jobs.
Furthermore, the bound ignores memory constraints. Nevertheless, in our
experiments, our best algorithms are on average at most a factor 7 away from
this loose bound. We conclude that our algorithms lead to good performance in an
absolute sense.

\begin{figure}[htb]
\centering
  \includegraphics[width=0.45\textwidth]{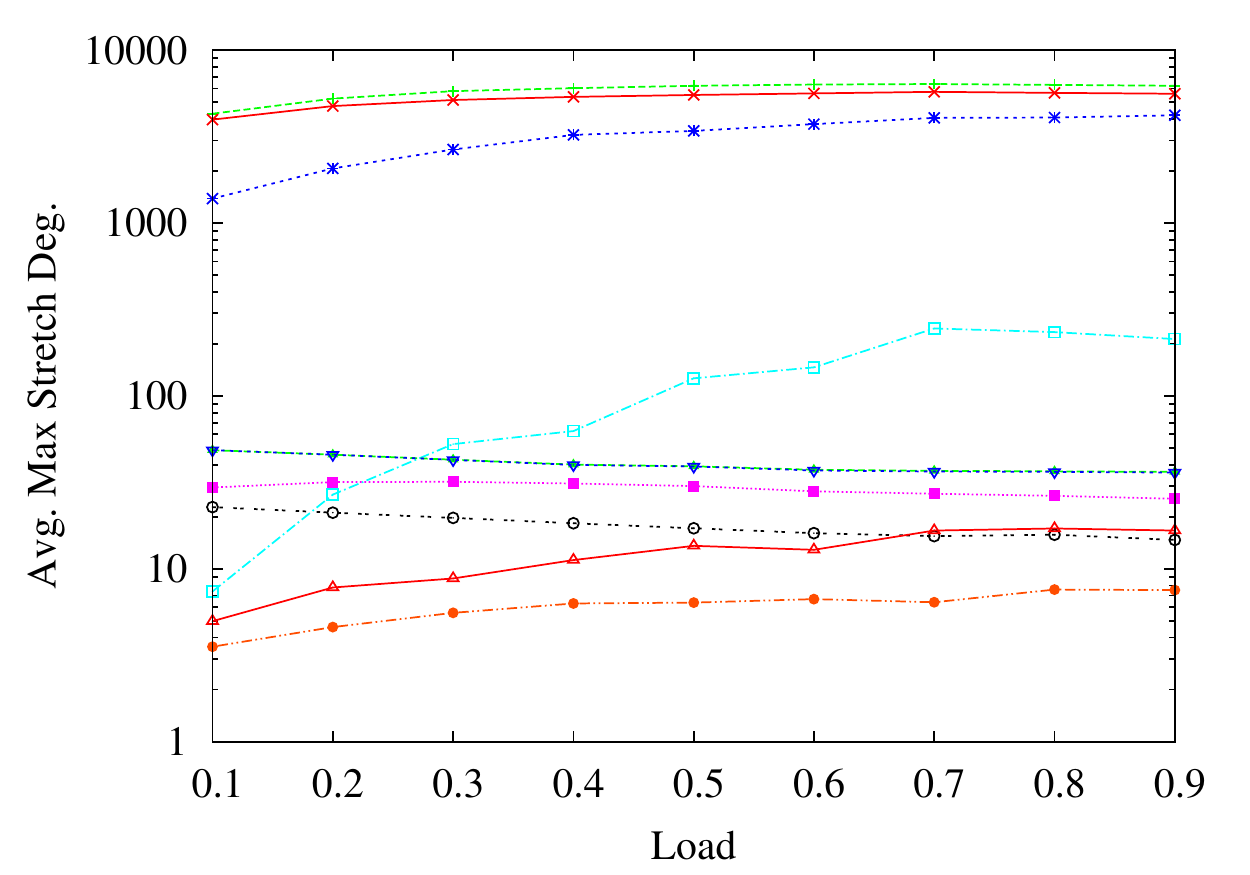}
~\\
  \includegraphics[width=0.4\columnwidth]{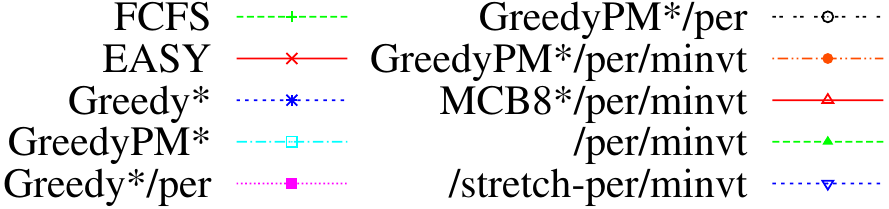}

  \label{fig.maxsdfbound-vs-load-for-300-delay}
\caption{Average degradation from bound vs. load for selected algorithms on
         the scaled synthetic dataset, assuming a 5-minute rescheduling penalty.
         Each data point shows average values over 100 instances. All algorithms
         use \optmin (except /\mcbsp, which uses \optmax) and use \mvt[600] if 
         \mvt is specified.\label{fig.maxdeg}}
\end{figure}

Figure~\ref{fig.maxdeg} plots average degradation factors vs. load for selected
algorithms when applied to scaled synthetic workloads and using the 5-minute
rescheduling penalty, using a logarithmic scale on the vertical axis. We plot
only the \greedy and \greedypm versions of the algorithms, and omit the \optmin
and the ``=600'' parts of the names so as not to clutter the caption. Most of
the observations made for the results in 
Table~\ref{tab.deg-from-bound-300-delay} hold for the results in the figure. One
notable feature here is that the \greedypm\activeres/\optmin algorithm performs
well under low load conditions but quickly leads to poorer performance than most
other algorithms as the load becomes larger than 0.3. Under low-load conditions,
the periodic use of \mcb to remap tasks is not as critical to ensure good
performance. The main observation in this figure is that our best algorithm,
\greedypm\activeres/\periodic/\optmin/\mvt[600], is the best algorithm across
the board regardless of the load level. It is about a factor 3 away from the
bound at low loads, and a factor 7.5 away from it at high loads. 

\subsection{Algorithm Execution Time}

To show that our approach produces resource allocations quickly we record the
running time of \mcb, the most computationally expensive component of our
algorithms by orders of magnitude.  We ran the simulator using the
\mcb\activeres algorithm, which applies \mcb the most frequently, on a system
with a 3.2GHz Intel Xeon CPU and 4GB RAM for the 100 unscaled traces generated
by the Lublin model. This experiment resulted in a total of 197,808
observations. For 67.25\% of these observations \mcb produced allocations for 10
or fewer jobs in less than 0.001 seconds. The remaining observations were for 11
to 102 jobs. The average compute time was about 0.25 seconds with the maximum
under 4.5 seconds. Since typical job inter-arrival times are orders of magnitude
larger~\cite{lublin2003wps}, we conclude that our best algorithms can be used in
practice.

\subsection{Impact of Preemptions and Migrations}
\label{sec.bandwidthresults}

\begin{table*}[Htb]
\centering
\caption{%
     Preemption and migration costs in terms of average bandwidth consumption,
     number of preemption and migration occurrences per hour, and number of
     preemption and migration occurrences per job. Average values over scaled
     synthetic traces with load $\geq 0.7$, with maximum values in 
     parentheses.\label{tab.costs}}
\resizebox{\textwidth}{!} {
\begin{tabular}{|l||r|r|r|r|r|r|}

\hline
          & \multicolumn{2}{c|}{Bandwidth Consumption} &
            \multicolumn{2}{c|}{Frequency of Occurrence} &
            \multicolumn{2}{c|}{\# occurrences per job}\\
Algorithm & \multicolumn{2}{c|}{(GB / sec)} &
            \multicolumn{2}{c|}{(\# occurrences / hour)} &
            \multicolumn{2}{c|}{ }\\
\cline{2-7}
          & pmtn & mig & pmtn & mig & pmtn & mig\\
\hline

EASY&0.00 (0.00)&0.00 (0.00)&0.00 (0.00)&0.00 (0.00)&0.00 (0.00)&0.00 (0.00)\\
FCFS&0.00 (0.00)&0.00 (0.00)&0.00 (0.00)&0.00 (0.00)&0.00 (0.00)&0.00 (0.00)\\
\hline
\greedy\activeres/\optmin&0.00 (0.00)&0.00 (0.00)&0.00 (0.00)&0.00 (0.00)&0.00 (0.00)&0.00 (0.00)\\
\greedyp\activeres/\optmin&0.06 (0.17)&0.00 (0.00)&5.67 (18.00)&0.00 (0.00)&0.57 (2.04)&0.00 (0.00)\\
\greedypm\activeres/\optmin&0.03 (0.07)&0.02 (0.05)&2.25 (10.08)&3.69 (13.32)&0.23 (1.19)&0.36 (1.22)\\
\hline
\greedy/\periodic/\optmin&0.48 (1.08)&0.21 (0.60)&32.58 (83.52)&38.79 (110.52)&5.41 (21.76)&4.81 (16.17)\\
\greedyp/\periodic/\optmin&0.50 (1.11)&0.20 (0.60)&33.34 (85.32)&37.52 (107.64)&5.54 (22.20)&4.67 (15.76)\\
\greedypm/\periodic/\optmin&0.49 (1.10)&0.21 (0.60)&32.93 (84.24)&39.21 (112.32)&5.52 (21.65)&4.85 (16.03)\\
\hline
\greedy\activeres/\periodic/\optmin&0.50 (1.29)&0.27 (0.66)&29.27 (84.96)&58.06 (124.56)&4.49 (22.55)&6.94 (17.65)\\
\greedyp\activeres/\periodic/\optmin&0.58 (1.37)&0.28 (0.65)&39.67 (97.92)&58.56 (127.08)&5.87 (24.87)&7.09 (17.91)\\
\greedypm\activeres/\periodic/\optmin&0.56 (1.37)&0.29 (0.66)&35.78 (96.48)&62.95 (135.72)&5.42 (24.00)&7.59 (18.32)\\
\hline
\greedy\activeres/\periodic/\optmin/\mvt[600]&0.49 (1.27)&0.24 (0.62)&28.08 (83.52)&51.97 (117.36)&4.29 (21.80)&6.25 (16.53)\\
\greedyp\activeres/\periodic/\optmin/\mvt[600]&0.56 (1.36)&0.24 (0.63)&37.70 (97.20)&52.14 (119.16)&5.56 (23.88)&6.34 (16.48)\\
\greedypm\activeres/\periodic/\optmin/\mvt[600]&0.54 (1.34)&0.26 (0.62)&33.80 (94.32)&56.45 (127.08)&5.11 (23.23)&6.84 (16.94)\\
\hline
\mcb\activeres/\optmin&0.42 (1.26)&1.17 (2.36)&61.66 (230.40)&490.48 (1,005.48)&6.67 (18.83)&57.46 (90.29)\\
\mcb\activeres/\periodic/\optmin&0.72 (1.15)&1.21 (2.68)&77.04 (170.28)&479.31 (1,109.52)&13.01 (28.93)&73.61 (117.15)\\
\mcb\activeres/\periodic/\optmin/\mvt[600]&0.54 (1.11)&0.56 (1.53)&37.94 (85.32)&194.57 (836.28)&6.57 (26.52)&22.55 (50.29)\\
\hline
/\periodic/\optmin&0.49 (1.07)&0.21 (0.62)&33.83 (84.24)&38.69 (111.24)&5.65 (23.23)&4.90 (16.58)\\
\hline
\mcbsp/\optmax&0.28 (0.65)&0.39 (0.81)&20.41 (45.36)&67.26 (159.48)&3.79 (16.71)&10.41 (26.96)\\
\hline
\end{tabular}
}
\end{table*}

The previous section shows that our algorithms wildly outperform EASY. One may,
however, wonder whether network and I/O resources are not overly taxed by
preemption and/or migration activity. Table~\ref{tab.costs} shows results
obtained using our synthetic scaled trace data, and only for traces with load
levels 0.7 or higher. The table shows averages and, in parentheses, maxima,
computed over all traces for our two best algorithms from the previous section.
For each algorithm, the first two result columns of the table show the overall
bandwidth consumption due to preemptions and migrations in GB/sec. The next two
columns show the frequency of job preemptions and migrations, and the last two
columns show the number of preemptions and migrations per job. These numbers are
computed by saying that a job is preempted (resp. migrated) whenever one or more
of its tasks is preempted (resp. migrated).

The main observation from the results in Table~\ref{tab.costs} is that the total
bandwidth consumption is reasonable. The two algorithms have a total average
bandwidth consumption under 0.80 GB/sec. Even accounting for maximum bandwidth
consumption, i.e., for the trace that causes the most traffic due to preemptions
and migrations, bandwidth consumptions are under 2.0 GB/sec.  Such numbers
represent only a small fraction of the bandwidth capacity of current
interconnect technology for cluster platforms. The numbers of preemptions per
hour show that under 40 preemptions and under 60 migrations occur each hour,
with each job being preempted under 6 times and migrating under 7 times during
its lifetime, on average.  Again, in light of the good performance achieved by
these algorithms, these numbers are reasonable in spite of our conservative
5-minute rescheduling penalty.  

Results are similar for other algorithms that use the greedy approach and apply
\mcb periodically. The algorithms that use the greedy approach and do not employ
\mcb periodically expectedly lead to lower numbers of migrations and preemptions
per hour (under 10 for each, on average). The algorithms that use the \mcb
algorithm at each job arrival, and possibly at each job completion, lead to
higher number of preemptions per hour and, more noticeably, much higher numbers
of migrations per hour (above 120 and up to 490). These algorithms suffer from
job thrashing, which has been identified in the previous section as the main
cause for their poor performance.  This phenomenon is also seen in per-job
numbers. While all other algorithms migrate a job less than 8 times on average,
these algorithms migrate a job 14 times on average and up to 73 times in the
worst case.

Overall, the bandwidth consumption due to preemption/migration and the number of
preemption/migration events reported in Table~\ref{tab.costs} show that our
algorithms can be put in practice for real-world platforms. Complete tables
detailing the preemption and migration costs for all algorithms are provided in
Appendix~\ref{sec.extratabs}.

\subsection{Platform Utilization}
\label{sec.utilization}

In this section we investigate how our best algorithms, in terms of maximum
stretch, compare to EASY in terms of platform utilization. We introduce a new
metric, called \emph{underutilization}. We contend that this metric helps
quantify schedule quality in terms of utilization, while remaining agnostic to
conditions that can confound traditional metrics, as explained hereafter.

\subsubsection{Measuring System Underutilization}

\begin{figure}[htb]
\centering
\includegraphics[width=0.45\columnwidth]{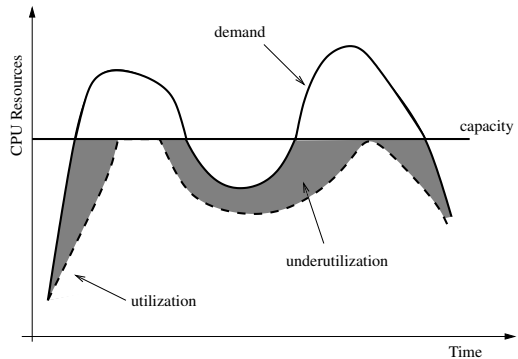}
\caption{Illustration of underutilization}
\label{fig.underutil}
\end{figure}

We have shown that DFRS algorithms are capable of dramatically improving the
performance delivered to jobs when compared with traditional batch scheduling
approaches.  However, job performance is not the sole concern for system
administrators. In fact, a common goal is to maximize platform utilization so as
to both improve throughput and justify the costs associated with maintaining a
cluster. Metrics used to evaluate machine utilization include throughput, daily
peak instantaneous utilization, and average instantaneous utilization over time.
However, these metrics are not appropriate for open, online systems as they are
highly dependent on the arrival process and the requirements of jobs in the
workload~\cite{frachtenberg2005ppj}. Another potential candidate, the makespan,
i.e., the amount of time elapsed between the submission of the first job in the
workload and the completion of last job, suffers from the problem that a very
short job may be submitted at the last possible instant, resulting in all
the scheduling algorithms having (nearly) the same
makespan~\cite{frachtenberg2005ppj}.

Instead, we argue that the quality of a scheduling algorithm (not considering
fairness to jobs) should be judged based on how well it meets demand over the
course of time, bounded by the resource constraints of the system. We call our
measure of this quantity the \emph{underutilization} and define it as follows.
For a fixed set of nodes $\procset$, let $\demand^{\sigma}_{\jobset}(t)$ be the
total CPU demand (i.e., sum of CPU needs) by jobs from $\jobset$ that have been
submitted but have not yet completed at time $t$ when scheduled using algorithm
$\sigma$, and let $u^{\sigma}_{\jobset}(t)$ represent the total system
utilization at time $t$ (i.e., sum of allocated CPU fractions) under the same
conditions. The underutilization of a system (assuming that all release dates
are $\ge 0$) is given by $\int_{0}^{\infty}
\min\{|\procset|,\demand^{\sigma}_{\jobset}(t)\} - u^{\sigma}_{\jobset}(t) dt$.
Note that $u^{\sigma}_{\jobset}(t)$ is constrained to always be less than both
$|\procset|$ and $\demand^{\sigma}_{\jobset}(t)$, and that
$\demand^{\sigma}_{\jobset}(t) = 0$ outside of the time span between when the
first job is submitted and the last one completes.

Figure~\ref{fig.underutil} provides a visual explanation. The horizontal
axis represents time while the vertical axis shows CPU resources. The solid
horizontal line shows the total system CPU capacity, $|\procset|$.  The solid
curve shows the total instantaneous demand, $\demand^{\sigma}_{\jobset}(t)$.
The dashed curve shows the total instantaneous utilization over time,
$u^{\sigma}_{\jobset}(t)$. The gray area, bounded above by the minimum of CPU
capacity and CPU demand, and below by CPU utilization, represents what we call
the \emph{underutilization} over the given time period.

Essentially, underutilization represents a cumulative measure over time of
computational power that could, at least theoretically, be exploited, but that
is instead sitting idle (or being used for non-useful work, such as preemption
or migration). Thus, lower values for underutilization are preferable to higher
values. As this quantity depends on total workload demand, which can vary
considerably, when combining results over a number of workloads we consider the
normalized underutilization, or the underutilization as a fraction of the total
resources required to execute the workload. A normalized underutilization value
of 1.0 would mean that, over the course of executing a workload, another
scheduling algorithm could have used an average of twice as much CPU resource
at any given time (ignoring memory and migration constraints). We contend that
for a fixed platform algorithms that do a better job of allocating resources
will tend to have smaller values for normalized underutilization in the average
case on a given set of workloads. It is important to note that normalized
underutilization will also vary with platform and workload characteristics, and
thus it should not be taken as an absolute measure of algorithm efficiency.

\subsubsection{Experimental Results}
\label{sec.period}

\begin{table*}[Htb]
\caption{Average normalized underutilization for selected algorithms on all
datasets}
\label{tab.underutilization}
\centering
\begin{tabular}{|l|c|c|c|}
\hline
Algorithm 
    & Real-world trace & Unscaled synthetic traces & Scaled synthetic traces\\
\hline
EASY                                    & 0.064  & 0.349 & 0.384\\
\greedyp\activeres/\periodic/\optmin/\mvt[600]  & 0.344  & 0.497 & 0.607\\
\greedypm\activeres/\periodic/\optmin/\mvt[600] & 0.344  & 0.497 & 0.608\\
\hline
\end{tabular}
\end{table*}

In this section we consider only the EASY batch scheduling algorithm and the
two best algorithms identified in Section~\ref{sec.stretchresults}:
\greedyp\activeres/\periodic/\optmin/\mvt[600] and
\greedypm\activeres/\periodic/\optmin/\mvt[600].
Table~\ref{tab.underutilization} shows average normalized underutilization
results for our three sets of traces. We can see that our two algorithms lead to
similar underutilization, and that they both lead to higher underutilization
than EASY.  This is particularly striking for the real-world HPC2N trace, for
which EASY leads to underutilization under 7\% while our algorithms are above
34\%.

While EASY scores better on this initial comparison, there are several factors
to consider. The first is that the very low underutilization shown by EASY
against the real-world trace may be partially due to an artifact of either user
or system behavior (e.g., the system from which the log was harvested also made
use of the EASY scheduler). Also for the real trace, the fact that we
arbitrarily split the entire HPC2N trace into 182 week-long periods may be a
factor. A simulation for the full trace yields a normalized underutilization of
8\% for EASY and of 10\% for our algorithms.

Still, our proposed algorithms also perform worse in terms of underutilization
on the synthetic traces, which should be free from both of the above problems.
We hypothesize that this lower efficiency is caused by time spent doing
preemption and migration. Recall that all our algorithms use a 600 second
period, equal to the rescheduling penalty, by default. By increasing the period,
within reasonable bounds, one can then hope to decrease underutilization. The
trade-off, however, is that the maximum stretch may be increased.

\begin{figure}[htb]
\centering
\subfigure[Real-world traces]{
  \includegraphics[width=0.3\textwidth]{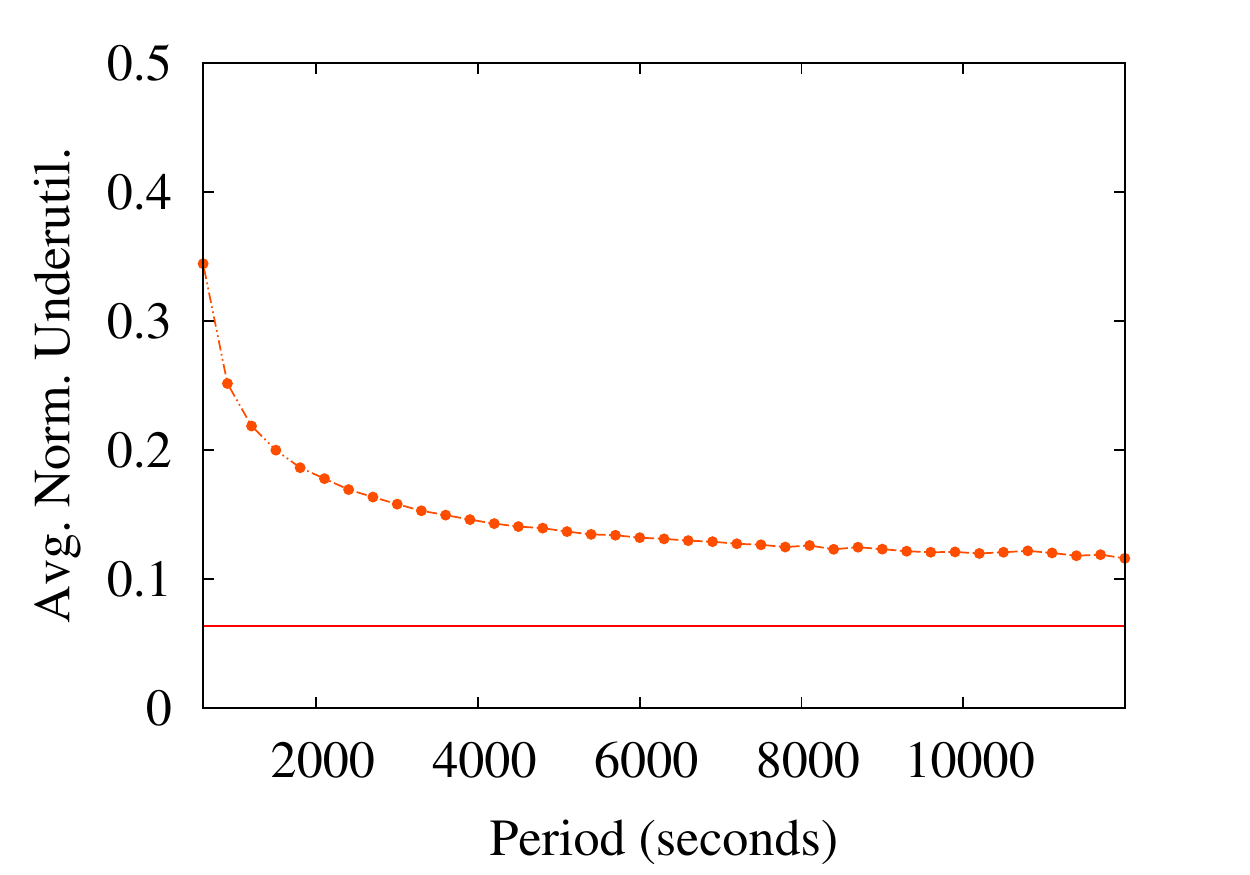}
  \label{fig.underutilization-vs-period-hpc2n}
}
\subfigure[Unscaled synthetic traces]{
  \includegraphics[width=0.3\textwidth]{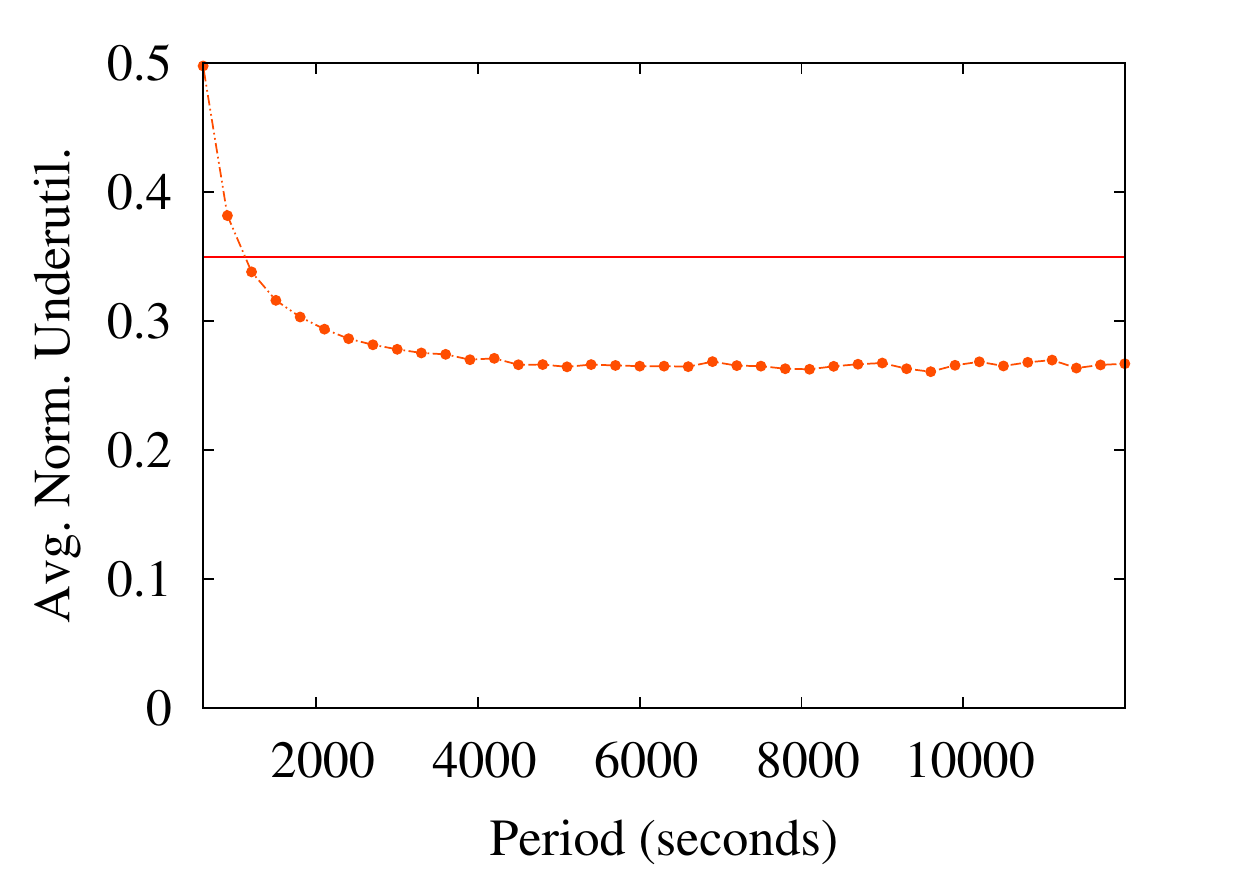}
  \label{fig.underutilization-vs-period-unscaled}
}
\subfigure[Scaled synthetic traces]{
  \includegraphics[width=0.3\textwidth]{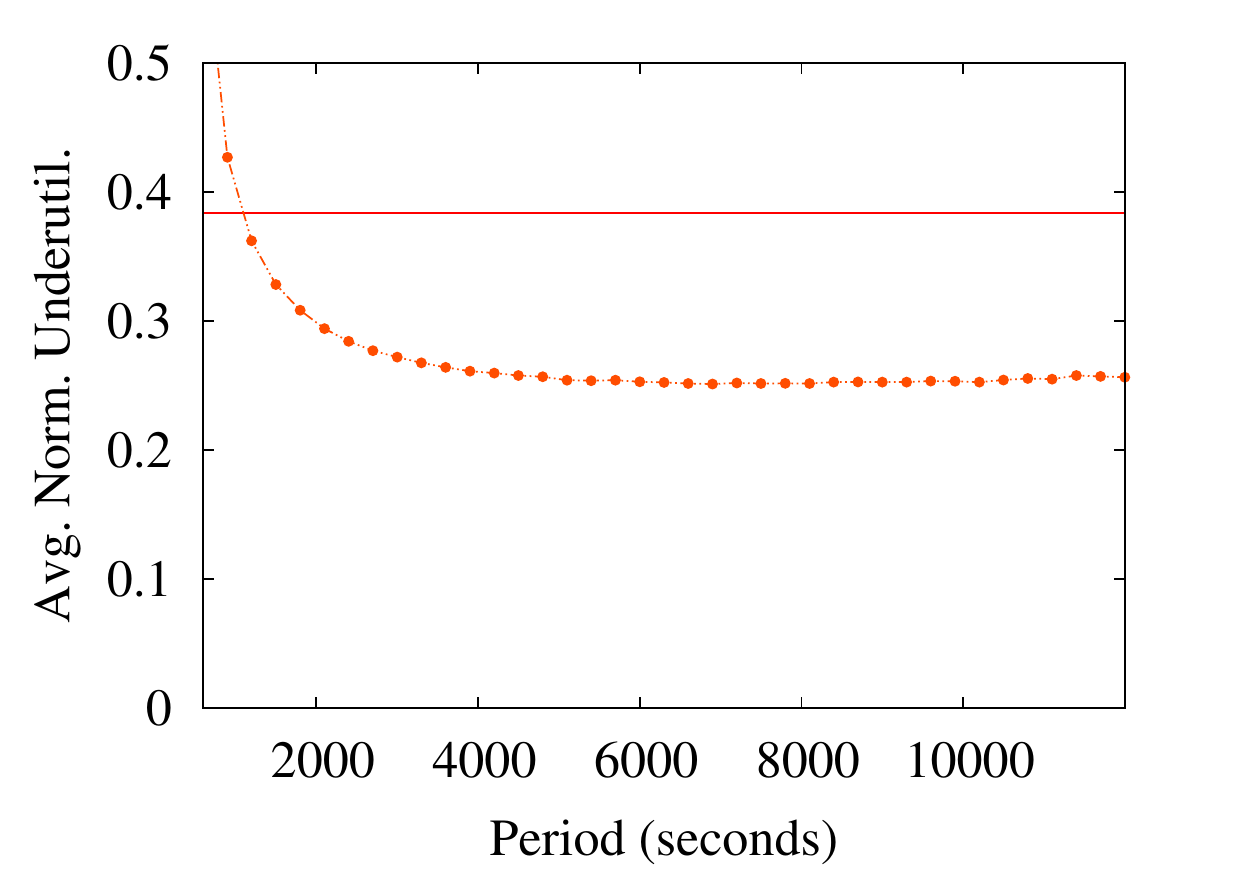}
  \label{fig.underutilization-vs-period-lubscaled}
}
\caption{Average normalized underutilization vs. period for EASY (solid) and  
  \greedypm\activeres/\periodic/\optmin/\mvt[600] (dots)}
\label{fig.underutilperiodic}
\end{figure}

Figure~\ref{fig.underutilperiodic} shows average normalized underutilization
results on our three sets of traces for EASY and the
\greedypm\activeres/\periodic/\optmin/\mvt[600] algorithm (results for
\greedyp\activeres/\periodic/\optmin/\mvt[600] are similar). For each set, we
plot the average normalized underutilization versus the period. In all graphs
the period varies from 600 to 12,000 seconds, i.e., from 2x to 20x the
rescheduling penalty. We do not use a period equal to the penalty as for some
traces it can result in a situation where no job ever makes any progress due to
thrashing. In all graphs normalized underutilization is shown to decrease
steadily. Additional results provided in Appendix~\ref{sec.extragraphs} show
that for extremely large periods (over 15,000 seconds for the synthetic traces)
underutilization expectedly begins to increase. For the two sets of synthetic
traces, as the period becomes larger than 900s (i.e., 1.5x the rescheduling
penalty), our algorithm achieves better average normalized underutilization than
EASY. For the real-world trace, our algorithm achieves higher values than EASY
regardless of the period: EASY achieves low values at around 6.4\%, while our
algorithm plateaus at around 11.8\%, before slowly starting to increase again
for larger periods (see Appendix~\ref{sec.extragraphs}).

\begin{figure}[htb]
\centering
\includegraphics[width=0.45\textwidth]{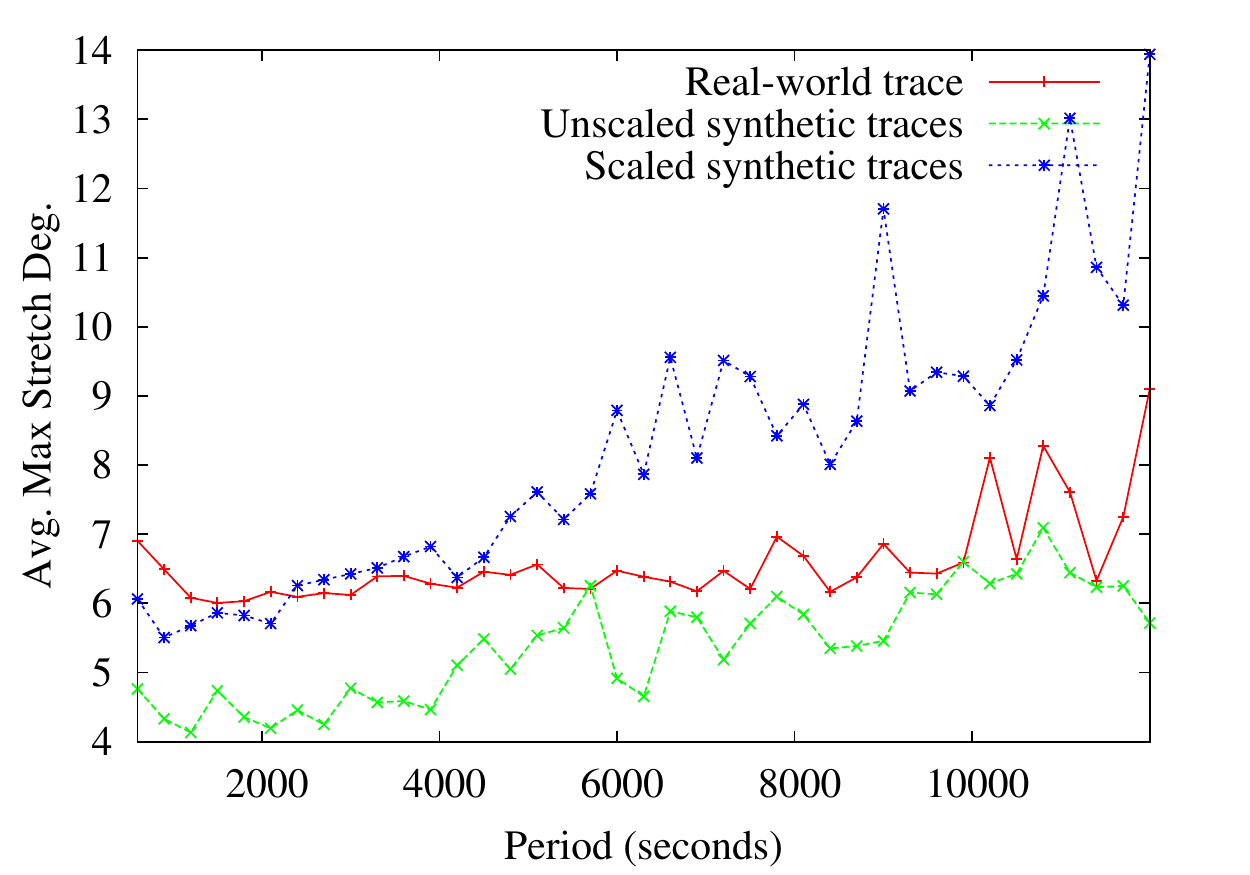}
\caption{Maximum stretch degradation from bound vs. scheduling period for 
    \greedypm\activeres/\periodic/\optmin/\mvt[600] for all three trace sets}
\label{fig.underutilperiodicstretch}
\end{figure}

Figure~\ref{fig.underutilperiodicstretch} shows the average maximum stretch for
our algorithm as the period increases, for each set of traces. The main
observation is that the average maximum stretch degradation increases slowly as
the period increases. The largest increase is seen for the scaled synthetic
traces, for which a 20-fold increase of the period leads to an increase of the
maximum stretch degradation by less than a factor 3. Results for the unscaled
synthetic traces and the real-world trace show much smaller increases (by less
than a factor 1.5 and a factor 2, respectively). Recall from
Section~\ref{sec.stretchresults} that EASY leads to maximum stretch degradation
orders of magnitude larger than our algorithms. Consequently, increasing the
period significantly, i.e., up to 20x the rescheduling penalty, still allows our
algorithms to outperform EASY by at least two orders of magnitude on average.

We conclude that our algorithms can outperform EASY by orders of magnitude in
terms of maximum stretch, and lead to only slightly higher or significantly
lower underutilization provided the period is set appropriately. Our results
indicate that picking a period roughly between 5x and 20x the rescheduling
penalty leads to good results. Our approach is thus robust to the choice of the
period, and in practice a period of, say, 1 hour, is appropriate. Results
summarized in Figure~\ref{fig.underutilperiodicbandwidth} of
Appendix~\ref{sec.extragraphs} show that with a 1-hour period the bandwidth
consumption due to preemptions and migrations is roughly 4 times lower
than that reported in Section~\ref{sec.bandwidthresults}, e.g., under 0.2GB/sec
on the average when considering the scaled synthetic traces with load values
$\geq 0.7$. As a conclusion, we recommend using the
\greedypm\activeres/\periodic/\optmin/\mvt[600] algorithm with a period equal to
10 times the rescheduling penalty.

\section{Related Work}
\label{sec.related}

Gang scheduling~\cite{ousterhout1982stc} is the classical approach to
time-sharing cluster resources. In gang scheduling, tasks in a parallel job are
executed during the same synchronized time slices across cluster nodes. This
requires distributed synchronized context-switching, which may impose
significant overhead and thus long time slices, although solutions have been
proposed~\cite{hori1998oap}. In this work we simply achieve time-sharing in an
uncoordinated and low-overhead manner via VM technology.

A second drawback of gang scheduling is the memory pressure problem, i.e., the
overhead of swapping to disk~\cite{batat2000gsm}. In our approach we completely
avoid swapping by enforcing the rule that a task may be assigned to a node only
if physical memory capacity is not exceeded. This precludes some time sharing
when compared to standard gang scheduling. However, this constraint is eminently
sensible given the performance penalties associated with swapping to disk and
the observation that many jobs in HPC workloads use only a fraction of physical
memory~\cite{setia1999ijm, batat2000gsm, chiang2001cls, li2004wcm}.

Other works have explored the problem of scheduling jobs without knowledge of
their processing time. The famous ``scheduling in the dark''
approach~\cite{edmonds1999sid} shows that, in the absence of knowledge, giving
equal resource shares to jobs is theoretically sound. We use the same approach
by ensuring that all jobs achieve the same yield. Our problem is also related to
the thread scheduling done in operating system kernels, given that thread
processing times are unknown. Our work differs in that we strive to optimize a
precisely defined objective function.

Our work is also related to several previous works that have explored
algorithmic issues pertaining to bin packing and/or multiprocessor scheduling.
There are obvious connections to \emph{fully dynamic} bin packing, a formulation
where items may arrive or depart at discrete time intervals and the goal is to
minimize the maximum number of bins required while limiting re-packing, as
studied by Ivkovic and Lloyd~\cite{ivkovic1999fda}. Coffman studies bin
stretching, a version of bin packing in which a bin may be stretched beyond its
normal capacity~\cite{coffman2006aae}. Epstein studies the online bin stretching
problem as a scheduling problem with the goal of minimizing
makespan~\cite{epstein2003bsr}. 

Our scheduling problem is strongly related to vector packing, i.e., bin packing
with multi-dimensional items and bins. Vector packing has been studied from both
a theoretical standpoint (i.e., guaranteed algorithms) and a pragmatic one
(i.e., efficient algorithms). In this work we employ an algorithm based on a
particular vector packing algorithm, MCB (Multi-Capacity Bin Packing), proposed
by Leinberger et al.~\cite{leinberger1999mcb}. We refer the reader
to~\cite{stillwell2010rav} for an extensive review of the literature on vector
packing algorithms.

Yossi Azar has studied online load balancing of temporary tasks on identical
machines with assignment restrictions~\cite{azar1997lbt}. Each task has a
weight and a duration. The weight is known when the task arrives, but the
duration is not known until the task completes. The problem therein is to assign
incoming tasks to nodes permanently so that the maximum load is minimized over
both nodes and time, which is related to maximizing the minimum yield. 

Finally, previous works have explored the use of VM technology in the HPC
domain. Some of the potential benefits of cluster virtualization include
increased reliability~\cite{emeneker2006irt}, load balancing, and consolidation
to reduce the number of active nodes~\cite{hermenier2009ecm}. Current VM
technology also allows for preemption and checkpointing of MPI
applications~\cite{emeneker2006irt}, as assumed in this work. The broad
consensus is that VM overhead does not represent a barrier to mainstream
deployment~\cite{huang2006chp}. Additional research has shown that the
performance impact on MPI applications is minimal~\cite{youseff2006epi} and that
cache interference effects do not cause significant performance degradation in
commonly-used libraries such as LAPACK and BLAS~\cite{youseff2008ipm}. 

Several groups in academia and the private sector are currently working on
platforms for centralized control of virtualized resources in a distributed
environment~\cite{aron2000crm, bhatia2007vcm, grit2006vmh, grit2007hvm,
hermenier2009ecm, mcnett2007uef, nurmi2008eoc, irwin2006snr,
ramakrishnan2006tdc, virtual_center, xen_enterprise}. These platforms generally
allow a central controller to create VMs, specify resource consumption levels,
migrate VMs between nodes, suspend running instances to disk, and, when
necessary, delete unruly instances. We base our model on such a system and the
capabilities it offers. The Entropy system recently developed by Hermenier et
al.~\cite{hermenier2009ecm} in fact implements all of the basic system
capabilities that we propose to exploit as well as its own set of resource
allocation algorithms, but their approach is based on searching for solutions to
an NP-complete constraint satisfaction problem, while our approach is to develop
polynomial time heuristic algorithms for a well-defined optimization problem.

Other groups have proposed VM-enhanced scheduling algorithms, but for the most
part these have been refinements of or extensions to existing schemes, such as
combining best-effort and reservation based jobs~\cite{sotomayor2008cbe}.
Fallenbeck et al. developed a novel solution that allows two virtual clusters to
coexist on the same physical cluster, with each physical node mapped to two
virtual nodes, but their implementation is meant to solve a problem specific to
a particular site and does not add to the literature in terms of general
scheduling algorithms~\cite{fallenbeck2006xac}. The system proposed by Ruth et
al.  attempts to dynamically migrate tasks from ``overloaded'' hosts, but their
definition of overloaded is vague and they do not propose a well-defined
objective function~\cite{ruth2006ala}.

\section{Conclusion}
\label{sec.conclusion}

In this paper we have proposed DFRS, a novel approach for job scheduling on a
homogeneous cluster. We have focused on an online, non-clairvoyant scenario in
which job processing times are unknown ahead of time. We have proposed several
scheduling algorithms and have compared them to standard batch scheduling
approaches using both real-world and synthetic workloads.  We have found that
several DFRS algorithms lead to dramatic improvement over batch scheduling in
terms of maximum (bounded) stretch. In particular, algorithms that periodically
apply the \mcb vector packing algorithms lead to the best results.  Our results
also show that the network bandwidth consumption of these algorithms for job
preemptions and migrations is only a small fraction of that available in current
clusters. Finally, we have shown that these algorithms can lead to good platform
utilization as long as the period at which \mcb is applied is chosen within a
broad range. The improvements shown in our results are likely to be larger in
practice due the many conservative assumptions in our evaluation methodology.

This work opens a number of promising directions for future research. Our
scheduling algorithms could be improved with a strategy for reducing the yield
of long running jobs. This strategy, inspired by thread scheduling in operating
systems kernels, would be particularly useful for mitigating the negative impact
of long running jobs on shorter ones, thereby improving fairness. While we
considered the case for HPC jobs composed of tasks with homogeneous resource
requirements and needs, the techniques that we developed could easily be
modified to allow for heterogeneous tasks as well (see our paper on the offline
problem~\cite{stillwell2010rav} for an expanded discussion of this issue). Also,
mechanisms for implementing user priorities, such as those supported in batch
scheduling systems, are needed. More broadly, a logical next step is to
implement and benchmark our algorithms as part of a prototype virtual cluster
management system that uses some of the resource need discovery techniques
described in Section~\ref{sec.system}.

\section*{Acknowledgment}

Simulations were carried out using the Grid'5000 experimental testbed, being
developed under the INRIA ALADDIN development action with support from CNRS,
RENATER and several Universities as well as other funding bodies (see
\url{https://www.grid5000.fr}).

\bibliographystyle{IEEEtran}
\bibliography{journals,biblio,procs}

\clearpage

\appendix

\section{Additional Tables}
\label{sec.extratabs}

\begin{center}
\begin{longtable}{|l||r|r|r||}
\caption[Average degradation from bound results for the HPC2N workload.]{%
  Average degradation from bound results for the real-world HPC2N workload. All
  results are for a 5-minute rescheduling penalty.%
  \label{sec.extratabs:tab.deg-from-bound-hpc2n-300-delay}}\\
\hline
Algorithm  & \multicolumn{3}{c||}{Degradation from bound}\\
            & avg. & std. & max \\
\hline
\endfirsthead
\caption[]{(Continued) Average degradation from bound results for the
  real-world HPC2N workload. All results are for a 5-minute rescheduling 
  penalty.}\\
\hline
Algorithm  & \multicolumn{3}{c||}{Degradation from bound}\\
            & avg. & std. & max\\
\hline
\endhead
\hline
\endfoot
\hline
\endlastfoot
\input{tabledata/extratabs/deg-from-bound-hpc2n-300-delay}
\end{longtable}
\end{center}

\pagebreak
\begin{center}
\begin{longtable}{|l||r|r|r||}
\caption[Average degradation from bound results for the unscaled synthetic
  traces.]{%
  Average degradation from bound results for the unscaled synthetic traces. All
  results are for a 5-minute rescheduling penalty.%
  \label{sec.extratabs:tab.deg-from-bound-lubtraces-300-delay}}\\
\hline
Algorithm  & \multicolumn{3}{c||}{Degradation from bound}\\
            & avg. & std. & max \\
\hline
\endfirsthead
\caption[]{(Continued) Average degradation from bound results for the unscaled
  synthetic traces. All results are for a 5-minute rescheduling penalty.}\\
\hline
Algorithm  & \multicolumn{3}{c||}{Degradation from bound}\\
            & avg. & std. & max\\
\hline
\endhead
\hline
\endfoot
\hline
\endlastfoot
\input{tabledata/extratabs/deg-from-bound-lubtraces-300-delay}
\end{longtable}
\end{center}

\pagebreak
\begin{center}
\begin{longtable}{|l||r|r|r||}
\caption[Average degradation from bound results for the scaled synthetic
  traces.]{%
  Average degradation from bound results for the scaled synthetic traces. All
  results are for a 5-minute rescheduling penalty.%
  \label{sec.extratabs:tab.deg-from-bound-lubscaled-300-delay}}\\
\hline
Algorithm  & \multicolumn{3}{c||}{Degradation from bound}\\
            & avg. & std. & max \\
\hline
\endfirsthead
\caption[]{(Continued) Average degradation from bound results for the scaled
  synthetic traces. All results are for a 5-minute rescheduling penalty.}\\
\hline
Algorithm  & \multicolumn{3}{c||}{Degradation from bound}\\
            & avg. & std. & max\\
\hline
\endhead
\hline
\endfoot
\hline
\endlastfoot
\input{tabledata/extratabs/deg-from-bound-lubscaled-300-delay}
\end{longtable}
\end{center}

\pagebreak
\begin{center}
\begin{longtable}{|l||r|r||r|r||}
\caption[Preemption and migration bandwidth consumption.]{%
  Preemption and migration bandwidth consumption for DFRS algorithms.
  Average and maximum values over scaled synthetic traces with load $\geq 0.7$.%
  \label{sec.extratabs:tab.bandwidth}}\\
\hline
Algorithm & \multicolumn{4}{c||}{Bandwidth consumption}\\
          & \multicolumn{4}{c||}{(GB / sec)}\\
          & \multicolumn{2}{c||}{pmtn} & \multicolumn{2}{c||}{mig}\\
          & avg. & max & avg. & max\\
\hline
\endfirsthead
\caption[]{(Continued) Preemption and migration bandwidth consumption
  for DFRS algorithms. Average and maximum values over scaled synthetic traces
  with load $\geq 0.7$.}\\
\hline
Algorithm & \multicolumn{4}{c||}{Bandwidth consumption}\\
          & \multicolumn{4}{c||}{(GB / sec)}\\
          & \multicolumn{2}{c||}{pmtn} & \multicolumn{2}{c||}{mig}\\
          & avg. & max & avg. & max\\
\hline
\endhead
\hline
\endfoot
\hline
\endlastfoot
\input{tabledata/extratabs/bandwidth-lubscaled-128-sload-gte-7-300-delay}
\end{longtable}
\end{center}

\pagebreak
\begin{center}
\begin{longtable}{|l||r|r||r|r||}
\caption[Preemption and migration frequency in occurrences per hour.]{%
  Preemption and migration frequency in terms of number of preemption and
  migration occurrences per hour. Average and maximum values over scaled
  synthetic traces with load $\geq 0.7$.%
  \label{sec.extratabs:tab.eventsperhour}}\\
\hline
Algorithm & \multicolumn{4}{c||}{Occurrences / hour}\\
          & \multicolumn{2}{c||}{pmtn} & \multicolumn{2}{c||}{mig}\\ 
          & avg. & max & avg. & max\\
\hline
\endfirsthead
\caption[]{(Continued) Preemption and migration frequency in terms of number of 
  preemption and migration occurrences per hour. Average and maximum values over
  scaled synthetic traces with load $\geq 0.7$.}\\
\hline
Algorithm & \multicolumn{4}{c||}{Occurrences / hour}\\
          & \multicolumn{2}{c||}{pmtn} & \multicolumn{2}{c||}{mig}\\ 
          & avg. & max & avg. & max\\
\hline
\endhead
\hline
\endfoot
\hline
\endlastfoot
\input{tabledata/extratabs/eventsperhour-lubscaled-sload-gte-7-300-delay}
\end{longtable}
\end{center}

\pagebreak
\begin{center}
\begin{longtable}{|l||r|r||r|r||}
\caption[Preemption and migration frequency in occurrences per job.]{%
  Preemption and migration frequency in terms of number of preemption and
  migration occurrences per job. Average and maximum values over scaled
  synthetic traces with load $\geq 0.7$.%
  \label{sec.extratabs:tab.eventsperjob}}\\
\hline
Algorithm & \multicolumn{4}{c||}{Occurrences / job}\\
          & \multicolumn{2}{c||}{pmtn} & \multicolumn{2}{c||}{mig}\\ 
          & avg. & max & avg. & max\\
\hline
\endfirsthead
\caption[]{(Continued) Preemption and migration frequency in terms of number of
  preemption and migration occurrences per job. Average and maximum values over
  scaled synthetic traces with load $\geq 0.7$.}\\
\hline
Algorithm & \multicolumn{4}{c||}{Occurrences / job}\\
          & \multicolumn{2}{c||}{pmtn} & \multicolumn{2}{c||}{mig}\\ 
          & avg. & max & avg. & max\\
\hline
\endhead
\hline
\endfoot
\hline
\endlastfoot
\input{tabledata/extratabs/eventsperjob-lubscaled-sload-gte-7-300-delay}
\end{longtable}
\end{center}

\clearpage

\section{Additional Graphs}
\label{sec.extragraphs}

\begin{figure}[htb]
\begin{minipage}[b]{0.45\linewidth}
\centering
\includegraphics[width=0.9\textwidth]{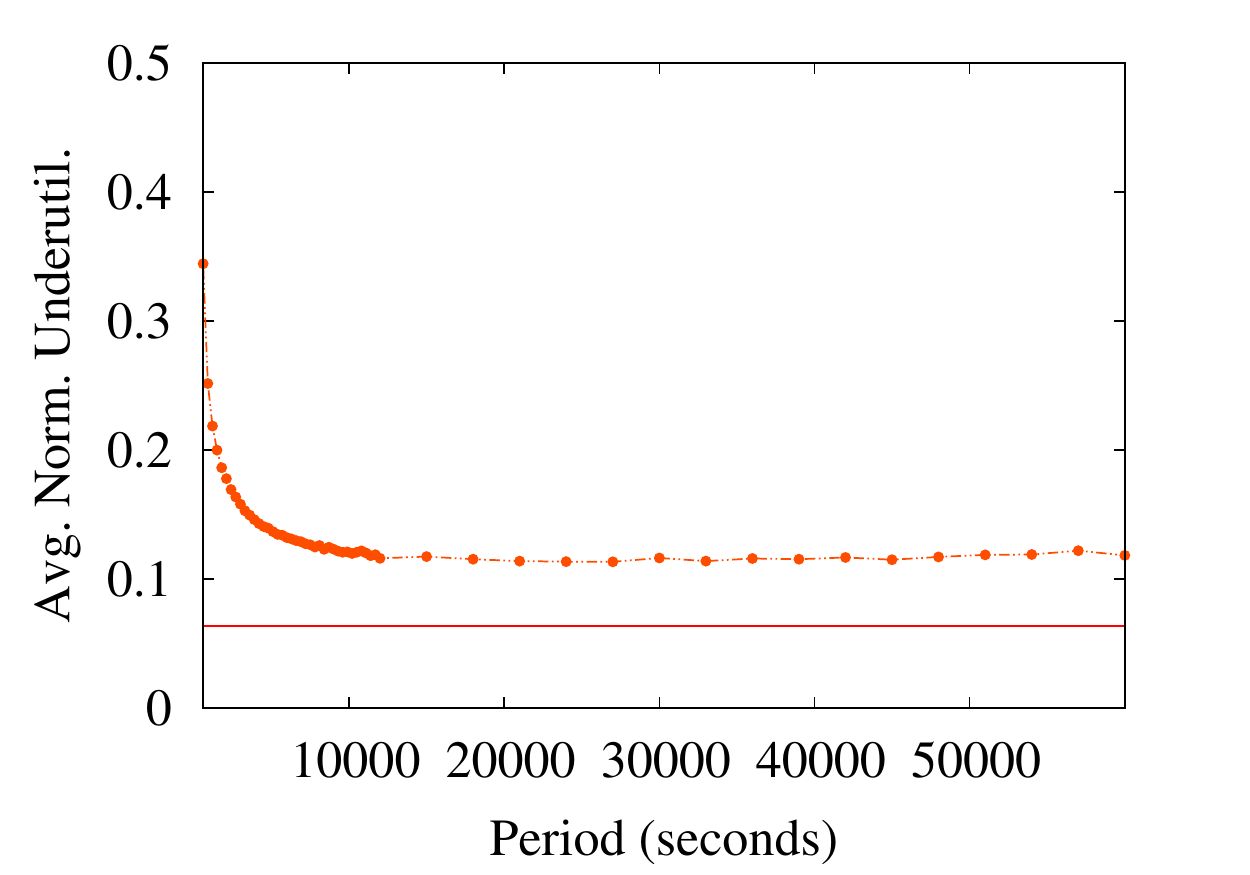}
\caption{Average normalized underutilization vs. period for EASY (solid) and  
  \greedypm\activeres/\periodic/\optmin/\mvt[600] (dots) 
  on Real-world traces, to 60,000 seconds}
\label{fig.underutilperiodic-hpc2n-120-ext}
\end{minipage}
\hspace{0.5cm}
\begin{minipage}[b]{0.45\linewidth}
\centering
\includegraphics[width=0.9\textwidth]{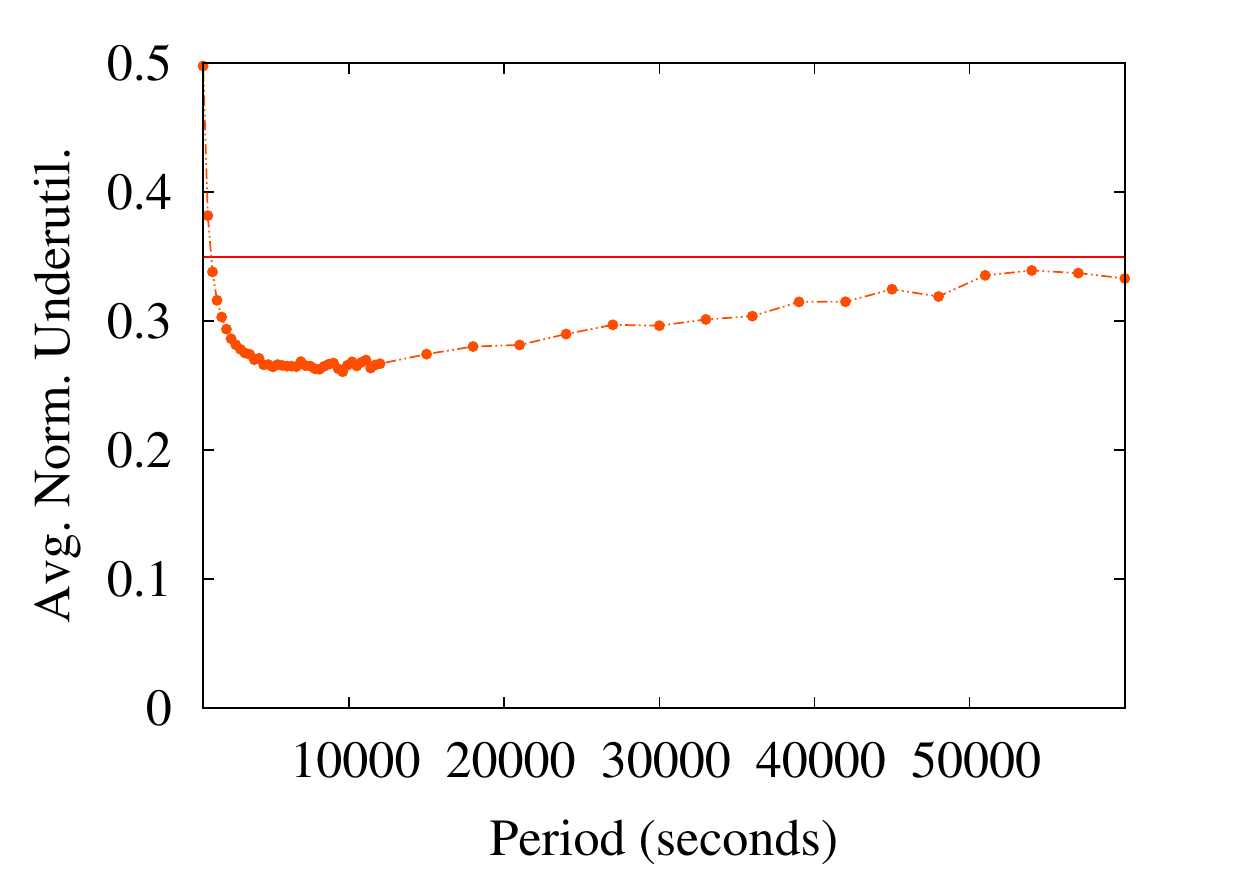}
\caption{Average normalized underutilization vs. period for EASY (solid) and  
  \greedypm\activeres/\periodic/\optmin/\mvt[600] (dots) 
  on Unscaled synthetic traces, to 60,000 seconds}
\label{fig.underutilperiodic-lubtraces-128-ext}
\end{minipage}
~\\
\begin{minipage}[b]{0.45\linewidth}
\centering
\includegraphics[width=0.9\textwidth]{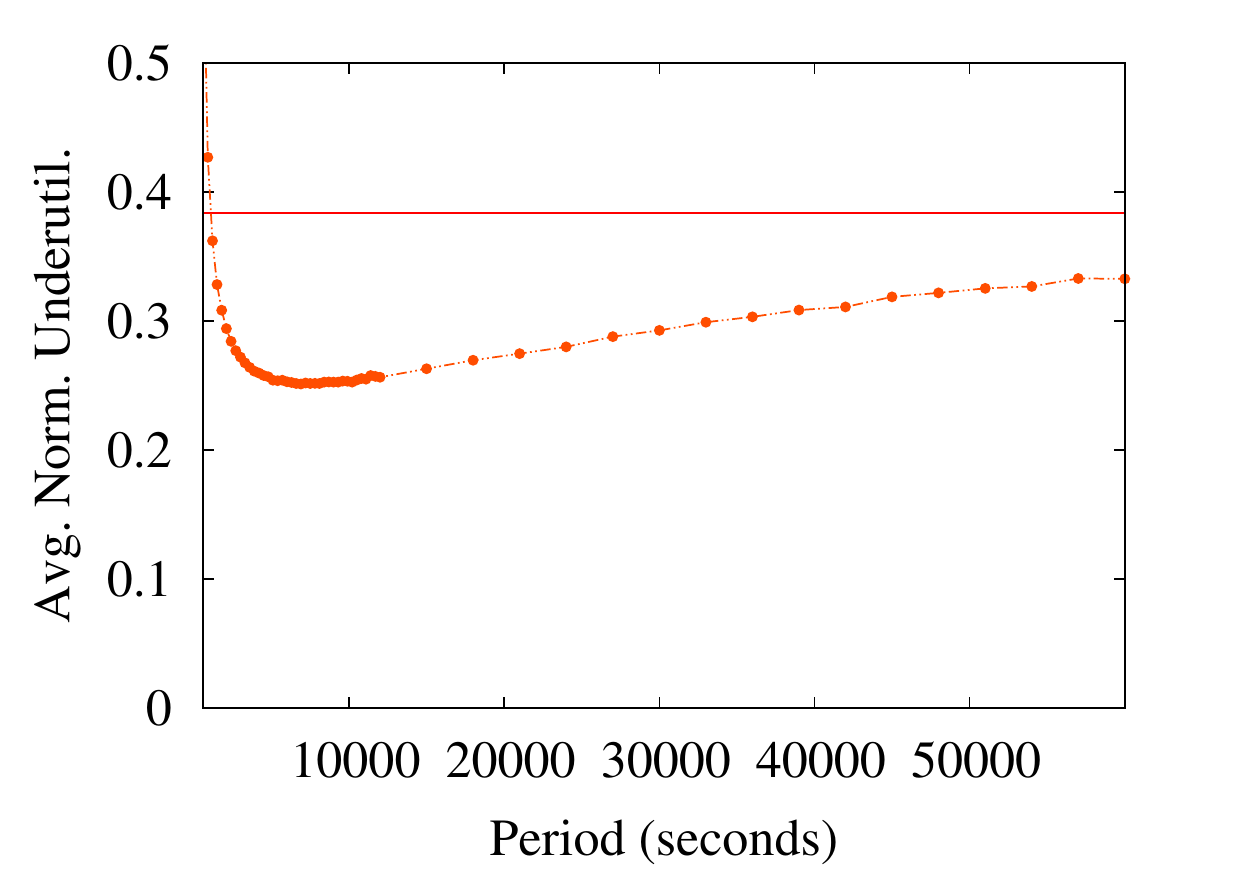}
\caption{Average normalized underutilization vs. period for EASY (solid) and
  \greedypm\activeres/\periodic/\optmin/\mvt[600] (dots)
  on Scaled synthetic traces, to 60,000 seconds}
\label{fig.underutilperiodic-lubscaled-128-ext}
\end{minipage}
\hspace{0.5cm}
\begin{minipage}[b]{0.45\linewidth}
~
\end{minipage}
\end{figure}

\begin{figure}[htb]
\centering
\begin{minipage}[b]{0.66\linewidth}
\centering
\includegraphics[width=0.9\textwidth]{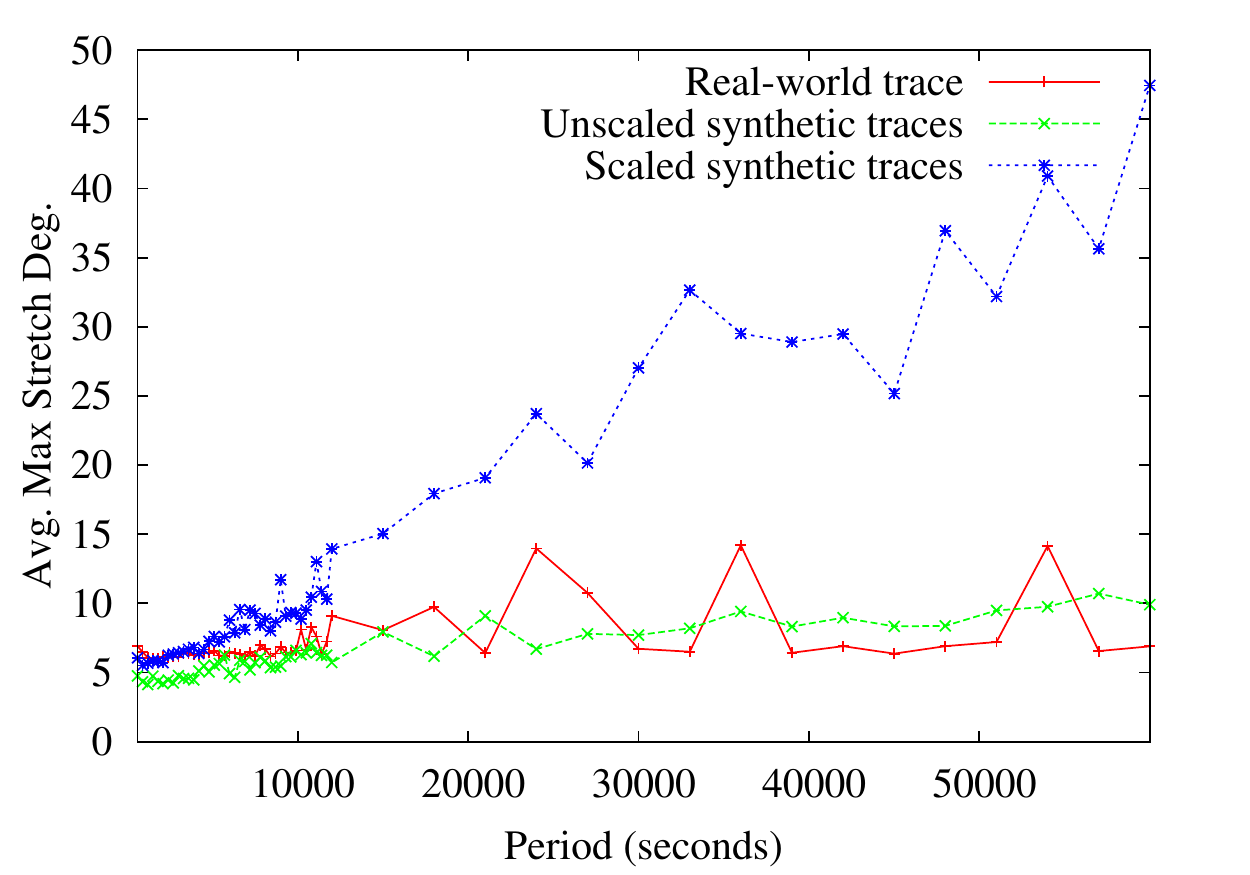}
\caption{Maximum stretch degradation from bound vs. scheduling period for
    \greedypm\activeres/\periodic/\optmin/\mvt[600] for all three trace sets,
    to 60,000 seconds}
\label{fig.underutilperiodicstretch-ext}
\end{minipage}
~\\
\begin{minipage}[b]{0.66\linewidth}
\centering
\includegraphics[width=0.9\textwidth]{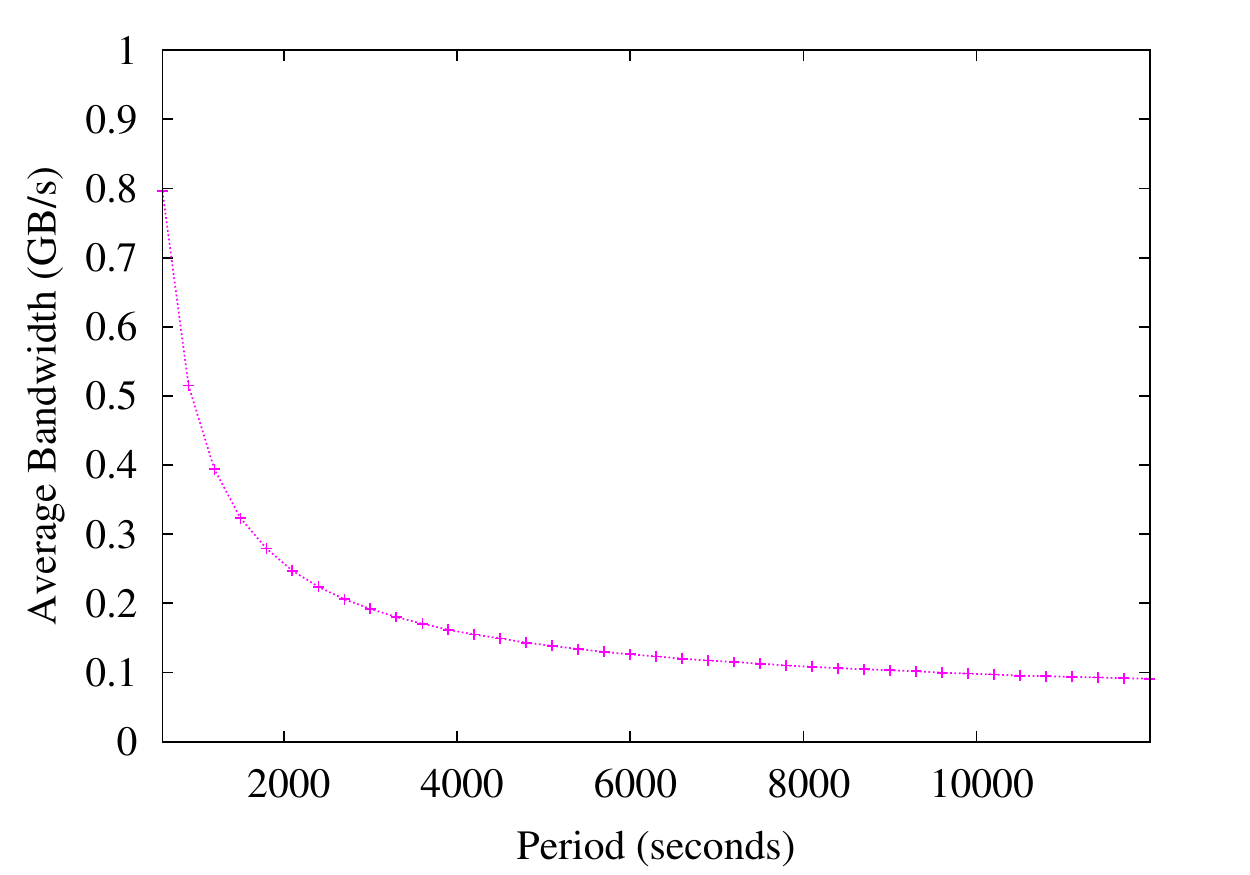}
\caption{Bandwidth consumption vs. period for
    \greedypm\activeres/\periodic/\optmin/\mvt[600] for the scaled synthetic
traces with load values $\geq 0.7$, to 12,000 seconds}
\label{fig.underutilperiodicbandwidth}
\end{minipage}
\end{figure}

\end{document}

%% file: remark.tex
\newif\ifremark
\long\def\remark#1{
\ifremark%
        \begingroup%
        \dimen0=\columnwidth
        \advance\dimen0 by -1in%
        \setbox0=\hbox{\parbox[b]{\dimen0}{\protect\em #1}}
        \dimen1=\ht0\advance\dimen1 by 2pt%
        \dimen2=\dp0\advance\dimen2 by 2pt%
        \vskip 0.25pt%
        \hbox to \columnwidth{%
                \vrule height\dimen1 width 3pt depth\dimen2%
                \hss\copy0\hss%
                \vrule height\dimen1 width 3pt depth\dimen2%
        }%
        \endgroup%
\fi}

%% file: tabledata/extratabs/deg-from-bound-hpc2n-300-delay.tex
FCFS&3,578.5&3,727.8&21,718.4\\
EASY&3,041.9&3,438.0&21,317.4\\
\greedy\activeres/\optavg&1,012.2&2,229.5&19,799.1\\
\greedy\activeres/\optmin&949.8&1,828.5&11,778.4\\
\greedy/\periodic/\optavg&28.9&27.0&212.9\\
\greedy/\periodic/\optavg/\mft[300]&23.3&26.9&182.3\\
\greedy/\periodic/\optavg/\mft[600]&23.5&27.6&212.2\\
\greedy/\periodic/\optavg/\mvt[300]&23.9&26.9&182.3\\
\greedy/\periodic/\optavg/\mvt[600]&23.8&27.8&182.3\\
\greedy/\periodic/\optmin&28.3&24.9&163.7\\
\greedy/\periodic/\optmin/\mft[300]&23.4&26.0&152.0\\
\greedy/\periodic/\optmin/\mft[600]&23.1&24.8&152.5\\
\greedy/\periodic/\optmin/\mvt[300]&23.5&26.9&182.8\\
\greedy/\periodic/\optmin/\mvt[600]&23.0&25.9&152.0\\
\greedy\activeres/\periodic/\optavg&24.5&15.9&81.6\\
\greedy\activeres/\periodic/\optavg/\mft[300]&19.8&17.8&85.6\\
\greedy\activeres/\periodic/\optavg/\mft[600]&19.3&17.6&85.6\\
\greedy\activeres/\periodic/\optavg/\mvt[300]&19.2&17.5&85.6\\
\greedy\activeres/\periodic/\optavg/\mvt[600]&18.9&17.3&74.9\\
\greedy\activeres/\periodic/\optmin&24.3&15.9&81.6\\
\greedy\activeres/\periodic/\optmin/\mft[300]&19.5&17.7&85.6\\
\greedy\activeres/\periodic/\optmin/\mft[600]&19.1&17.5&85.6\\
\greedy\activeres/\periodic/\optmin/\mvt[300]&18.9&17.2&85.6\\
\greedy\activeres/\periodic/\optmin/\mvt[600]&19.0&17.4&66.1\\
\greedyp\activeres/\optavg&20.4&116.7&1,254.2\\
\greedyp\activeres/\optmin&13.5&68.0&819.2\\
\greedyp/\periodic/\optavg&18.4&18.6&152.4\\
\greedyp/\periodic/\optavg/\mft[300]&9.1&18.8&152.4\\
\greedyp/\periodic/\optavg/\mft[600]&9.0&18.7&152.4\\
\greedyp/\periodic/\optavg/\mvt[300]&9.0&18.8&152.4\\
\greedyp/\periodic/\optavg/\mvt[600]&8.9&18.9&152.4\\
\greedyp/\periodic/\optmin&18.5&18.6&152.4\\
\greedyp/\periodic/\optmin/\mft[300]&9.1&18.8&152.4\\
\greedyp/\periodic/\optmin/\mft[600]&9.0&18.9&152.4\\
\greedyp/\periodic/\optmin/\mvt[300]&9.0&18.9&152.4\\
\greedyp/\periodic/\optmin/\mvt[600]&8.9&18.9&152.4\\
\greedyp\activeres/\periodic/\optavg&17.9&19.7&213.5\\
\greedyp\activeres/\periodic/\optavg/\mft[300]&8.2&19.4&213.5\\
\greedyp\activeres/\periodic/\optavg/\mft[600]&7.7&18.5&198.8\\
\greedyp\activeres/\periodic/\optavg/\mvt[300]&7.0&14.0&149.3\\
\greedyp\activeres/\periodic/\optavg/\mvt[600]&6.9&14.0&149.3\\
\greedyp\activeres/\periodic/\optmin&17.9&19.6&213.5\\
\greedyp\activeres/\periodic/\optmin/\mft[300]&7.9&19.0&213.5\\
\greedyp\activeres/\periodic/\optmin/\mft[600]&7.5&18.0&198.8\\
\greedyp\activeres/\periodic/\optmin/\mvt[300]&6.9&14.0&149.3\\
\greedyp\activeres/\periodic/\optmin/\mvt[600]&6.9&14.2&149.3\\
\greedypm\activeres/\optavg&14.1&72.7&880.1\\
\greedypm\activeres/\optmin&13.8&68.2&819.2\\
\greedypm/\periodic/\optavg&18.5&19.2&158.7\\
\greedypm/\periodic/\optavg/\mft[300]&9.1&19.4&158.7\\
\greedypm/\periodic/\optavg/\mft[600]&9.0&19.3&158.7\\
\greedypm/\periodic/\optavg/\mvt[300]&9.1&19.4&158.7\\
\greedypm/\periodic/\optavg/\mvt[600]&8.9&19.5&158.7\\
\greedypm/\periodic/\optmin&18.4&18.8&158.7\\
\greedypm/\periodic/\optmin/\mft[300]&9.2&19.2&158.7\\
\greedypm/\periodic/\optmin/\mft[600]&9.1&19.3&158.7\\
\greedypm/\periodic/\optmin/\mvt[300]&8.9&18.8&158.7\\
\greedypm/\periodic/\optmin/\mvt[600]&8.8&18.9&158.7\\
\greedypm\activeres/\periodic/\optavg&17.8&19.3&198.6\\
\greedypm\activeres/\periodic/\optavg/\mft[300]&8.2&19.1&198.6\\
\greedypm\activeres/\periodic/\optavg/\mft[600]&7.8&19.0&198.8\\
\greedypm\activeres/\periodic/\optavg/\mvt[300]&7.0&14.1&149.6\\
\greedypm\activeres/\periodic/\optavg/\mvt[600]&6.9&14.2&149.6\\
\greedypm\activeres/\periodic/\optmin&17.9&19.3&198.6\\
\greedypm\activeres/\periodic/\optmin/\mft[300]&8.1&19.2&198.6\\
\greedypm\activeres/\periodic/\optmin/\mft[600]&7.9&19.1&198.8\\
\greedypm\activeres/\periodic/\optmin/\mvt[300]&6.9&14.3&149.6\\
\greedypm\activeres/\periodic/\optmin/\mvt[600]&6.9&14.4&149.6\\
\mcb\activeres/\optavg&346.3&1,223.1&13,399.7\\
\mcb\activeres/\optavg/\mft[300]&44.7&144.2&1,082.8\\
\mcb\activeres/\optavg/\mft[600]&20.6&63.9&672.0\\
\mcb\activeres/\optavg/\mvt[300]&14.1&33.2&370.0\\
\mcb\activeres/\optavg/\mvt[600]&12.1&32.8&370.0\\
\mcb\activeres/\optmin&345.6&1,241.5&13,668.8\\
\mcb\activeres/\optmin/\mft[300]&44.7&138.4&1,126.3\\
\mcb\activeres/\optmin/\mft[600]&19.1&59.2&672.0\\
\mcb\activeres/\optmin/\mvt[300]&14.0&33.2&370.0\\
\mcb\activeres/\optmin/\mvt[600]&12.0&32.8&370.0\\
\mcb/\periodic/\optavg&171.4&702.4&8,383.0\\
\mcb/\periodic/\optavg/\mft[300]&15.0&33.5&279.3\\
\mcb/\periodic/\optavg/\mft[600]&12.1&26.9&292.4\\
\mcb/\periodic/\optavg/\mvt[300]&11.5&25.3&287.6\\
\mcb/\periodic/\optavg/\mvt[600]&10.7&25.2&287.6\\
\mcb/\periodic/\optmin&169.1&691.4&8,362.3\\
\mcb/\periodic/\optmin/\mft[300]&15.1&34.5&309.3\\
\mcb/\periodic/\optmin/\mft[600]&12.1&26.9&292.4\\
\mcb/\periodic/\optmin/\mvt[300]&11.5&25.3&287.6\\
\mcb/\periodic/\optmin/\mvt[600]&10.8&25.3&287.6\\
\mcb\activeres/\periodic/\optavg&394.5&1,562.9&17,862.0\\
\mcb\activeres/\periodic/\optavg/\mft[300]&56.1&197.4&1,849.6\\
\mcb\activeres/\periodic/\optavg/\mft[600]&23.6&71.3&799.3\\
\mcb\activeres/\periodic/\optavg/\mvt[300]&15.5&30.6&318.9\\
\mcb\activeres/\periodic/\optavg/\mvt[600]&13.7&30.3&318.9\\
\mcb\activeres/\periodic/\optmin&389.1&1,522.1&17,313.6\\
\mcb\activeres/\periodic/\optmin/\mft[300]&56.8&198.7&1,738.4\\
\mcb\activeres/\periodic/\optmin/\mft[600]&25.1&81.4&799.3\\
\mcb\activeres/\periodic/\optmin/\mvt[300]&15.4&30.6&318.9\\
\mcb\activeres/\periodic/\optmin/\mvt[600]&13.6&30.2&318.9\\
/\periodic/\optavg&105.0&445.6&5,011.9\\
/\periodic/\optavg/\mft[300]&105.0&445.6&5,011.9\\
/\periodic/\optavg/\mft[600]&105.0&445.6&5,011.9\\
/\periodic/\optavg/\mvt[300]&105.0&445.6&5,011.9\\
/\periodic/\optavg/\mvt[600]&105.0&445.6&5,011.9\\
/\periodic/\optmin&105.0&445.6&5,011.9\\
/\periodic/\optmin/\mft[300]&105.0&445.6&5,011.9\\
/\periodic/\optmin/\mft[600]&105.0&445.6&5,011.9\\
/\periodic/\optmin/\mvt[300]&105.0&445.6&5,011.9\\
/\periodic/\optmin/\mvt[600]&105.0&445.6&5,011.9\\
/\mcbsp/\optavg&105.0&445.6&5,011.9\\
/\mcbsp/\optavg/\mft[300]&105.0&445.6&5,011.9\\
/\mcbsp/\optavg/\mft[600]&105.0&445.6&5,011.9\\
/\mcbsp/\optavg/\mvt[300]&105.0&445.6&5,011.9\\
/\mcbsp/\optavg/\mvt[600]&105.0&445.6&5,011.9\\
/\mcbsp/\optmax&105.0&445.6&5,011.9\\
/\mcbsp/\optmax/\mft[300]&105.0&445.6&5,011.9\\
/\mcbsp/\optmax/\mft[600]&105.0&445.6&5,011.9\\
/\mcbsp/\optmax/\mvt[300]&105.0&445.6&5,011.9\\
/\mcbsp/\optmax/\mvt[600]&105.0&445.6&5,011.9\\

%% file: tabledata/extratabs/deg-from-bound-lubtraces-300-delay.tex
FCFS&5,457.2&2,958.5&15,102.7\\
EASY&4,955.4&2,730.6&14,036.8\\
\greedy\activeres/\optavg&2,527.1&2,472.3&12,487.5\\
\greedy\activeres/\optmin&2,435.0&2,285.6&11,229.9\\
\greedy/\periodic/\optavg&30.0&10.2&58.1\\
\greedy/\periodic/\optavg/\mft[300]&26.5&14.4&58.1\\
\greedy/\periodic/\optavg/\mft[600]&25.6&14.2&57.8\\
\greedy/\periodic/\optavg/\mvt[300]&25.7&14.5&57.8\\
\greedy/\periodic/\optavg/\mvt[600]&25.5&14.2&57.8\\
\greedy/\periodic/\optmin&30.1&10.2&58.1\\
\greedy/\periodic/\optmin/\mft[300]&26.0&14.3&58.1\\
\greedy/\periodic/\optmin/\mft[600]&25.9&14.5&58.0\\
\greedy/\periodic/\optmin/\mvt[300]&25.9&14.5&57.9\\
\greedy/\periodic/\optmin/\mvt[600]&25.9&14.2&58.0\\
\greedy\activeres/\periodic/\optavg&30.5&9.8&65.7\\
\greedy\activeres/\periodic/\optavg/\mft[300]&25.6&14.4&58.7\\
\greedy\activeres/\periodic/\optavg/\mft[600]&25.0&14.3&57.5\\
\greedy\activeres/\periodic/\optavg/\mvt[300]&25.3&14.4&57.5\\
\greedy\activeres/\periodic/\optavg/\mvt[600]&24.7&14.1&54.1\\
\greedy\activeres/\periodic/\optmin&30.4&9.7&65.7\\
\greedy\activeres/\periodic/\optmin/\mft[300]&25.1&14.3&57.5\\
\greedy\activeres/\periodic/\optmin/\mft[600]&24.9&14.3&57.5\\
\greedy\activeres/\periodic/\optmin/\mvt[300]&24.9&14.2&54.1\\
\greedy\activeres/\periodic/\optmin/\mvt[600]&24.6&14.3&54.1\\
\greedyp\activeres/\optavg&32.7&146.9&1,230.9\\
\greedyp\activeres/\optmin&37.5&156.0&1,204.9\\
\greedyp/\periodic/\optavg&20.2&7.2&38.1\\
\greedyp/\periodic/\optavg/\mft[300]&6.3&4.3&38.1\\
\greedyp/\periodic/\optavg/\mft[600]&6.1&4.4&38.1\\
\greedyp/\periodic/\optavg/\mvt[300]&6.0&3.9&27.5\\
\greedyp/\periodic/\optavg/\mvt[600]&6.0&4.5&38.1\\
\greedyp/\periodic/\optmin&20.1&7.3&38.1\\
\greedyp/\periodic/\optmin/\mft[300]&6.1&3.8&27.5\\
\greedyp/\periodic/\optmin/\mft[600]&6.1&4.5&38.1\\
\greedyp/\periodic/\optmin/\mvt[300]&5.9&3.8&27.5\\
\greedyp/\periodic/\optmin/\mvt[600]&5.9&4.5&38.1\\
\greedyp\activeres/\periodic/\optavg&20.4&6.8&32.0\\
\greedyp\activeres/\periodic/\optavg/\mft[300]&5.5&2.8&18.0\\
\greedyp\activeres/\periodic/\optavg/\mft[600]&5.1&2.8&18.0\\
\greedyp\activeres/\periodic/\optavg/\mvt[300]&4.9&2.4&13.6\\
\greedyp\activeres/\periodic/\optavg/\mvt[600]&4.8&2.4&13.6\\
\greedyp\activeres/\periodic/\optmin&20.3&6.8&32.0\\
\greedyp\activeres/\periodic/\optmin/\mft[300]&5.2&2.4&13.7\\
\greedyp\activeres/\periodic/\optmin/\mft[600]&5.0&2.7&18.0\\
\greedyp\activeres/\periodic/\optmin/\mvt[300]&4.9&2.7&18.0\\
\greedyp\activeres/\periodic/\optmin/\mvt[600]&4.9&2.9&19.2\\
\greedypm\activeres/\optavg&28.2&104.4&676.2\\
\greedypm\activeres/\optmin&33.8&154.0&1,321.7\\
\greedypm/\periodic/\optavg&20.2&7.2&38.1\\
\greedypm/\periodic/\optavg/\mft[300]&6.3&3.7&27.5\\
\greedypm/\periodic/\optavg/\mft[600]&6.1&4.4&38.1\\
\greedypm/\periodic/\optavg/\mvt[300]&6.2&4.5&38.1\\
\greedypm/\periodic/\optavg/\mvt[600]&5.9&4.4&38.1\\
\greedypm/\periodic/\optmin&20.2&7.3&38.1\\
\greedypm/\periodic/\optmin/\mft[300]&6.1&3.6&27.5\\
\greedypm/\periodic/\optmin/\mft[600]&6.0&4.4&38.1\\
\greedypm/\periodic/\optmin/\mvt[300]&6.0&3.9&27.5\\
\greedypm/\periodic/\optmin/\mvt[600]&5.9&4.5&38.1\\
\greedypm\activeres/\periodic/\optavg&20.4&6.8&32.0\\
\greedypm\activeres/\periodic/\optavg/\mft[300]&5.5&2.6&13.7\\
\greedypm\activeres/\periodic/\optavg/\mft[600]&5.0&2.5&13.7\\
\greedypm\activeres/\periodic/\optavg/\mvt[300]&4.9&2.5&13.8\\
\greedypm\activeres/\periodic/\optavg/\mvt[600]&4.8&2.4&13.6\\
\greedypm\activeres/\periodic/\optmin&20.3&6.9&32.0\\
\greedypm\activeres/\periodic/\optmin/\mft[300]&5.3&2.7&18.0\\
\greedypm\activeres/\periodic/\optmin/\mft[600]&4.9&2.5&13.7\\
\greedypm\activeres/\periodic/\optmin/\mvt[300]&4.9&2.7&18.0\\
\greedypm\activeres/\periodic/\optmin/\mvt[600]&4.8&2.4&13.6\\
\mcb\activeres/\optavg&245.1&130.3&634.2\\
\mcb\activeres/\optavg/\mft[300]&18.0&23.2&206.3\\
\mcb\activeres/\optavg/\mft[600]&9.8&6.4&43.6\\
\mcb\activeres/\optavg/\mvt[300]&8.6&5.6&43.9\\
\mcb\activeres/\optavg/\mvt[600]&7.7&6.9&44.8\\
\mcb\activeres/\optmin&233.2&117.1&634.2\\
\mcb\activeres/\optmin/\mft[300]&16.6&22.8&206.3\\
\mcb\activeres/\optmin/\mft[600]&9.9&8.1&51.3\\
\mcb\activeres/\optmin/\mvt[300]&9.2&8.0&65.3\\
\mcb\activeres/\optmin/\mvt[600]&6.9&5.4&44.4\\
\mcb/\periodic/\optavg&134.7&57.1&324.1\\
\mcb/\periodic/\optavg/\mft[300]&15.2&18.7&173.0\\
\mcb/\periodic/\optavg/\mft[600]&10.2&8.1&65.7\\
\mcb/\periodic/\optavg/\mvt[300]&9.2&6.8&51.3\\
\mcb/\periodic/\optavg/\mvt[600]&8.2&7.0&53.3\\
\mcb/\periodic/\optmin&133.7&57.5&323.7\\
\mcb/\periodic/\optmin/\mft[300]&14.5&18.6&173.0\\
\mcb/\periodic/\optmin/\mft[600]&10.0&8.1&65.7\\
\mcb/\periodic/\optmin/\mvt[300]&9.0&6.7&51.3\\
\mcb/\periodic/\optmin/\mvt[600]&8.1&6.6&53.3\\
\mcb\activeres/\periodic/\optavg&252.1&126.3&634.2\\
\mcb\activeres/\periodic/\optavg/\mft[300]&19.5&35.4&349.2\\
\mcb\activeres/\periodic/\optavg/\mft[600]&10.7&5.6&37.1\\
\mcb\activeres/\periodic/\optavg/\mvt[300]&8.8&3.5&19.0\\
\mcb\activeres/\periodic/\optavg/\mvt[600]&7.8&3.8&21.4\\
\mcb\activeres/\periodic/\optmin&250.6&125.0&634.2\\
\mcb\activeres/\periodic/\optmin/\mft[300]&19.0&35.3&349.2\\
\mcb\activeres/\periodic/\optmin/\mft[600]&10.6&5.7&37.1\\
\mcb\activeres/\periodic/\optmin/\mvt[300]&8.9&3.5&19.0\\
\mcb\activeres/\periodic/\optmin/\mvt[600]&7.8&3.9&21.9\\
/\periodic/\optavg&43.1&19.7&134.7\\
/\periodic/\optavg/\mft[300]&43.0&19.7&134.7\\
/\periodic/\optavg/\mft[600]&43.0&19.7&134.7\\
/\periodic/\optavg/\mvt[300]&43.0&19.8&134.7\\
/\periodic/\optavg/\mvt[600]&43.1&19.7&134.7\\
/\periodic/\optmin&43.0&19.8&134.7\\
/\periodic/\optmin/\mft[300]&43.0&19.8&134.7\\
/\periodic/\optmin/\mft[600]&43.0&19.8&134.7\\
/\periodic/\optmin/\mvt[300]&43.0&19.8&134.7\\
/\periodic/\optmin/\mvt[600]&43.0&19.7&134.7\\
/\mcbsp/\optavg&43.1&19.5&134.7\\
/\mcbsp/\optavg/\mft[300]&43.1&19.5&134.7\\
/\mcbsp/\optavg/\mft[600]&43.1&19.5&134.7\\
/\mcbsp/\optavg/\mvt[300]&43.1&19.5&134.7\\
/\mcbsp/\optavg/\mvt[600]&43.0&19.6&134.7\\
/\mcbsp/\optmax&43.0&19.6&134.7\\
/\mcbsp/\optmax/\mft[300]&43.0&19.6&134.7\\
/\mcbsp/\optmax/\mft[600]&43.0&19.6&134.7\\
/\mcbsp/\optmax/\mvt[300]&43.0&19.6&134.7\\
/\mcbsp/\optmax/\mvt[600]&43.0&19.6&134.7\\

%% file: tabledata/extratabs/deg-from-bound-lubscaled-300-delay.tex
FCFS&5,869.3&2,789.1&17,403.3\\
EASY&5,262.0&2,588.9&14,534.1\\
\greedy\activeres/\optavg&3,326.7&2,561.2&18,310.2\\
\greedy\activeres/\optmin&3,204.3&2,517.5&19,129.2\\
\greedy/\periodic/\optavg&29.2&14.3&153.2\\
\greedy/\periodic/\optavg/\mft[300]&27.4&15.5&153.2\\
\greedy/\periodic/\optavg/\mft[600]&27.3&15.5&152.6\\
\greedy/\periodic/\optavg/\mvt[300]&27.5&15.8&153.2\\
\greedy/\periodic/\optavg/\mvt[600]&27.4&15.9&153.0\\
\greedy/\periodic/\optmin&29.3&14.3&153.2\\
\greedy/\periodic/\optmin/\mft[300]&27.4&15.7&153.2\\
\greedy/\periodic/\optmin/\mft[600]&27.4&15.9&152.6\\
\greedy/\periodic/\optmin/\mvt[300]&27.6&15.9&153.4\\
\greedy/\periodic/\optmin/\mvt[600]&27.0&15.5&152.8\\
\greedy\activeres/\periodic/\optavg&29.2&11.9&101.4\\
\greedy\activeres/\periodic/\optavg/\mft[300]&26.7&13.6&87.1\\
\greedy\activeres/\periodic/\optavg/\mft[600]&25.5&13.2&95.2\\
\greedy\activeres/\periodic/\optavg/\mvt[300]&25.6&13.0&81.0\\
\greedy\activeres/\periodic/\optavg/\mvt[600]&25.4&13.1&95.2\\
\greedy\activeres/\periodic/\optmin&29.1&12.3&101.4\\
\greedy\activeres/\periodic/\optmin/\mft[300]&26.4&13.2&87.1\\
\greedy\activeres/\periodic/\optmin/\mft[600]&25.5&13.3&103.9\\
\greedy\activeres/\periodic/\optmin/\mvt[300]&25.4&13.0&81.0\\
\greedy\activeres/\periodic/\optmin/\mvt[600]&25.1&13.0&95.2\\
\greedyp\activeres/\optavg&114.3&617.3&9,490.0\\
\greedyp\activeres/\optmin&115.7&644.0&10,354.2\\
\greedyp/\periodic/\optavg&18.0&9.7&84.6\\
\greedyp/\periodic/\optavg/\mft[300]&7.7&7.9&84.6\\
\greedyp/\periodic/\optavg/\mft[600]&7.5&7.8&84.6\\
\greedyp/\periodic/\optavg/\mvt[300]&7.4&7.8&84.6\\
\greedyp/\periodic/\optavg/\mvt[600]&7.3&8.4&96.8\\
\greedyp/\periodic/\optmin&17.8&9.6&84.6\\
\greedyp/\periodic/\optmin/\mft[300]&7.6&7.9&84.6\\
\greedyp/\periodic/\optmin/\mft[600]&7.3&7.6&84.6\\
\greedyp/\periodic/\optmin/\mvt[300]&7.3&7.7&84.6\\
\greedyp/\periodic/\optmin/\mvt[600]&7.3&8.5&96.8\\
\greedyp\activeres/\periodic/\optavg&18.1&8.6&89.9\\
\greedyp\activeres/\periodic/\optavg/\mft[300]&7.1&5.6&90.2\\
\greedyp\activeres/\periodic/\optavg/\mft[600]&6.7&6.3&103.5\\
\greedyp\activeres/\periodic/\optavg/\mvt[300]&6.5&6.7&103.5\\
\greedyp\activeres/\periodic/\optavg/\mvt[600]&6.3&6.3&103.5\\
\greedyp\activeres/\periodic/\optmin&17.9&8.6&89.9\\
\greedyp\activeres/\periodic/\optmin/\mft[300]&6.8&6.4&103.5\\
\greedyp\activeres/\periodic/\optmin/\mft[600]&6.3&5.4&90.2\\
\greedyp\activeres/\periodic/\optmin/\mvt[300]&6.1&5.4&90.2\\
\greedyp\activeres/\periodic/\optmin/\mvt[600]&6.1&6.3&103.5\\
\greedypm\activeres/\optavg&124.9&658.4&9,404.5\\
\greedypm\activeres/\optmin&124.0&673.5&9,598.8\\
\greedypm/\periodic/\optavg&18.1&9.9&93.3\\
\greedypm/\periodic/\optavg/\mft[300]&7.8&7.5&84.6\\
\greedypm/\periodic/\optavg/\mft[600]&7.5&7.5&84.6\\
\greedypm/\periodic/\optavg/\mvt[300]&7.4&7.5&84.6\\
\greedypm/\periodic/\optavg/\mvt[600]&7.3&8.0&96.8\\
\greedypm/\periodic/\optmin&17.9&9.8&93.0\\
\greedypm/\periodic/\optmin/\mft[300]&7.6&7.4&84.6\\
\greedypm/\periodic/\optmin/\mft[600]&7.4&7.5&84.6\\
\greedypm/\periodic/\optmin/\mvt[300]&7.3&7.6&84.6\\
\greedypm/\periodic/\optmin/\mvt[600]&7.3&8.1&96.8\\
\greedypm\activeres/\periodic/\optavg&18.2&8.7&89.9\\
\greedypm\activeres/\periodic/\optavg/\mft[300]&7.1&5.2&80.1\\
\greedypm\activeres/\periodic/\optavg/\mft[600]&6.8&6.4&103.5\\
\greedypm\activeres/\periodic/\optavg/\mvt[300]&6.4&5.6&90.2\\
\greedypm\activeres/\periodic/\optavg/\mvt[600]&6.5&6.6&103.5\\
\greedypm\activeres/\periodic/\optmin&17.9&8.6&89.9\\
\greedypm\activeres/\periodic/\optmin/\mft[300]&6.9&6.5&103.5\\
\greedypm\activeres/\periodic/\optmin/\mft[600]&6.4&6.3&103.5\\
\greedypm\activeres/\periodic/\optmin/\mvt[300]&6.3&5.6&90.2\\
\greedypm\activeres/\periodic/\optmin/\mvt[600]&6.1&5.4&90.2\\
\mcb\activeres/\optavg&750.1&1,100.8&6,274.4\\
\mcb\activeres/\optavg/\mft[300]&121.9&351.2&3,609.6\\
\mcb\activeres/\optavg/\mft[600]&33.2&95.9&1,509.2\\
\mcb\activeres/\optavg/\mvt[300]&15.3&20.9&270.9\\
\mcb\activeres/\optavg/\mvt[600]&14.5&40.9&1,068.1\\
\mcb\activeres/\optmin&742.4&1,103.0&6,130.4\\
\mcb\activeres/\optmin/\mft[300]&117.8&358.4&3,680.3\\
\mcb\activeres/\optmin/\mft[600]&31.7&78.0&1,216.7\\
\mcb\activeres/\optmin/\mvt[300]&15.7&22.5&270.9\\
\mcb\activeres/\optmin/\mvt[600]&13.2&21.6&270.9\\
\mcb/\periodic/\optavg&155.8&122.8&913.3\\
\mcb/\periodic/\optavg/\mft[300]&23.5&26.0&231.5\\
\mcb/\periodic/\optavg/\mft[600]&15.8&19.2&231.4\\
\mcb/\periodic/\optavg/\mvt[300]&12.1&12.3&127.5\\
\mcb/\periodic/\optavg/\mvt[600]&11.1&12.6&127.5\\
\mcb/\periodic/\optmin&153.0&118.1&909.5\\
\mcb/\periodic/\optmin/\mft[300]&22.1&24.0&231.5\\
\mcb/\periodic/\optmin/\mft[600]&15.2&18.9&231.4\\
\mcb/\periodic/\optmin/\mvt[300]&12.3&14.2&223.0\\
\mcb/\periodic/\optmin/\mvt[600]&11.0&12.6&127.5\\
\mcb\activeres/\periodic/\optavg&959.5&1,469.0&8,299.4\\
\mcb\activeres/\periodic/\optavg/\mft[300]&168.2&516.8&5,469.8\\
\mcb\activeres/\periodic/\optavg/\mft[600]&40.3&155.5&2,941.5\\
\mcb\activeres/\periodic/\optavg/\mvt[300]&14.2&15.7&195.7\\
\mcb\activeres/\periodic/\optavg/\mvt[600]&12.0&14.5&195.7\\
\mcb\activeres/\periodic/\optmin&956.8&1,486.7&8,398.3\\
\mcb\activeres/\periodic/\optmin/\mft[300]&161.4&481.8&4,590.5\\
\mcb\activeres/\periodic/\optmin/\mft[600]&37.9&126.9&2,400.6\\
\mcb\activeres/\periodic/\optmin/\mvt[300]&14.4&17.4&222.2\\
\mcb\activeres/\periodic/\optmin/\mvt[600]&12.2&15.3&195.7\\
/\periodic/\optavg&40.4&25.1&238.3\\
/\periodic/\optavg/\mft[300]&40.4&25.0&238.3\\
/\periodic/\optavg/\mft[600]&40.4&25.1&238.3\\
/\periodic/\optavg/\mvt[300]&40.4&25.0&238.3\\
/\periodic/\optavg/\mvt[600]&40.4&25.0&238.3\\
/\periodic/\optmin&40.4&25.1&238.3\\
/\periodic/\optmin/\mft[300]&40.4&25.1&238.3\\
/\periodic/\optmin/\mft[600]&40.4&25.1&238.3\\
/\periodic/\optmin/\mvt[300]&40.4&25.0&238.3\\
/\periodic/\optmin/\mvt[600]&40.4&25.0&238.3\\
/\mcbsp/\optavg&40.2&24.8&236.5\\
/\mcbsp/\optavg/\mft[300]&40.2&24.8&236.5\\
/\mcbsp/\optavg/\mft[600]&40.2&24.8&236.5\\
/\mcbsp/\optavg/\mvt[300]&40.2&24.8&236.5\\
/\mcbsp/\optavg/\mvt[600]&40.2&24.8&236.5\\
/\mcbsp/\optmax&40.2&24.8&236.9\\
/\mcbsp/\optmax/\mft[300]&40.2&24.8&236.9\\
/\mcbsp/\optmax/\mft[600]&40.2&24.8&236.9\\
/\mcbsp/\optmax/\mvt[300]&40.2&24.8&236.9\\
/\mcbsp/\optmax/\mvt[600]&40.2&24.8&236.9\\

%% file: tabledata/extratabs/bandwidth-lubscaled-128-sload-gte-7-300-delay.tex
\greedy\activeres/\optavg&0.00&0.00&0.00&0.00\\
\greedy\activeres/\optmin&0.00&0.00&0.00&0.00\\
\greedy/\periodic/\optavg&0.44&1.02&0.21&0.63\\
\greedy/\periodic/\optavg/\mft[300]&0.44&1.04&0.20&0.62\\
\greedy/\periodic/\optavg/\mft[600]&0.44&1.03&0.20&0.62\\
\greedy/\periodic/\optavg/\mvt[300]&0.44&1.03&0.20&0.63\\
\greedy/\periodic/\optavg/\mvt[600]&0.44&1.04&0.19&0.60\\
\greedy/\periodic/\optmin&0.48&1.08&0.21&0.60\\
\greedy/\periodic/\optmin/\mft[300]&0.47&1.08&0.20&0.59\\
\greedy/\periodic/\optmin/\mft[600]&0.47&1.08&0.19&0.58\\
\greedy/\periodic/\optmin/\mvt[300]&0.47&1.06&0.19&0.57\\
\greedy/\periodic/\optmin/\mvt[600]&0.47&1.07&0.18&0.58\\
\greedy\activeres/\periodic/\optavg&0.45&1.28&0.28&0.69\\
\greedy\activeres/\periodic/\optavg/\mft[300]&0.45&1.27&0.27&0.66\\
\greedy\activeres/\periodic/\optavg/\mft[600]&0.44&1.26&0.26&0.67\\
\greedy\activeres/\periodic/\optavg/\mvt[300]&0.44&1.26&0.26&0.65\\
\greedy\activeres/\periodic/\optavg/\mvt[600]&0.44&1.26&0.25&0.65\\
\greedy\activeres/\periodic/\optmin&0.50&1.29&0.27&0.66\\
\greedy\activeres/\periodic/\optmin/\mft[300]&0.50&1.29&0.26&0.65\\
\greedy\activeres/\periodic/\optmin/\mft[600]&0.50&1.29&0.26&0.63\\
\greedy\activeres/\periodic/\optmin/\mvt[300]&0.50&1.27&0.26&0.63\\
\greedy\activeres/\periodic/\optmin/\mvt[600]&0.49&1.27&0.24&0.62\\
\greedyp\activeres/\optavg&0.06&0.17&0.00&0.00\\
\greedyp\activeres/\optmin&0.06&0.17&0.00&0.00\\
\greedyp/\periodic/\optavg&0.46&1.05&0.20&0.64\\
\greedyp/\periodic/\optavg/\mft[300]&0.46&1.07&0.19&0.61\\
\greedyp/\periodic/\optavg/\mft[600]&0.46&1.06&0.19&0.60\\
\greedyp/\periodic/\optavg/\mvt[300]&0.46&1.07&0.19&0.62\\
\greedyp/\periodic/\optavg/\mvt[600]&0.46&1.05&0.18&0.60\\
\greedyp/\periodic/\optmin&0.50&1.11&0.20&0.60\\
\greedyp/\periodic/\optmin/\mft[300]&0.49&1.11&0.19&0.58\\
\greedyp/\periodic/\optmin/\mft[600]&0.49&1.10&0.18&0.58\\
\greedyp/\periodic/\optmin/\mvt[300]&0.49&1.10&0.18&0.57\\
\greedyp/\periodic/\optmin/\mvt[600]&0.49&1.11&0.18&0.57\\
\greedyp\activeres/\periodic/\optavg&0.53&1.36&0.28&0.68\\
\greedyp\activeres/\periodic/\optavg/\mft[300]&0.52&1.35&0.27&0.69\\
\greedyp\activeres/\periodic/\optavg/\mft[600]&0.52&1.36&0.26&0.69\\
\greedyp\activeres/\periodic/\optavg/\mvt[300]&0.52&1.35&0.26&0.66\\
\greedyp\activeres/\periodic/\optavg/\mvt[600]&0.51&1.35&0.25&0.65\\
\greedyp\activeres/\periodic/\optmin&0.58&1.37&0.28&0.65\\
\greedyp\activeres/\periodic/\optmin/\mft[300]&0.57&1.38&0.26&0.65\\
\greedyp\activeres/\periodic/\optmin/\mft[600]&0.57&1.38&0.26&0.65\\
\greedyp\activeres/\periodic/\optmin/\mvt[300]&0.57&1.37&0.25&0.62\\
\greedyp\activeres/\periodic/\optmin/\mvt[600]&0.56&1.36&0.24&0.63\\
\greedypm\activeres/\optavg&0.03&0.08&0.02&0.05\\
\greedypm\activeres/\optmin&0.03&0.07&0.02&0.05\\
\greedypm/\periodic/\optavg&0.46&1.05&0.21&0.64\\
\greedypm/\periodic/\optavg/\mft[300]&0.46&1.05&0.20&0.64\\
\greedypm/\periodic/\optavg/\mft[600]&0.45&1.04&0.20&0.61\\
\greedypm/\periodic/\optavg/\mvt[300]&0.45&1.06&0.20&0.62\\
\greedypm/\periodic/\optavg/\mvt[600]&0.45&1.06&0.19&0.60\\
\greedypm/\periodic/\optmin&0.49&1.10&0.21&0.60\\
\greedypm/\periodic/\optmin/\mft[300]&0.49&1.10&0.20&0.61\\
\greedypm/\periodic/\optmin/\mft[600]&0.49&1.10&0.19&0.58\\
\greedypm/\periodic/\optmin/\mvt[300]&0.49&1.10&0.19&0.57\\
\greedypm/\periodic/\optmin/\mvt[600]&0.49&1.10&0.18&0.57\\
\greedypm\activeres/\periodic/\optavg&0.51&1.33&0.29&0.68\\
\greedypm\activeres/\periodic/\optavg/\mft[300]&0.50&1.33&0.28&0.69\\
\greedypm\activeres/\periodic/\optavg/\mft[600]&0.50&1.35&0.27&0.68\\
\greedypm\activeres/\periodic/\optavg/\mvt[300]&0.49&1.33&0.27&0.67\\
\greedypm\activeres/\periodic/\optavg/\mvt[600]&0.49&1.34&0.26&0.65\\
\greedypm\activeres/\periodic/\optmin&0.56&1.37&0.29&0.66\\
\greedypm\activeres/\periodic/\optmin/\mft[300]&0.55&1.36&0.27&0.66\\
\greedypm\activeres/\periodic/\optmin/\mft[600]&0.55&1.36&0.27&0.65\\
\greedypm\activeres/\periodic/\optmin/\mvt[300]&0.55&1.36&0.27&0.64\\
\greedypm\activeres/\periodic/\optmin/\mvt[600]&0.54&1.34&0.26&0.62\\
\mcb\activeres/\optavg&0.38&1.15&1.16&2.35\\
\mcb\activeres/\optavg/\mft[300]&0.18&0.99&0.86&2.98\\
\mcb\activeres/\optavg/\mft[600]&0.14&0.67&0.73&2.60\\
\mcb\activeres/\optavg/\mvt[300]&0.13&0.48&0.61&2.09\\
\mcb\activeres/\optavg/\mvt[600]&0.12&0.39&0.53&1.49\\
\mcb\activeres/\optmin&0.42&1.26&1.17&2.36\\
\mcb\activeres/\optmin/\mft[300]&0.20&0.98&0.86&2.89\\
\mcb\activeres/\optmin/\mft[600]&0.15&0.78&0.75&2.72\\
\mcb\activeres/\optmin/\mvt[300]&0.13&0.72&0.63&2.12\\
\mcb\activeres/\optmin/\mvt[600]&0.13&0.37&0.53&1.51\\
\mcb/\periodic/\optavg&0.52&1.07&0.69&2.92\\
\mcb/\periodic/\optavg/\mft[300]&0.50&1.05&0.55&2.14\\
\mcb/\periodic/\optavg/\mft[600]&0.49&1.04&0.51&1.66\\
\mcb/\periodic/\optavg/\mvt[300]&0.49&1.07&0.47&1.32\\
\mcb/\periodic/\optavg/\mvt[600]&0.49&1.06&0.43&1.18\\
\mcb/\periodic/\optmin&0.56&1.10&0.69&3.03\\
\mcb/\periodic/\optmin/\mft[300]&0.54&1.10&0.55&2.04\\
\mcb/\periodic/\optmin/\mft[600]&0.53&1.09&0.51&1.68\\
\mcb/\periodic/\optmin/\mvt[300]&0.53&1.11&0.47&1.40\\
\mcb/\periodic/\optmin/\mvt[600]&0.53&1.12&0.43&1.12\\
\mcb\activeres/\periodic/\optavg&0.67&1.08&1.21&2.54\\
\mcb\activeres/\periodic/\optavg/\mft[300]&0.54&1.05&0.89&3.12\\
\mcb\activeres/\periodic/\optavg/\mft[600]&0.51&1.03&0.77&2.69\\
\mcb\activeres/\periodic/\optavg/\mvt[300]&0.50&1.07&0.65&2.08\\
\mcb\activeres/\periodic/\optavg/\mvt[600]&0.50&1.09&0.57&1.53\\
\mcb\activeres/\periodic/\optmin&0.72&1.15&1.21&2.68\\
\mcb\activeres/\periodic/\optmin/\mft[300]&0.58&1.12&0.89&3.02\\
\mcb\activeres/\periodic/\optmin/\mft[600]&0.55&1.10&0.77&2.76\\
\mcb\activeres/\periodic/\optmin/\mvt[300]&0.54&1.11&0.65&2.16\\
\mcb\activeres/\periodic/\optmin/\mvt[600]&0.54&1.11&0.56&1.53\\
/\periodic/\optavg&0.45&1.02&0.21&0.64\\
/\periodic/\optavg/\mft[300]&0.45&1.03&0.21&0.64\\
/\periodic/\optavg/\mft[600]&0.45&1.04&0.21&0.64\\
/\periodic/\optavg/\mvt[300]&0.45&1.02&0.21&0.63\\
/\periodic/\optavg/\mvt[600]&0.45&1.03&0.20&0.62\\
/\periodic/\optmin&0.49&1.07&0.21&0.62\\
/\periodic/\optmin/\mft[300]&0.49&1.07&0.21&0.62\\
/\periodic/\optmin/\mft[600]&0.49&1.07&0.21&0.62\\
/\periodic/\optmin/\mvt[300]&0.49&1.08&0.20&0.60\\
/\periodic/\optmin/\mvt[600]&0.49&1.08&0.19&0.58\\
/\mcbsp/\optavg&0.28&0.66&0.39&0.79\\
/\mcbsp/\optavg/\mft[300]&0.28&0.66&0.39&0.78\\
/\mcbsp/\optavg/\mft[600]&0.28&0.66&0.39&0.78\\
/\mcbsp/\optavg/\mvt[300]&0.28&0.68&0.38&0.79\\
/\mcbsp/\optavg/\mvt[600]&0.28&0.68&0.37&0.78\\
/\mcbsp/\optmax&0.28&0.65&0.39&0.81\\
/\mcbsp/\optmax/\mft[300]&0.28&0.65&0.39&0.81\\
/\mcbsp/\optmax/\mft[600]&0.28&0.65&0.39&0.81\\
/\mcbsp/\optmax/\mvt[300]&0.28&0.65&0.38&0.81\\
/\mcbsp/\optmax/\mvt[600]&0.28&0.64&0.37&0.78\\

%% file: tabledata/extratabs/eventsperhour-lubscaled-sload-gte-7-300-delay.tex
\greedy\activeres/\optavg&0.00&0.00&0.00&0.00\\
\greedy\activeres/\optmin&0.00&0.00&0.00&0.00\\
\greedy/\periodic/\optavg&30.12&75.96&38.81&124.56\\
\greedy/\periodic/\optavg/\mft[300]&29.82&75.24&36.41&110.88\\
\greedy/\periodic/\optavg/\mft[600]&29.63&75.96&35.46&112.32\\
\greedy/\periodic/\optavg/\mvt[300]&29.69&74.88&35.46&118.08\\
\greedy/\periodic/\optavg/\mvt[600]&29.55&74.52&33.91&107.64\\
\greedy/\periodic/\optmin&32.58&83.52&38.79&110.52\\
\greedy/\periodic/\optmin/\mft[300]&32.28&82.80&36.23&107.28\\
\greedy/\periodic/\optmin/\mft[600]&32.18&82.80&35.15&106.56\\
\greedy/\periodic/\optmin/\mvt[300]&32.25&82.80&35.11&103.68\\
\greedy/\periodic/\optmin/\mvt[600]&32.04&83.16&33.60&103.68\\
\greedy\activeres/\periodic/\optavg&26.07&73.08&55.45&126.00\\
\greedy\activeres/\periodic/\optavg/\mft[300]&25.68&71.64&53.16&119.16\\
\greedy\activeres/\periodic/\optavg/\mft[600]&25.39&71.28&51.96&117.36\\
\greedy\activeres/\periodic/\optavg/\mvt[300]&25.23&71.28&51.74&123.12\\
\greedy\activeres/\periodic/\optavg/\mvt[600]&24.86&70.56&50.07&120.24\\
\greedy\activeres/\periodic/\optmin&29.27&84.96&58.06&124.56\\
\greedy\activeres/\periodic/\optmin/\mft[300]&28.74&83.16&55.46&123.84\\
\greedy\activeres/\periodic/\optmin/\mft[600]&28.58&83.88&54.32&123.12\\
\greedy\activeres/\periodic/\optmin/\mvt[300]&28.50&83.52&53.86&120.96\\
\greedy\activeres/\periodic/\optmin/\mvt[600]&28.08&83.52&51.97&117.36\\
\greedyp\activeres/\optavg&5.78&20.52&0.00&0.00\\
\greedyp\activeres/\optmin&5.67&18.00&0.00&0.00\\
\greedyp/\periodic/\optavg&30.86&76.68&37.70&111.96\\
\greedyp/\periodic/\optavg/\mft[300]&30.46&74.88&35.37&108.72\\
\greedyp/\periodic/\optavg/\mft[600]&30.31&75.24&34.51&110.88\\
\greedyp/\periodic/\optavg/\mvt[300]&30.39&77.04&34.63&106.56\\
\greedyp/\periodic/\optavg/\mvt[600]&30.17&77.04&33.08&103.68\\
\greedyp/\periodic/\optmin&33.34&85.32&37.52&107.64\\
\greedyp/\periodic/\optmin/\mft[300]&32.95&83.88&35.05&104.40\\
\greedyp/\periodic/\optmin/\mft[600]&32.79&84.60&34.07&103.68\\
\greedyp/\periodic/\optmin/\mvt[300]&32.88&84.60&34.17&105.84\\
\greedyp/\periodic/\optmin/\mvt[600]&32.70&83.52&32.74&101.52\\
\greedyp\activeres/\periodic/\optavg&37.08&83.52&56.42&123.12\\
\greedyp\activeres/\periodic/\optavg/\mft[300]&35.90&84.96&53.48&123.12\\
\greedyp\activeres/\periodic/\optavg/\mft[600]&35.58&84.60&52.38&120.60\\
\greedyp\activeres/\periodic/\optavg/\mvt[300]&35.33&84.24&52.26&117.00\\
\greedyp\activeres/\periodic/\optavg/\mvt[600]&34.70&82.80&50.41&114.84\\
\greedyp\activeres/\periodic/\optmin&39.67&97.92&58.56&127.08\\
\greedyp\activeres/\periodic/\optmin/\mft[300]&38.67&98.28&55.64&124.20\\
\greedyp\activeres/\periodic/\optmin/\mft[600]&38.39&96.48&54.46&122.40\\
\greedyp\activeres/\periodic/\optmin/\mvt[300]&38.27&97.56&53.82&119.52\\
\greedyp\activeres/\periodic/\optmin/\mvt[600]&37.70&97.20&52.14&119.16\\
\greedypm\activeres/\optavg&2.31&9.72&3.75&14.04\\
\greedypm\activeres/\optmin&2.25&10.08&3.69&13.32\\
\greedypm/\periodic/\optavg&30.36&77.04&39.63&114.84\\
\greedypm/\periodic/\optavg/\mft[300]&29.96&75.24&37.04&114.12\\
\greedypm/\periodic/\optavg/\mft[600]&29.79&77.40&36.13&109.80\\
\greedypm/\periodic/\optavg/\mvt[300]&29.90&75.60&36.03&113.40\\
\greedypm/\periodic/\optavg/\mvt[600]&29.70&77.04&34.83&116.28\\
\greedypm/\periodic/\optmin&32.93&84.24&39.21&112.32\\
\greedypm/\periodic/\optmin/\mft[300]&32.46&84.60&36.78&108.36\\
\greedypm/\periodic/\optmin/\mft[600]&32.34&83.88&35.87&108.72\\
\greedypm/\periodic/\optmin/\mvt[300]&32.44&84.24&35.78&108.72\\
\greedypm/\periodic/\optmin/\mvt[600]&32.24&83.16&34.31&106.56\\
\greedypm\activeres/\periodic/\optavg&32.78&80.28&60.90&129.60\\
\greedypm\activeres/\periodic/\optavg/\mft[300]&31.93&81.00&57.83&123.48\\
\greedypm\activeres/\periodic/\optavg/\mft[600]&31.57&81.00&56.72&119.88\\
\greedypm\activeres/\periodic/\optavg/\mvt[300]&31.19&78.84&56.59&121.32\\
\greedypm\activeres/\periodic/\optavg/\mvt[600]&30.60&79.56&54.88&120.96\\
\greedypm\activeres/\periodic/\optmin&35.78&96.48&62.95&135.72\\
\greedypm\activeres/\periodic/\optmin/\mft[300]&34.88&95.04&59.84&132.48\\
\greedypm\activeres/\periodic/\optmin/\mft[600]&34.45&92.88&58.68&130.32\\
\greedypm\activeres/\periodic/\optmin/\mvt[300]&34.25&94.32&58.31&128.52\\
\greedypm\activeres/\periodic/\optmin/\mvt[600]&33.80&94.32&56.45&127.08\\
\mcb\activeres/\optavg&56.54&214.56&470.14&998.28\\
\mcb\activeres/\optavg/\mft[300]&20.23&162.72&327.56&1,313.64\\
\mcb\activeres/\optavg/\mft[600]&12.86&90.36&279.56&1,201.68\\
\mcb\activeres/\optavg/\mvt[300]&10.04&51.84&238.86&1,101.60\\
\mcb\activeres/\optavg/\mvt[600]&9.27&37.80&206.62&822.60\\
\mcb\activeres/\optmin&61.66&230.40&490.48&1,005.48\\
\mcb\activeres/\optmin/\mft[300]&21.60&174.96&337.86&1,343.16\\
\mcb\activeres/\optmin/\mft[600]&13.24&102.96&291.29&1,306.80\\
\mcb\activeres/\optmin/\mvt[300]&10.25&78.48&247.80&1,083.24\\
\mcb\activeres/\optmin/\mvt[600]&9.51&38.16&212.16&825.48\\
\mcb/\periodic/\optavg&38.66&83.52&235.57&1,181.88\\
\mcb/\periodic/\optavg/\mft[300]&35.19&80.28&164.35&798.48\\
\mcb/\periodic/\optavg/\mft[600]&34.72&79.56&149.35&709.92\\
\mcb/\periodic/\optavg/\mvt[300]&34.57&80.64&139.75&600.84\\
\mcb/\periodic/\optavg/\mvt[600]&34.15&78.48&124.47&480.24\\
\mcb/\periodic/\optmin&41.59&86.76&243.13&1,269.00\\
\mcb/\periodic/\optmin/\mft[300]&38.09&86.04&167.67&826.56\\
\mcb/\periodic/\optmin/\mft[600]&37.47&84.96&151.97&755.64\\
\mcb/\periodic/\optmin/\mvt[300]&37.32&86.40&142.42&619.92\\
\mcb/\periodic/\optmin/\mvt[600]&36.93&86.76&126.13&490.68\\
\mcb\activeres/\periodic/\optavg&70.98&163.44&461.18&1,045.44\\
\mcb\activeres/\periodic/\optavg/\mft[300]&44.52&128.52&314.64&1,351.80\\
\mcb\activeres/\periodic/\optavg/\mft[600]&38.07&83.52&263.46&1,245.60\\
\mcb\activeres/\periodic/\optavg/\mvt[300]&35.99&80.64&222.61&1,057.32\\
\mcb\activeres/\periodic/\optavg/\mvt[600]&35.09&77.04&191.91&795.24\\
\mcb\activeres/\periodic/\optmin&77.04&170.28&479.31&1,109.52\\
\mcb\activeres/\periodic/\optmin/\mft[300]&47.63&135.00&325.09&1,370.16\\
\mcb\activeres/\periodic/\optmin/\mft[600]&41.08&93.60&271.59&1,337.76\\
\mcb\activeres/\periodic/\optmin/\mvt[300]&38.88&86.40&227.80&1,129.32\\
\mcb\activeres/\periodic/\optmin/\mvt[600]&37.94&85.32&194.57&836.28\\
/\periodic/\optavg&31.29&76.32&38.76&113.76\\
/\periodic/\optavg/\mft[300]&31.17&75.96&39.00&114.84\\
/\periodic/\optavg/\mft[600]&31.23&75.96&38.96&113.40\\
/\periodic/\optavg/\mvt[300]&31.16&75.24&37.86&111.24\\
/\periodic/\optavg/\mvt[600]&31.04&75.24&36.62&111.60\\
/\periodic/\optmin&33.83&84.24&38.69&111.24\\
/\periodic/\optmin/\mft[300]&33.83&84.24&38.69&111.24\\
/\periodic/\optmin/\mft[600]&33.83&84.24&38.69&111.24\\
/\periodic/\optmin/\mvt[300]&33.88&83.52&37.62&107.64\\
/\periodic/\optmin/\mvt[600]&33.76&84.60&36.21&104.04\\
/\mcbsp/\optavg&20.64&43.20&62.89&140.04\\
/\mcbsp/\optavg/\mft[300]&20.65&42.48&63.03&140.40\\
/\mcbsp/\optavg/\mft[600]&20.62&43.56&62.88&138.96\\
/\mcbsp/\optavg/\mvt[300]&20.62&43.20&61.58&136.80\\
/\mcbsp/\optavg/\mvt[600]&20.60&43.56&59.78&132.48\\
/\mcbsp/\optmax&20.41&45.36&67.26&159.48\\
/\mcbsp/\optmax/\mft[300]&20.41&45.36&67.26&159.48\\
/\mcbsp/\optmax/\mft[600]&20.41&45.36&67.26&159.48\\
/\mcbsp/\optmax/\mvt[300]&20.35&46.44&65.60&158.40\\
/\mcbsp/\optmax/\mvt[600]&20.25&46.80&63.87&153.36\\

%% file: tabledata/extratabs/eventsperjob-lubscaled-sload-gte-7-300-delay.tex
\greedy\activeres/\optavg&0.00&0.00&0.00&0.00\\
\greedy\activeres/\optmin&0.00&0.00&0.00&0.00\\
\greedy/\periodic/\optavg&5.03&21.58&4.79&18.26\\
\greedy/\periodic/\optavg/\mft[300]&4.98&21.52&4.52&15.80\\
\greedy/\periodic/\optavg/\mft[600]&4.94&20.66&4.41&16.26\\
\greedy/\periodic/\optavg/\mvt[300]&4.95&21.03&4.40&17.08\\
\greedy/\periodic/\optavg/\mvt[600]&4.91&21.01&4.22&15.65\\
\greedy/\periodic/\optmin&5.41&21.76&4.81&16.17\\
\greedy/\periodic/\optmin/\mft[300]&5.36&21.55&4.52&15.40\\
\greedy/\periodic/\optmin/\mft[600]&5.33&21.31&4.39&15.27\\
\greedy/\periodic/\optmin/\mvt[300]&5.34&21.83&4.37&15.21\\
\greedy/\periodic/\optmin/\mvt[600]&5.29&21.94&4.20&14.83\\
\greedy\activeres/\periodic/\optavg&3.95&20.26&6.59&17.29\\
\greedy\activeres/\periodic/\optavg/\mft[300]&3.89&19.76&6.36&16.62\\
\greedy\activeres/\periodic/\optavg/\mft[600]&3.85&19.57&6.23&16.52\\
\greedy\activeres/\periodic/\optavg/\mvt[300]&3.81&19.60&6.18&16.42\\
\greedy\activeres/\periodic/\optavg/\mvt[600]&3.75&19.24&5.98&16.01\\
\greedy\activeres/\periodic/\optmin&4.49&22.55&6.94&17.65\\
\greedy\activeres/\periodic/\optmin/\mft[300]&4.41&22.07&6.66&17.31\\
\greedy\activeres/\periodic/\optmin/\mft[600]&4.38&21.89&6.55&17.38\\
\greedy\activeres/\periodic/\optmin/\mvt[300]&4.37&21.71&6.48&17.26\\
\greedy\activeres/\periodic/\optmin/\mvt[600]&4.29&21.80&6.25&16.53\\
\greedyp\activeres/\optavg&0.58&2.09&0.00&0.00\\
\greedyp\activeres/\optmin&0.57&2.04&0.00&0.00\\
\greedyp/\periodic/\optavg&5.16&21.11&4.67&16.20\\
\greedyp/\periodic/\optavg/\mft[300]&5.09&21.05&4.41&15.45\\
\greedyp/\periodic/\optavg/\mft[600]&5.05&20.97&4.32&15.12\\
\greedyp/\periodic/\optavg/\mvt[300]&5.06&20.52&4.33&15.29\\
\greedyp/\periodic/\optavg/\mvt[600]&5.02&20.36&4.14&14.68\\
\greedyp/\periodic/\optmin&5.54&22.20&4.67&15.76\\
\greedyp/\periodic/\optmin/\mft[300]&5.47&21.45&4.39&15.09\\
\greedyp/\periodic/\optmin/\mft[600]&5.45&21.41&4.28&15.30\\
\greedyp/\periodic/\optmin/\mvt[300]&5.46&21.94&4.29&15.04\\
\greedyp/\periodic/\optmin/\mvt[600]&5.41&21.37&4.11&14.29\\
\greedyp\activeres/\periodic/\optavg&5.37&22.68&6.78&17.07\\
\greedyp\activeres/\periodic/\optavg/\mft[300]&5.21&22.27&6.47&17.30\\
\greedyp\activeres/\periodic/\optavg/\mft[600]&5.17&21.34&6.35&17.24\\
\greedyp\activeres/\periodic/\optavg/\mvt[300]&5.12&21.58&6.32&16.78\\
\greedyp\activeres/\periodic/\optavg/\mvt[600]&5.03&21.52&6.09&16.26\\
\greedyp\activeres/\periodic/\optmin&5.87&24.87&7.09&17.91\\
\greedyp\activeres/\periodic/\optmin/\mft[300]&5.73&23.92&6.78&17.45\\
\greedyp\activeres/\periodic/\optmin/\mft[600]&5.69&24.22&6.65&17.17\\
\greedyp\activeres/\periodic/\optmin/\mvt[300]&5.66&24.30&6.54&17.04\\
\greedyp\activeres/\periodic/\optmin/\mvt[600]&5.56&23.88&6.34&16.48\\
\greedypm\activeres/\optavg&0.24&1.11&0.37&1.18\\
\greedypm\activeres/\optmin&0.23&1.19&0.36&1.22\\
\greedypm/\periodic/\optavg&5.11&20.79&4.88&16.67\\
\greedypm/\periodic/\optavg/\mft[300]&5.03&20.61&4.60&15.78\\
\greedypm/\periodic/\optavg/\mft[600]&5.01&20.41&4.50&15.52\\
\greedypm/\periodic/\optavg/\mvt[300]&5.02&20.99&4.48&15.41\\
\greedypm/\periodic/\optavg/\mvt[600]&4.98&21.06&4.33&16.83\\
\greedypm/\periodic/\optmin&5.52&21.65&4.85&16.03\\
\greedypm/\periodic/\optmin/\mft[300]&5.43&21.97&4.58&15.71\\
\greedypm/\periodic/\optmin/\mft[600]&5.41&21.39&4.49&15.59\\
\greedypm/\periodic/\optmin/\mvt[300]&5.42&21.67&4.45&15.15\\
\greedypm/\periodic/\optmin/\mvt[600]&5.37&22.19&4.28&15.07\\
\greedypm\activeres/\periodic/\optavg&4.88&22.02&7.30&18.27\\
\greedypm\activeres/\periodic/\optavg/\mft[300]&4.76&21.49&6.98&17.71\\
\greedypm\activeres/\periodic/\optavg/\mft[600]&4.70&21.30&6.85&17.02\\
\greedypm\activeres/\periodic/\optavg/\mvt[300]&4.65&21.03&6.81&17.11\\
\greedypm\activeres/\periodic/\optavg/\mvt[600]&4.55&20.70&6.61&17.24\\
\greedypm\activeres/\periodic/\optmin&5.42&24.00&7.59&18.32\\
\greedypm\activeres/\periodic/\optmin/\mft[300]&5.29&23.42&7.27&18.23\\
\greedypm\activeres/\periodic/\optmin/\mft[600]&5.23&23.76&7.14&17.80\\
\greedypm\activeres/\periodic/\optmin/\mvt[300]&5.19&23.12&7.07&17.32\\
\greedypm\activeres/\periodic/\optmin/\mvt[600]&5.11&23.23&6.84&16.94\\
\mcb\activeres/\optavg&6.10&17.36&55.01&86.16\\
\mcb\activeres/\optavg/\mft[300]&2.02&14.26&31.45&80.96\\
\mcb\activeres/\optavg/\mft[600]&1.31&9.68&25.24&65.40\\
\mcb\activeres/\optavg/\mvt[300]&1.03&2.88&20.89&42.69\\
\mcb\activeres/\optavg/\mvt[600]&0.97&3.42&18.17&41.67\\
\mcb\activeres/\optmin&6.67&18.83&57.46&90.29\\
\mcb\activeres/\optmin/\mft[300]&2.18&15.58&32.59&89.17\\
\mcb\activeres/\optmin/\mft[600]&1.37&8.67&26.28&68.39\\
\mcb\activeres/\optmin/\mvt[300]&1.06&3.08&21.73&46.10\\
\mcb\activeres/\optmin/\mvt[600]&1.00&2.89&18.68&41.15\\
\mcb/\periodic/\optavg&6.34&25.03&25.02&41.57\\
\mcb/\periodic/\optavg/\mft[300]&6.04&24.52&18.31&36.68\\
\mcb/\periodic/\optavg/\mft[600]&5.97&24.57&16.87&35.67\\
\mcb/\periodic/\optavg/\mvt[300]&5.94&25.07&15.76&33.88\\
\mcb/\periodic/\optavg/\mvt[600]&5.85&24.90&14.18&33.19\\
\mcb/\periodic/\optmin&6.83&26.38&25.74&43.24\\
\mcb/\periodic/\optmin/\mft[300]&6.51&26.03&18.75&38.89\\
\mcb/\periodic/\optmin/\mft[600]&6.42&26.10&17.23&36.68\\
\mcb/\periodic/\optmin/\mvt[300]&6.39&26.37&16.13&36.79\\
\mcb/\periodic/\optmin/\mvt[600]&6.30&26.13&14.45&34.48\\
\mcb\activeres/\periodic/\optavg&11.97&26.78&70.57&106.14\\
\mcb\activeres/\periodic/\optavg/\mft[300]&7.54&26.28&39.94&105.55\\
\mcb\activeres/\periodic/\optavg/\mft[600]&6.59&24.99&31.32&78.30\\
\mcb\activeres/\periodic/\optavg/\mvt[300]&6.24&25.69&25.53&53.96\\
\mcb\activeres/\periodic/\optavg/\mvt[600]&6.08&25.58&22.24&50.41\\
\mcb\activeres/\periodic/\optmin&13.01&28.93&73.61&117.15\\
\mcb\activeres/\periodic/\optmin/\mft[300]&8.11&27.86&41.44&112.18\\
\mcb\activeres/\periodic/\optmin/\mft[600]&7.13&27.14&32.52&81.60\\
\mcb\activeres/\periodic/\optmin/\mvt[300]&6.72&27.27&26.33&54.97\\
\mcb\activeres/\periodic/\optmin/\mvt[600]&6.57&26.52&22.55&50.29\\
/\periodic/\optavg&5.25&22.34&4.89&16.87\\
/\periodic/\optavg/\mft[300]&5.24&22.85&4.92&17.02\\
/\periodic/\optavg/\mft[600]&5.24&22.90&4.92&16.88\\
/\periodic/\optavg/\mvt[300]&5.23&22.16&4.78&16.52\\
/\periodic/\optavg/\mvt[600]&5.20&22.95&4.63&16.52\\
/\periodic/\optmin&5.65&23.23&4.90&16.58\\
/\periodic/\optmin/\mft[300]&5.65&23.23&4.90&16.58\\
/\periodic/\optmin/\mft[600]&5.65&23.23&4.90&16.58\\
/\periodic/\optmin/\mvt[300]&5.65&23.25&4.77&16.00\\
/\periodic/\optmin/\mvt[600]&5.63&22.78&4.59&15.51\\
/\mcbsp/\optavg&3.79&16.03&9.58&23.98\\
/\mcbsp/\optavg/\mft[300]&3.79&15.86&9.60&24.02\\
/\mcbsp/\optavg/\mft[600]&3.78&16.13&9.58&23.79\\
/\mcbsp/\optavg/\mvt[300]&3.77&16.24&9.36&23.27\\
/\mcbsp/\optavg/\mvt[600]&3.77&16.00&9.11&22.63\\
/\mcbsp/\optmax&3.79&16.71&10.41&26.96\\
/\mcbsp/\optmax/\mft[300]&3.79&16.71&10.41&26.96\\
/\mcbsp/\optmax/\mft[600]&3.79&16.71&10.41&26.96\\
/\mcbsp/\optmax/\mvt[300]&3.78&17.13&10.14&26.75\\
/\mcbsp/\optmax/\mvt[600]&3.76&17.52&9.87&25.78\\